\documentclass{IEEEtran}
\twocolumn
\usepackage[english]{babel}
\usepackage{amsmath,amsthm}
\usepackage{amssymb,mathrsfs}
\usepackage{amsfonts}
\usepackage[final,hyperfootnotes=false,hidelinks]{hyperref} 
\usepackage[final]{graphicx}
\usepackage{epstopdf}
\usepackage{tikz}
\usepackage{pgfplots}
\usepackage{subfig}
\usepackage{flushend}
\usepackage{cprotect}

\usepackage{xcolor}
\hypersetup{
    colorlinks,
    linkcolor={blue!80!black},
    citecolor={blue!80!black},
    urlcolor={blue!80!black}
}

\usepackage{tikz}
\usetikzlibrary{arrows,positioning,fit,backgrounds,shapes}
\usetikzlibrary{calc}

\newtheorem{thm}{Theorem}
\newtheorem{cor}[thm]{Corollary}
\newtheorem{lem}[thm]{Lemma}
\newtheorem{prop}[thm]{Proposition}
\newtheorem{ex}[thm]{Example}
\theoremstyle{definition}
\newtheorem{defn}[thm]{Definition}
\theoremstyle{remark}
\newtheorem{rem}[thm]{Remark}

\newcommand{\mb}{\mathbf}
\DeclareMathOperator*{\supp}{supp}
\DeclareMathOperator*{\esssup}{ess\,sup}
\DeclareMathOperator*{\Var}{Var}
\begin{document}

\title{$E_{\gamma}$-Resolvability}%
\author{\IEEEauthorblockN{Jingbo Liu~~~~~~~~~Paul Cuff~~~~~~~~Sergio Verd\'{u}}\\
\IEEEauthorblockA{Dept. of Electrical Eng., Princeton University, NJ 08544\\
\{jingbo,cuff,verdu\}@princeton.edu}}%

\maketitle

\begin{abstract}
The conventional channel resolvability refers to the minimum rate needed for an input process to approximate the channel output distribution in total variation distance. In this paper we study $E_{\gamma}$-resolvability,
in which total variation is replaced by the more general $E_{\gamma}$ distance. A general one-shot achievability bound for the precision of such an approximation is developed.
Let $Q_{\sf X|U}$ be a random transformation, $n$ be an integer, and $E\in(0,+\infty)$.
We show that in the asymptotic setting where $\gamma=\exp(nE)$, a (nonnegative) randomness rate above $\inf_{Q_{\sf U}: D(Q_{\sf X}\|{{\pi}}_{\sf X})\le E} \{D(Q_{\sf X}\|{{\pi}}_{\sf X})+I(Q_{\sf U},Q_{\sf X|U})-E\}$ is sufficient to approximate the output distribution ${{\pi}}_{\sf X}^{\otimes n}$ using the channel $Q_{\sf X|U}^{\otimes n}$, where $Q_{\sf U}\to Q_{\sf X|U}\to Q_{\sf X}$, and is also necessary in the case of finite $\mathcal{U}$ and $\mathcal{X}$.
In particular, a randomness rate of $\inf_{Q_{\sf U}}I(Q_{\sf U},Q_{\sf X|U})-E$ is always sufficient.
We also study the convergence of the approximation error under the high probability criteria in the case of random codebooks.
Moreover, by developing simple bounds relating $E_{\gamma}$ and other distance measures, we are able to determine the exact linear growth rate of the approximation errors measured in relative entropy and smooth R\'{e}nyi divergences for a fixed-input randomness rate.
The new resolvability result is then used to derive 1) a one-shot upper bound on the probability of excess distortion in lossy compression, which is exponentially tight in the i.i.d.~setting, 2) a one-shot version of the mutual covering lemma, and 3) a lower bound on the size of the eavesdropper list to include the actual message and a lower bound on the eavesdropper false-alarm probability in the wiretap channel problem, which is (asymptotically) ensemble-tight.
\end{abstract}

\newcommand\blfootnote[1]{%
  \begingroup
  \renewcommand\thefootnote{}\footnote{#1}%
  \addtocounter{footnote}{-1}%
  \endgroup
}
\blfootnote{This paper was presented in part at 2015 IEEE International Symposium on Information Theory (ISIT).}

\begin{IEEEkeywords}
Resolvability,
source coding,
broadcast channel,
mutual covering lemma,
wiretap channel
\end{IEEEkeywords}

\section{Introduction}\label{secBackground}
\emph{Channel resolvability}, introduced by Han and Verd\'u \cite{han1993approximation}, is defined as the
minimum randomness rate required to synthesize an input
so that its corresponding output distribution approximates a target output distribution.
While the resolvability problem itself differs from classical topics in information theory such as data compression and transmission,
\cite{han1993approximation}
unveils its potential utility in operational problems through the solution of the strong converse problem of identification coding \cite{ahlswede1989identification}.
Other applications of distribution approximation
in information theory include
common randomness of two random variables \cite{wyner1975common},
strong converse in identification through channels \cite{han1993approximation},
random process simulation \cite{steinberg1996simulation},
secrecy \cite{csiszar1996}\cite{hayashi2006general}\cite{bloch2013strong}\cite{han2013reliability}, channel synthesis \cite{steinberg1994}\cite{cuff2012distributed},
lossless and lossy source coding \cite{steinberg1996simulation}\cite{han1993approximation}\cite{song}, and the empirical distribution of a capacity-achieving code \cite{shamai1997empirical}\cite{polyanskiy2014empirical}.
The achievability part
of resolvability (also known as the soft-covering lemma in \cite{cuff2012distributed})
is particularly useful, and coding theorems via resolvability
have certain advantages over what is obtained from traditional
typicality-based approaches (see e.g. \cite{bloch2013strong}).

If the channel is stationary memoryless and the target output distribution is induced by a stationary memoryless input,
then the resolvability is the minimum mutual information over all input distributions inducing the (per-letter) target output distribution,
no matter when the approximation error is measured in total variation distance \cite{han1993approximation}\cite[Theorem~1]{watanabe2014strong}, normalized relative entropy \cite[Theorem~6.3]{wyner1975common}\cite{han1993approximation}, or unnormalized relative entropy \cite{han2013reliability}.
In contrast, relatively few measures for the quality of the approximation of output statistics have been proposed for which the resolvability can be strictly smaller than mutual information.
As shown by Steinberg and Verd\'{u} \cite{steinberg1996simulation}, one exception is the Wasserstein distance measure, in which case the finite precision resolvability for the identity channel can be related to the rate-distortion function of the source where the distortion is the metric in the definition of the Wasserstein distance \cite{steinberg1996simulation}.

In this paper we generalize the theory of resolvability by considering a distance measure, $E_{\gamma}$, defined in Section~\ref{sec_pre}, of which the total variation distance is a special case where $\gamma=1$. The $E_{\gamma}$ metric\footnote{``Metric'' or ``distance'' are used informally since, other than nonnegativity, $E_{\gamma}$, in general, does not satisfy any of the other three requirements for a metric.} is more forgiving than total variation distance when $\gamma>1$, and the larger $\gamma$ is, the less randomness is needed at the input for approximation in $E_{\gamma}$.
Various achievability and converse bounds for resolvability in the $E_{\gamma}$ metric are derived in Section~\ref{sec_source}-\ref{sec_channel}.

If we fix an input randomness rate and consider the minimum exponential growth rate of $\gamma$ such that the approximation error measured in $E_{\gamma}$ is small, we are effectively dealing with resolvability in a large deviations regime. Using general bounds on $E_{\gamma}$ and related metrics (not specific to the resolvability problem) developed in Section~\ref{sec_renyi},
we conclude that, in fact, the growth rate of the exponent of $\gamma$ is the same as (see Section~\ref{sec_univ}):
\begin{enumerate}
  \item The minimum exponential growth of a threshold such that the cdf of the relative information between the true distribution and a target distribution is close to one. (That is, the \emph{excess relative information} defined in Section~\ref{sec_pre} is small.)
  \item The linear growth rate of the minimum relative entropy between the true output distribution and the target distribution.
  \item The linear growth rate of the minimum \emph{smooth R\'{e}nyi $\alpha$-divergence} (of any order except for $\alpha=1$) between the true output distribution and the target distribution.
\end{enumerate}
In the case of a discrete memoryless channel with a \emph{given} stationary memoryless target output, we provide a single-letter characterization of the minimum exponential growth rate of $\gamma$ to achieve approximation in $E_{\gamma}$ (Theorem~\ref{thm_conv}).
The corresponding problem for the \emph{worst case} target distribution (which is generally not stationary memoryless even for stationary memoryless channels) has a different flavor; a converse bound (Theorem~\ref{thm28}) can be derived by drawing connections to the identification coding problem (as did in \cite{han1993approximation}), but which generally does not match the achievability bound unless $\gamma=1$.

In addition to achievability results in terms of the expectation of the approximation error over a random codebook, we also prove achievability under the high probability criteria\footnote{\label{ft_high}More precisely, by ``the high probability criteria'' we that mean that the approximation error satisfies a Gaussian concentration result (which ensures a doubly exponential decay in blocklength when the single-shot result is applied to the multi-letter setting).} for a random codebook (Section~\ref{sec_tail}).
The implications of the latter problem in secrecy has been noted by several authors \cite{cuff2016}\cite{Goldfeld2016}\cite{tahmasbi2016}.
Here we adopt a simple non-asymptotic approach based on concentration inequalities, dispensing with the finiteness or stationarity assumptions on the alphabet required by the previous proof method \cite{cuff2016} based on Chernoff bounds.

The $E_{\gamma}$ metric provides a very convenient tool for change-of-measure purposes: if $E_{\gamma}(P\|Q)$ is small, and the probability of some event is large under $P$, then the probability of this event under $Q$ is essentially lower-bounded by $\frac{1}{\gamma}$ (see \eqref{e_prop3}). In the special case of $\gamma=1$ (total variation distance), this change-of-measure trick has been widely used, see \cite{cuff2012distributed}\cite{song},
but the general $\gamma\ge1$ case is more interesting, since $P$ and $Q$ need not be essentially the same, thereby opening up novel applications of the resolvability theorem (e.g.~the one-shot mutual covering lemma in Section~\ref{sec_mutual}). In this paper we present three information theoretic applications of $E_{\gamma}$-resolvability (Sections~\ref{seclikelihood}-\ref{sec_wiretap}):
\begin{itemize}
  \item One-shot lower bound on the probability that the distortion lies below a certain threshold in lossy compression (successful decompression). Compared with the proof based on the soft-covering lemma (achievability part of resolvability in total variation distance) \cite{song}, the new bound is capable of recovering the exact success exponent, previously obtained using the method of types (see \cite{csiszar1981information}) for discrete memoryless settings.
      In contrast, our derivation applies to general sources and dispenses with memoryless and finite alphabet assumptions.
\item A one-shot generalization of the mutual covering lemma, with a proof significantly different from the original one based on second moments  \cite{el1981proof}\cite{el2011network}. In \cite{jingbo2015marton} we applied the one-shot mutual covering lemma to derive a one-shot version of Marton's inner bound for the broadcast channel with a common message, without using time-sharing/common randomness.
 \item One-shot achievability for wiretap channels, where a novel secrecy measure in terms of the eavesdropper ability to perform list decoding and detect the absence of message is proposed. The previous proofs for wiretap channels using the conventional resolvability (soft-covering lemma) \cite{hayashi2006general}\cite{bloch2013strong} are only suitable when the rate is low enough to achieve perfect secrecy.
     In contrast, $E_{\gamma}$-resolvability yields lower bounds on the minimum size of the eavesdropper list for an arbitrary rate of communication. This interpretation of security in terms of list size is reminiscent of equivocation \cite{wyner1975wire}, and indeed we recover the same formula in the asymptotic setting,
     even though there is no direct correspondence between both criteria.
     Moreover, we also consider a more general case where the eavesdropper wishes to reliably detect whether a message is sent,
     while being able to produce a list including the actual message if it decides it is present.
     This is a practical setup because ``no message'' may be valuable information which the eavesdropper wants to ascertain reliably.
     The idea is reminiscent of the stealth communication problem (see \cite{hou2014effective}\cite{bloch2015} and the references therein) which also involves a hypothesis test on whether a message is sent. However, the setup and the analysis (including the covering lemma) are quite different from \cite{hou2014effective} and \cite{bloch2015}. In comparison, our results are more suitable for the regime with higher communication rates and lower secrecy demands. In the discrete memoryless case, we obtain single-letter expressions of the tradeoff between the transmission rate, eavesdropper list, and the exponent of the false-alarm probability for the eavesdropper (i.e. declaring the presence of a message when there is none).
\end{itemize}

\section{Preliminaries}\label{sec_pre}
In this paper, several distance measures between two probability distributions play an important role.
In this section, we introduce these distance measures and discuss some of their relations.
At the end of this section, the channel resolvability problem is formulated.

\subsection{Excess Relative Information Metric}
Given two nonnegative $\sigma$-finite measures $\nu\ll \mu$ on $\mathcal{X}$, define the \emph{relative information} (for each $x\in\mathcal{X}$) as the logarithm of the Radon-Nikodym derivative:
\begin{align}
\imath_{\nu\|\mu}(x):=\log\frac{{\rm d}\nu}{{\rm d}\mu}(x).
\end{align}
For $\gamma>0$ and a probability measure $P$, we define the \emph{excess relative information metric with threshold $\gamma$},
\begin{align}\label{excess}
\bar{F}_{\gamma}(P\|\mu):=\mathbb{P}[\imath_{P\|\mu}(X)>\log\gamma]
\end{align}
where $X\sim P$. As such, it can be expressed in terms of the \emph{relative information spectrum} (see \cite{verdubook}) as
\begin{align}
\bar{F}_{\gamma}(P\|\mu)=1-F_{P\|\mu}(\log \gamma).
\end{align}

Note that \eqref{excess} is nonnegative and vanishes when $P=\mu$ provided that $\gamma>1$, and can therefore be considered as a measure of the discrepancy between $P$ and $\mu$.
In addition to playing an important role in one-shot analysis (see \cite{verdu2012non}), \eqref{excess} provides richer information than the relative entropy measure since
\begin{align}
D(P\|\mu)
&:=\mathbb{E}[\imath_{P\|\mu}(X)]
\\
&=\int_{[0,+\infty)}\mathbb{P}[\imath_{P\|\mu}(X)>\tau]{\rm d}\tau
\nonumber
\\
&\quad-\int_{(-\infty,0]}(1-\mathbb{P}[\imath_{P\|\mu}(X)>\tau]){\rm d}\tau.
\end{align}
Moreover, the excess relative information is also related to total variation distance since for probability measures $P$ and $Q$,
\begin{align}
\frac{1}{2}|P-Q|=\mathbb{P}[\imath_{P\|Q}(X)>0]-\mathbb{P}
[\imath_{P\|Q}(Y)>0],
\end{align}
where $X\sim P$ and $Y\sim Q$ \cite{verdu2014total}. However, perhaps surprisingly, the excess relative information metric \emph{does not} satisfy the data processing inequality, in contrast to the relative entropy and total variation distance:
\begin{prop}\label{prop_egamma}
Suppose $P_X\to P_{Y|X} \to P_{Y}$, $Q_X\to P_{Y|X} \to Q_Y$ and $P_X\ll\gg Q_X$ (that is, $P_X\ll Q_X$ and $Q_X\ll P_X$). Then there exists $\gamma>0$ such that
\begin{align}
\bar{F}_{\gamma}(P_X\|Q_X)< \bar{F}_{\gamma}(P_Y\|Q_Y)
\end{align}
where $(X,Y)\sim P_{XY}$, unless $P_{X|Y}=Q_{X|Y}$ almost surely.
\end{prop}
Note that in the absence of the condition $P_X \ll\gg Q_X$ it is indeed possible to find examples
where the cdf of the relative information at the input is dominated by that at the output.
\begin{proof}
Suppose, for the sake of contradiction, 
\begin{align}\label{e1}
\mathbb{P}\left[\frac{{\rm d}Q_X}{{\rm d}P_X}(X)\ge \lambda\right]
\le \mathbb{P}\left[\frac{{\rm d}Q_Y}{{\rm d}P_Y}(Y)\ge \lambda\right]
\end{align}
for all $\lambda>0$, where $\frac{{\rm d}Q_Y}{{\rm d}P_Y}$ is well defined because $P_Y\ll\gg Q_Y$ follows from $P_X\ll\gg Q_X$. However, since
\begin{align}
1&=\mathbb{E}\left[\frac{{\rm d}Q_X}{{\rm d}P_X}(X)\right]
\\
&=\int_0^{\infty}\mathbb{P}\left[\frac{{\rm d}Q_X}{{\rm d}P_X}(X)\ge \lambda\right]{\rm d}\lambda
\\
&=\mathbb{E}\left[\frac{{\rm d}Q_Y}{{\rm d}P_Y}(Y)\right]
\\
&=\int_0^{\infty}\mathbb{P}\left[\frac{{\rm d}Q_Y}{{\rm d}P_Y}(Y)\ge \lambda\right]{\rm d}\lambda,
\end{align}
\eqref{e1} must hold with equality all $\lambda>0$ except for a set of measure zero. But both sides of \eqref{e1} are decreasing, left continuous functions of $\lambda$, so in fact the equality holds for all $\lambda>0$. This implies that $D(P_X\|Q_X)=D(P_Y\|Q_Y)$, and hence $P_{X|Y}=Q_{X|Y}$ almost surely.
\end{proof}
The failure of the data processing inequality suggests that the excess relative information itself has a limited operational significance; rather, we shall mainly resort to the simplicity of its definition and its connection to other distance measures which have more significant operational meanings. In \cite{verdu2014total} several bounds on total variation distance using the excess relative information metric are derived. Next we refine those results by providing the tightest possible bounds (i.e.~the locus of possible values of $\bar{F}_{\lambda}(P\|Q)$ as a function of the total variation distance $|P-Q|$; see Figure~\ref{fig1}).
\begin{prop}\label{prop1}
If $P\ll Q$ are distributions on $\mathcal{X}$ (not necessarily discrete), and $\delta:=\frac{1}{2}|P-Q|$, then
\begin{align}\label{ub}
\bar{F}_{\lambda}(P\|Q)
\le\left\{\begin{array}{cl}
            \frac{\lambda\delta}{\lambda-1} & \lambda\in[\frac{1}{1-\delta},\infty)\\
            1 & \lambda\in(0,\frac{1}{1-\delta})
          \end{array}
          \right.
\end{align}
and
\begin{align}\label{lb}
\bar{F}_{\lambda}(P\|Q)
\ge\left\{\begin{array}{cl}
            0 & \lambda\in(\frac{1}{1-\delta},\infty) \\
            1-(1-\delta)\lambda & \lambda\in(1,\frac{1}{1-\delta}] \\
            \delta & \lambda\in(1-\delta,1] \\
            1-\frac{\lambda\delta}{1-\lambda} & \lambda\in (0,1-\delta]
          \end{array}
          \right.
\end{align}
where $X\sim P$.
Moreover, the bounds above are tight (that is, the values given on the right hand sides are the supremum/infimum over $P\ll Q$ with a total variation distance $\delta$).
\end{prop}
\begin{proof}
See Appendix~\ref{pf_prop1}.
\end{proof}
As shown in \cite[Lemma~5]{han1993approximation}, a useful and compact bound on $\bar{F}_{\lambda}(P\|Q)$ is
\begin{align}
\frac{1}{2}|P-Q|\le \lambda+\bar{F}_{\lambda}(P\|Q).
\label{e_hanbound}
\end{align}
The bound in \eqref{e_hanbound} is not tight, in contrast to \eqref{lb}. In particular, \eqref{e_hanbound} does not show the fact the $\bar{F}_{\lambda}(P\|Q)$ tends to 1 as $\lambda\downarrow 0$.
\begin{figure}[h!]
  \centering
  \includegraphics[width=3.8in]{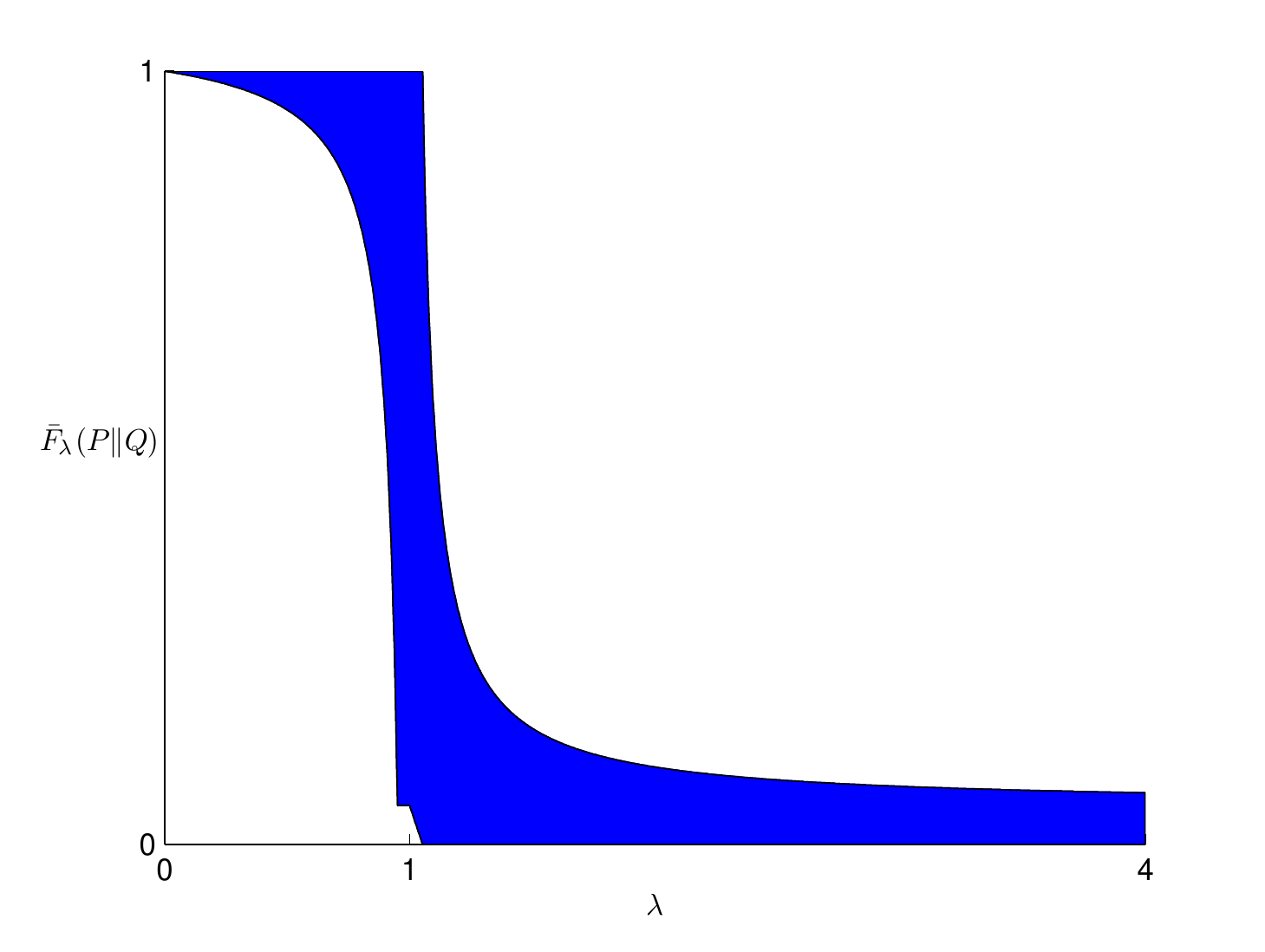}\\
  \caption{Locus of possible values of $\bar{F}_{\lambda}(P\|Q)$ for $|P-Q|=0.1$.}\label{fig1}
\end{figure}

\subsection{$E_{\gamma}$ Metric}\label{sec_related}
\begin{defn}
Given probability distributions $P\ll Q$ on the same alphabet and $\gamma\ge1$, define
\begin{align}
E_{\gamma}(P\|Q)
:=\mathbb{P}[\imath_{P\|Q}(X)>\log\gamma]
-\gamma\,\mathbb{P}[\imath_{P\|Q}(Y)>\log\gamma]\label{e_prop1}
\end{align}
where $X\sim P$ and $Y\sim Q$.
\end{defn}
We can see that $E_{\gamma}$ is an $f$-divergence \cite{csiszar1967} with
\begin{align}
f(x)=(x-\gamma)^+.\label{e_prop2}
\end{align}
From the Neyman-Pearson lemma we have an alternative formula for $E_{\gamma}$:
\begin{align}
E_{\gamma}(P\|Q)=\max_{\mathcal{A}}\left\{P(\mathcal{A})-\gamma Q(\mathcal{A})\right\}.
\label{e_prop3}
\end{align}
$E_{\gamma}$ is a basic quantity in binary hypothesis testing: in the Bayesian case, where $P$ and $Q$ have a priori probabilities
$\pi_P$, and $\pi_Q$, respectively, the probability of making the correct decision is given by
\begin{align}
&\quad\pi_P\mathbb{P}\left[\imath_{P\|Q}(X)>\log\frac{\pi_Q}{\pi_P}\right]
+\pi_Q\mathbb{P}\left[\imath_{P\|Q}(Y)\le\log\frac{\pi_Q}{\pi_P}\right]
\nonumber\\
&=\pi_Q + \pi_P \, E_{\frac{\pi_Q}{\pi_P}} (P \| Q).
\end{align}
$E_{\gamma}$ has been considered in various fields under different names; for example, in cryptography (more specifically, differential privacy \cite{dwork2008differential} \cite{barthe2013beyond}) a computation over a database is said to be $(\epsilon,\delta)$-differentially private if the $E_{\exp(\epsilon)}$ distance between the output distributions for any two databases which \emph{differ in at most one element} \cite{dwork2008differential} is upper-bounded by $\delta$. The function \eqref{e_prop2} is called a hockey-stick function in financial engineering \cite{hull2000other}, and so $E_{\gamma}$ is sometimes called a hockey-stick divergence \cite{sharma2012strong} \cite{sharma2013fundamental}.

It appears that, $E_{\gamma}$ was introduced to information theory by \cite{polyanskiy2010channel} to simplify the expression of the DT bound therein.
Also, \cite{polyanskiy2010arimoto} derived a general channel coding converse using any $g$-divergence (a divergence satisfying the data processing inequality, including all $f$-divergences), which recovers the Wolfowitz converse when specialized to $E_{\gamma}$, and admits generalization to the quantum setting \cite{sharma2012strong} \cite{sharma2013fundamental}. The use of $E_{\gamma}$ for change-of-measure appeared in \cite{polyanskiy2014private} \cite{jingbo2015egamma} \cite{jingbo2015marton}.

Below, we prove some basic properties of $E_{\gamma}$ useful for later sections. Additional properties of $E_{\gamma}$ can be found in \cite{sason2015}\cite[Theorem~21]{yurythesis}.
\begin{figure}[h!]
  \centering
  \includegraphics[width=3.8in]{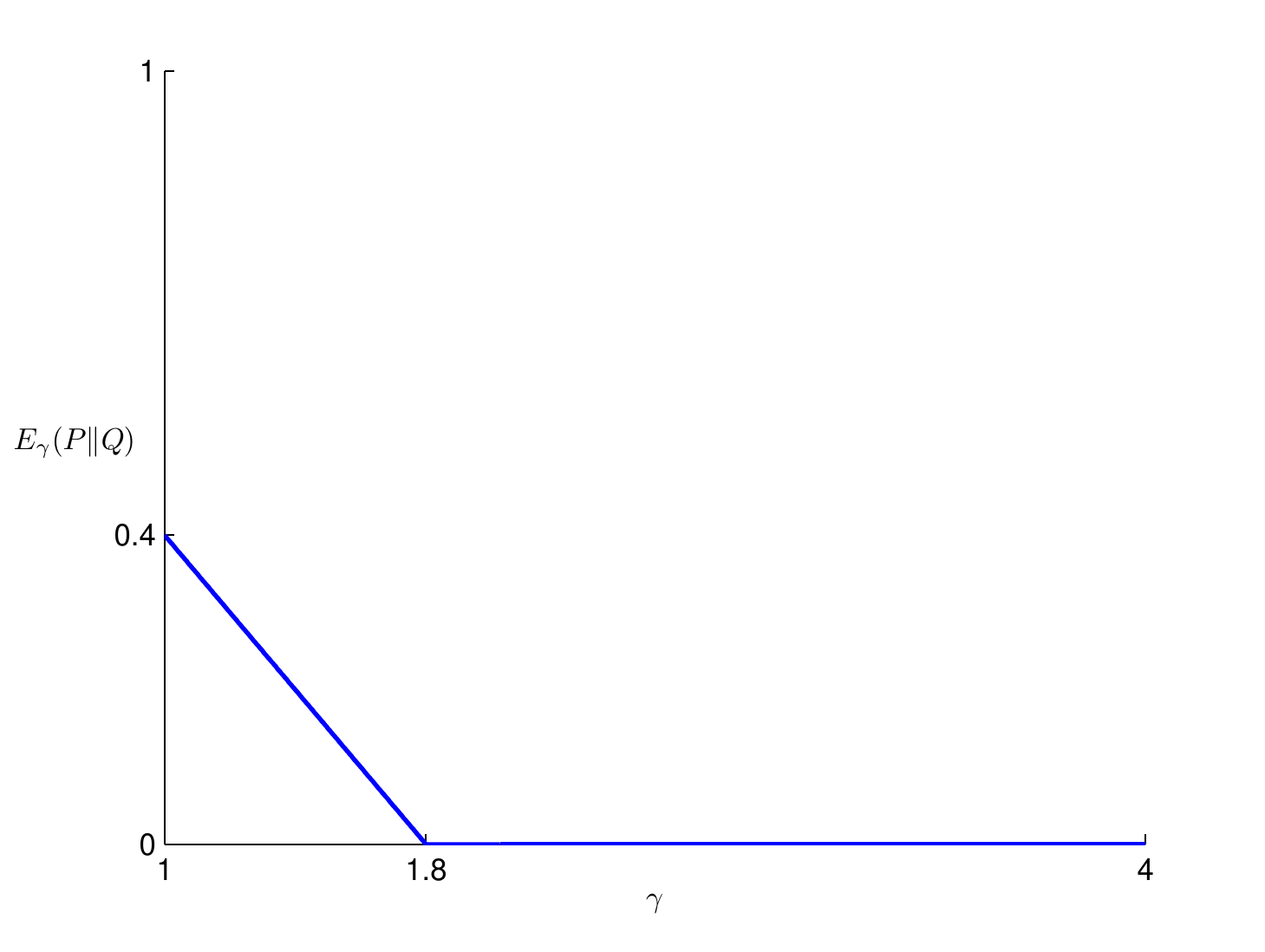}\\
  \caption{$E_{\gamma}(P\|Q)$ as a function of $\gamma$ where $P={\rm Ber}(0.1)$ and $Q={\rm Ber}(0.5)$.}\label{fig2}
\end{figure}
\begin{prop}\label{prop3}
Assume that $P\ll S\ll Q$ are probability distributions on the same alphabet, and $\gamma,\gamma_1,\gamma_2\ge1$.
\begin{enumerate}
\item $[1,\infty)\to [0,\infty)\colon \gamma\mapsto E_{\gamma}(P\|Q)$ is convex, non-increasing, and continuous.
\item\label{prop3_1} For any event $\mathcal{A}$,
\begin{align}
Q(\mathcal{A})\ge \frac{1}{\gamma}(P(\mathcal{A})-E_{\gamma}(P\|Q)).
\end{align}
\item\label{prop3_2} Triangle inequalities:
\begin{align}
E_{\gamma_1\gamma_2}(P\|Q)&\le E_{\gamma_1}(P\|S)+\gamma_1E_{\gamma_2}(S\|Q),
\label{e19}
\\
E_{\gamma}(P\|Q)+E_{\gamma}(P\|S)
&\ge\frac{\gamma}{2}|S-Q|+1-\gamma.\label{e20}
\end{align}
\item \label{prop3_3}Monotonicity: if $P_{XY}=P_XP_{Y|X}$ and $Q_{XY}=Q_XQ_{Y|X}$ are joint distributions on $\mathcal{X}\times\mathcal{Y}$, then
\begin{align}
E_{\gamma}(P_X\|Q_X)\le E_{\gamma}(P_{XY}\|Q_{XY})
\end{align}
where equality holds for all $\gamma \geq 1$ if and only if $P_{Y|X}=Q_{Y|X}$.
\item \label{pt3_4}Given $P_X$, $P_{Y|X}$ and $Q_{Y|X}$, define
\begin{align}
E_{\gamma}(P_{Y|X}\|Q_{Y|X}|P_X):=\mathbb{E}[E_{\gamma}(P_{Y|X}(\cdot|X)\|Q_{Y|X}(\cdot|X))]
\end{align}
where the expectation is w.r.t. $X\sim P_X$. Then
we have
\begin{align}
E_{\gamma}(P_XP_{Y|X}\|P_XQ_{Y|X})= E_{\gamma}(P_{Y|X}\|Q_{Y|X}|P_X).
\end{align}
\item
\begin{align}
1-\gamma\left(1-\frac{1}{2}|P-Q|\right)\le E_{\gamma}(P\|Q)\le \frac{1}{2}|P-Q|.\label{e74}
\end{align}
\end{enumerate}
\end{prop}
\begin{proof}
The proofs of 1), 2), 4), 5) and the second inequality in \eqref{e74}
are omitted since they follow either directly from \eqref{e_prop3} or are similar to the corresponding properties for total variation distance.

For 3), observe that
\begin{align}
E_{\gamma_1\gamma_2}(P\|Q)
&=\max_{\mathcal{A}}(P(\mathcal{A})-\gamma_1S(\mathcal{A})
+\gamma_1S(\mathcal{A})-\gamma_1\gamma_2 Q(\mathcal{A}))
\\
&\le\max_{\mathcal{A}}(P(\mathcal{A})-\gamma_1S(\mathcal{A}))
+\max_{\mathcal{A}}(\gamma_1S(\mathcal{A})-\gamma_1\gamma_2 Q(\mathcal{A}))
\\
&=E_{\gamma_1}(P\|S)+\gamma_1
E_{\gamma_2}(S\|Q),
\end{align}
and that
\begin{align}
E_{\gamma}(P\|Q)+E_{\gamma}(P\|S)
&=\max_{\mathcal{A}}(P(\mathcal{A})-\gamma Q(\mathcal{A}))
\nonumber\\
&\quad+\max_{\mathcal{A}}(1-P(\mathcal{A})-\gamma +\gamma S(\mathcal{A}))
\\
&\ge
\max_{\mathcal{A}}(P(\mathcal{A})-\gamma Q(\mathcal{A})
\nonumber\\
&\quad+1-P(\mathcal{A})-\gamma +\gamma S(\mathcal{A}))
\\
&=\gamma\max_{\mathcal{A}}(S(\mathcal{A})-Q(\mathcal{A}))
+1-\gamma
\\
&=\frac{\gamma}{2}|S-Q|
+1-\gamma.
\end{align}
As far as 6) note that the left inequality in \eqref{e74} follows
by setting $S=P$ in \eqref{e20}.
\end{proof}
In view of the right inequality in \eqref{e74},
in the resolvability problem, when the rate of the codebook is not large enough to soft-cover the output distribution in total variation, it may be still possible to do so in the $E_{\gamma}(P\|Q)$ metric.

\subsection{Smooth R\'{e}nyi divergence and Relationships between the Distance Measures}\label{sec_renyi}
To discuss the smooth R\'{e}nyi divergence, it is convenient to generalize some of our definitions to allow nonnegative finite measures that are not necessarily probability measures:
\begin{defn}
For $\gamma\ge1$, nonnegative finite measures $\mu$ and $\nu$ on $\mathcal{X}$, $\mu\ll\nu$,
\begin{align}
|\mu-\nu|:=\int|{\rm d}\mu-{\rm d}\nu|,
\end{align}
and\footnote{
Following established usage in measure theory,  we use $\lambda(\imath_{\mu\|\nu}>\log\gamma)$ as an abbreviation of $\lambda(\{x\colon\imath_{\mu\|\nu}(x)>\log\gamma\})$ for an arbitrary signed measure $\lambda$.}
  \begin{align}
  E_{\gamma}(\mu\|\nu)&
  :=\sup_{\mathcal{A}}\{\mu(\mathcal{A})-\gamma \nu(\mathcal{A})\}
  \label{e_egammadefn}
  \\
  &=(\mu-\gamma\nu)(\imath_{\mu\|\nu}>\log\gamma)
  \\
  &=\frac{1}{2}(\mu-\gamma\nu)(\mathcal{X})
  -\frac{1}{2}(\mu-\gamma\nu)(\imath_{\mu\|\nu}\le\log\gamma)
  \nonumber\\
  &\quad+\frac{1}{2}(\mu-\gamma\nu)(\imath_{\mu\|\nu}>\log\gamma)
  \\
  &=\frac{1}{2}\mu(\mathcal{X})-\frac{\gamma}{2} \nu(\mathcal{X})+\frac{1}{2}|\mu-\gamma \nu|.
  \label{e_37}
  \end{align}
\end{defn}
Note that $E_1(P\|\mu)\neq\frac{1}{2}|P-\mu|$ when $\mu$ is not a probability measure.

The following result is a generalization of the $\gamma_1=1$ case of \eqref{e19} to unnormalized measures, and the proof is immediate from the definition of \eqref{e_egammadefn} and the subadditivity of the sup operator.
\begin{prop}\label{prop_tri}
Triangle inequality: if $\mu$, $\nu$ and $\theta$ are nonnegative finite measures on the same alphabet, $\mu\ll\theta\ll\nu$, then
\begin{align}
E_{\gamma}(\mu\|\nu)\le E_1(\mu\|\theta)+E_{\gamma}(\theta\|\nu).
\label{e38}
\end{align}
\end{prop}
The R\'{e}nyi divergence, defined as follows, is not an $f$-divergence, but is a monotonic function of the Hellinger distance \cite{polyanskiy2010arimoto}.
\begin{defn}[R\'{e}nyi $\alpha$-divergence]
Let $\mu$ be a nonnegative finite measure and $Q$ a probability measure on $\mathcal{X}$, $\mu\ll Q$, $X\sim Q$. For $\alpha\in(0,1)\cup(1,+\infty)$,
  \begin{align}
  D_{\alpha}(\mu\|Q)
  :=\frac{1}{\alpha-1}\log\mathbb{E}\left[\left(\frac{{\rm d}\mu}{{\rm d}Q}(X)\right)^{\alpha}\right],
  \end{align}
  and
  \begin{align}
  D_0(\mu\|Q)&:=\log\frac{1}{Q(\imath_{\mu\|Q}>-\infty)},
  \\
  D_{\infty}(\mu\|Q)&:=\mu\mbox{-}\esssup \imath_{\mu\|Q}
  \end{align}
  which agree with the the limits as $\alpha\downarrow 0$ and $\alpha\uparrow\infty$.
\end{defn}
$D_{\alpha}(P\|Q)$ is non-negative and monotonically increasing in $\alpha$.
More properties about the R\'enyi divergence can be found, e.g.~in \cite{renyidef}\cite{erven2014}\cite{verdubook}.
\begin{defn}[smooth R\'{e}nyi $\alpha$-divergence]\label{defn8}
For $\alpha\in(0,1)$, $\epsilon\in(0,1)$,
  \begin{align}
  D_{\alpha}^{+\epsilon}(P\|Q):=\sup_{\mu\in B^{\epsilon}(P)}D_{\alpha}(\mu\|Q);\label{e_prop15}
  \end{align}
  for $\alpha\in (1,\infty]$, $\epsilon\in(0,1)$,
  \begin{align}
  D_{\alpha}^{-\epsilon}(P\|Q):=\inf_{\mu\in B^{\epsilon}(P)}D_{\alpha}(\mu\|Q),\label{e_prop16}
  \end{align}
  where $B^{\epsilon}(P):=\{\mu\textrm{ nonnegative}:~E_1(P\|\mu)\le\epsilon\}$ is the $\epsilon$-neighborhood of $P$ in $E_1$.
\end{defn}
\begin{rem}
Our smooth $\infty$-divergence agrees with the smooth max R\'{e}nyi divergence in \cite{warsi2013one}, although the definitions look different. However, our smooth
$0$-divergence is different from the smooth min R\'{e}nyi divergences in \cite{wang2009simple} and \cite[Definition~1]{warsi2013one} except for non-atomic measures.\footnote{With the definition of smooth min R\'{e}nyi divergence in \cite{wang2009simple} and \cite[Definition~1]{warsi2013one},
Proposition~\ref{prop_1}-\ref{ptD0}) would hold with the $\beta_{\alpha}$ as defined in \eqref{e_beta2} rather than \eqref{eq46}.}
\end{rem}
\begin{rem}
The smooth R\'{e}nyi divergence is a natural extension of the smooth R\'{e}nyi entropy $H^{\epsilon}_{\alpha}$ defined in \cite{renner2005simple} (which can be viewed as a special case where the reference measure is the counting measure). In \cite{warsi2013one} the smooth min and max R\'{e}nyi divergences are introduced, which correspond to the $\alpha=0$ and $\alpha=+\infty$ cases of Definition~\ref{defn8}.
Moreover, we have introduced $+/-$ in the notation to emphasize the difference between the two possible ways of smoothing in \eqref{e_prop15} and \eqref{e_prop16}.
\end{rem}
The following quantity, which characterizes the binary hypothesis testing error, is a $g$-divergence but not a monotonic function of any $f$-divergence \cite{polyanskiy2010arimoto}. We will see in Proposition~\ref{prop_1} that it is a monotonic function of the $0$-smooth R\'{e}nyi divergence.
\begin{defn}
For nonnegative finite measures $\mu$ and $\nu$ on $\mathcal{X}$, define
\begin{align}
\beta_{\alpha}(\mu,\nu):=
\min_{\mathcal{A}:\mu(\mathcal{A})\ge\alpha}
\nu(\mathcal{A}).
\label{eq46}
\end{align}
\end{defn}
\begin{rem}\label{rem12}
In the literature, the definition of $\beta_{\alpha}(P,Q)$
is usually restricted to probability measures and allows randomized tests:
\begin{align}
\beta_{\alpha}(P_W,Q_W):=\min
\int P_{Z|W}(1|w){\rm d}Q_W(w)
\label{e_beta2}
\end{align}
where the minimization is over all random transformations $P_{Z|W}\colon\mathcal{Z}\to\{0,1\}$ such that
\begin{align}
\int P_{Z|W}(1|w){\rm d}P_W(w)\ge\alpha.
\end{align}
In contrast to $E_{\gamma}$, allowing randomization in the definition can change the value of $\beta_{\alpha}$ except for non-atomic measures. Nevertheless, many important properties of $\beta_{\alpha}$ are not affected by this difference.
\end{rem}

We conclude this section with some inequalities relating those distance measures we have discussed, which will be used to establish the asymptotic equivalence in the large deviation regime in Section~\ref{sec_univ}.
\begin{prop}\label{prop_1}
Suppose $\mu$ is a finite nonnegative measure and $P$ and $Q$ are probability measures, all on $\mathcal{X}$, $X\sim P$, $\epsilon\in(0,1)$, and $\gamma\ge1$.
\begin{enumerate}
   \item \label{pt13_1}For $a>1$,
   \begin{align}
   E_{\gamma}(P\|Q)\le
   \bar{F}_{\gamma}(P\|Q)\le \frac{a}{a-1}E_{\frac{\gamma}{a}}(P\|Q).
   \label{e47}
   \end{align}
   \item \label{pt_mono}$D^{+\epsilon}_{\alpha}(P\|Q)$ is increasing in $\alpha\in[0,1]$ and $D^{-\epsilon}_{\alpha}(P\|Q)$ is increasing in $\alpha\in[1,\infty]$.
   \item If $P$ is a probability measure, then
   \begin{align}
   D(P\|Q)
      &\le\int_{0}^{\infty}
      \mathbb{P}[\imath_{P\|Q}(X)>\tau]\,
      {\rm d}\tau;\label{e_prop20}
      \\
   D(P\|Q)
  &\ge  E_{\gamma}(P\|Q)\log\gamma-2e^{-1}\log e.
  \label{e_prop_21}
   \end{align}

  \item If $\alpha\in[0,1)$ then
      \footnote{The special case of \eqref{e_prop21} for $\alpha=\frac{1}{2}$, $\mu(\mathcal{X})=1$ and $\gamma=1$ is equivalent to a well-known bound on a quantity called \emph{fidelity} using total variation distance, see \cite{wilde2013quantum}.}
  \begin{align}
D_{\alpha}(\mu\|Q)\le \log \gamma-
\frac{1}{1-\alpha}
\log\left(\mu(\mathcal{X})-E_{\gamma}(\mu\|Q)\right).
\label{e_prop21}
\end{align}
If $\alpha\in(1,\infty)$ then
  \begin{align}
D_{\alpha}(\mu\|Q)\ge \log \gamma+
\frac{1}{\alpha-1}
\log
E_{\gamma}(\mu\|Q).
\label{e_prop21_1}
\end{align}

  \item Suppose $\alpha\in[0,1)$, $E_{\gamma}(P\|Q)<1-\epsilon$, then
\begin{align}
D^{+\epsilon}_{\alpha}(P\|Q)\le \log \gamma-
\frac{1}{1-\alpha}
\log\left(1-\epsilon-E_{\gamma}(P\|Q)\right).
\label{e_prop23}
\end{align}
Suppose $\alpha\in(1,\infty]$, $E_{\gamma}(P\|Q)>\epsilon$, then
\begin{align}
D^{-\epsilon}_{\alpha}(P\|Q)\ge \log \gamma+
\frac{1}{\alpha-1}
\log
\left(E_{\gamma}(P\|Q)-\epsilon\right).
\label{e_prop23_1}
\end{align}

  \item \label{pt_13_6}$D_{\infty}^{-\epsilon}(P\|Q)\le \log\gamma\Leftrightarrow E_{\gamma}(P\|Q)\le\epsilon.
$ That is, $D^{-\epsilon}_{\infty}(P\|Q)=\log\inf\{\gamma:E_{\gamma}(P\|Q)\le\epsilon\}.
      $

  \item 
  \label{ptD0}
   $D_0^{+\epsilon}(P\|Q)=-\log\beta_{1-\epsilon}(P,Q)$.

  \item Fix $\tau\in\mathbb{R}$ and $\epsilon\ge\mathbb{P}[\imath_{P\|Q}(X)\le\tau]$, where $X\sim P$ and $P$ is non-atomic. Then
      \begin{align}
      D_0^{+\epsilon}(P\|Q)\ge \tau+\log\frac{1}{1-\epsilon}.
      \label{e_53}
      \end{align}
      Moreover, $D_0^{+\epsilon}(P\|Q)\ge \tau$ holds when $P$ is not necessarily non-atomic.
\end{enumerate}
\end{prop}
\begin{proof}
See Appendix~\ref{pf_prop_1}.
\end{proof}
\subsection{The Resolvability Problem}
The setup in Figure~\ref{fig2} is the same as in the original paper on channel resolvability under total variation distance \cite{han1993approximation}.
Given a random transformation $Q_{X|U}$ and a target distribution ${{\pi}}_X$, we wish to minimize the size $M$ of a codebook $c^M=(c_m)_{m=1}^M$ such that when the codewords are equiprobably selected,
the output distribution
\begin{align}
Q_{X[c^M]}:=\frac{1}{M}\sum_{m=1}^M Q_{X|U=c_m}
\label{e_pxc}
\end{align}
approximates ${{\pi}}_X$. The difference from \cite{han1993approximation} is that we use $E_{\gamma}$ (and other metrics) to measure the level of the approximation.
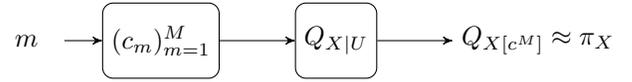
\begin{figure}[h!]
  \centering
\begin{tikzpicture}
[node distance=0.5cm,minimum height=10mm,minimum width=10mm,arw/.style={->,>=stealth'}]
  \node[rectangle,rounded corners] (SD) {$m$};
  \node[rectangle,draw,rounded corners] (PX) [right =of SD] {$(c_m)_{m=1}^M$};
  \node[rectangle,draw,rounded corners] (PZX) [right =1cm of PX] {$Q_{X|U}$};
  \node[rectangle,rounded corners] (PZ) [right =1cm of PZX] {$Q_{X[c^M]}\approx {{\pi}}_X$};

  \draw [arw] (SD) to node[midway,above]{} (PX);
  \draw [arw] (PX) to node[midway,above]{} (PZX);
  \draw [arw] (PZX) to node[midway,above]{} (PZ);
\end{tikzpicture}
\caption{Setup for channel resolvability.}
\label{fig2}
\end{figure}

The fundamental \emph{one-shot} tradeoff is the minimum $M$ required for a prespecified degree of approximation.
One-shot bounds are general in that no structural assumptions on either the random transformation or the target distribution are imposed.
The corresponding asymptotic results can usually be recovered quite simply from the one-shot bounds using, say, the law of large numbers in the memoryless case.
Asymptotic limits are of interest because of the compactness of the expressions,
and because good one-shot bounds are not always known, especially in the converse parts (see for example the converse of resolvability in Section~\ref{sec_tens}).
Unless otherwise stated, the alphabets considered in this paper are not restricted to be finite or countable.
This applies to all the one-shot results in this paper.
Finite alphabets are only assumed in the converses in Sections~\ref{sec_tens} and \ref{sec_worst}, for which we restrict to finite input alphabets or even to discrete memoryless channels (DMC)\footnote{A DMC is a stationary memoryless channel whose input and output alphabets are finite, which is denoted by the corresponding per-letter random transformation (such as $Q_{\sf X|U}$) in this paper.}.

Next, we define the achievable regions in the \emph{general asymptotic setting}
(when sources and channels are arbitrary, see \cite{han1993approximation}\cite{hayashi2006general}) as well as the case of stationary memoryless channels and memoryless outputs. Boldface letters such as $\mb{X}$ denote a general sequence of random variables $\left(X^n\right)_{n=1}^{\infty}$, and sanserif letters such as $\sf X$ denote the generic distributions in iid settings.
\begin{defn}\label{defn14}
Given a channel\footnote{In this setting a ``channel'' refers to a sequence of random transformations.} $(Q_{X^n|U^n})_{n=1}^{\infty}$ and a sequence of target distributions $({{\pi}}_{X^n})_{n=1}^{\infty}$, the triple $(G,R,\mb{X})$ is $\epsilon$-\emph{achievable} ($0<\epsilon<1$) if there exists $(c^{M_n})_{n=1}^{\infty}$,
where $c^{M_n}:=(c_1,\dots,c_{M_n})$ and $c_m\in \mathcal{U}^n$ for each $m=1,\dots,M_n$, and $(\gamma_n)_{n=1}^{\infty}$, so that
\begin{align}
\limsup_{n\to\infty}\frac{1}{n}\log M_n
&\le R;
\\
\limsup_{n\to\infty}\frac{1}{n}\log\gamma_n
&\le G;\label{e65}
\\
\sup_n E_{\gamma_n}(Q_{X^n[c^{M_n}]}\|
{{\pi}}_{X^n})
&\le \epsilon.\label{e66}
\end{align}
Moreover, $(G,R,\mb{X})$ is said to be \emph{achievable} if it is $\epsilon$-achievable for all $0<\epsilon<1$.
Define the asymptotic fundamental limits\footnote{We can write $\min$ instead of $\inf$ in \eqref{e53} and \eqref{e_54} since for fixed $\mb{X}$ the set of $(G,R)$ such that $(G,R,\mb{X})$ is $\epsilon$-achievable is necessarily closed; similarly in \eqref{e_55} and \eqref{e_56}.}
\begin{align}
G_{\epsilon}(R,\mb{X})&:=\min\{g:(g,R,\mb{X})\textrm{ is $\epsilon$-achievable}\};
\\
G(R,\mb{X})&:=\sup_{\epsilon>0}G_{\epsilon}(R,\mb{X}),
\label{e53}
\end{align}
and
\begin{align}
S_{\epsilon}(G,\mb{X})&:=\min\{r:(G,r,\mb{X})\textrm{ is $\epsilon$-achievable}\};
\label{e_54}
\\
S(G,\mb{X})&:=\sup_{\epsilon>0}S_{\epsilon}(G,\mb{X}),
\end{align}
which, in keeping with \cite{han1993approximation}, we refer to as the resolvability function.
In the special case of $Q_{X^n|U^n}=Q_{\sf X|U}^{\otimes n}$ and target ${{\pi}}_{X^n}={{\pi}}_{\sf X}^{\otimes n}$, we may write the quantities in \eqref{e53} and \eqref{e_54} as $G(R,{{\pi}}_{\sf X})$ and $S(G,{{\pi}}_{\sf X})$.
\end{defn}
Note that by \cite{han1993approximation}\cite{watanabe2014strong}, $(0,R,{{\pi}}_{\sf X})$ is achievable if and only if $R\ge I(P_{\sf U},Q_{\sf X|U})$ for some $P_{\sf U}$ satisfying $P_{\sf U}\to Q_{\sf X|U}\to {{\pi}}_{\sf X}$.

In Section~\ref{sec_univ} we show that the same exponent \eqref{e53} is obtained with other distance measures.

In addition to approximation of a given target distribution, \cite{han1993approximation} also considered the minimum rate of randomness needed to approximate a \emph{worst-case} output distribution under total variation distance, which has implications for identification coding.
The $E_{\gamma}$ version of this problem amounts to finding the achievable pairs defined as follows:
\begin{defn}\label{defn_worst}
Given a channel $(Q_{X^n|U^n})_{n=1}^{\infty}$, the pair
$(G,R)$ is $\epsilon$-\emph{achievable} ($0<\epsilon<1$) if for any sequence of \emph{input} distributions $P_{U^n}$,
the triple $(G,R,\mb{X})$ is $\epsilon$-achievable,
where $\mb{X}=(X^n)_{n=1}^{\infty}$ and $P_{U^n}\to P_{X^n|U^n}
\to P_{X^n}$. Moreover, $(G,R)$ is achievable if it is $\epsilon$-\emph{achievable} for all $0<\epsilon<1$.
We define the asymptotic fundamental limits
\begin{align}\label{e_55}
G(R):=\min\{g:(g,R)\textrm{ is achievable}\}
\end{align}
and the \emph{resolvability functions}
\begin{align}
S_{\epsilon}(G)&:=\min\{r:(G,r)\textrm{ is $\epsilon$-achievable}\};
\\
S(G)&:=\sup_{\epsilon>0}S_{\epsilon}(G).
\label{e_56}
\end{align}
\end{defn}
Note that Definition~\ref{defn_worst} considers output distributions that correspond to some input to the channel, which is a subclass of all the distributions on the output alphabet.
In contrast, the memoryless output distribution in Definition~\ref{defn14} need not correspond to any input distribution. The reason for this dichotomy will be explained in Remark~\ref{rem_27} and at the beginning of Section~\ref{sec_worst}.

A useful property is that the approximation error \eqref{e66} actually converges \emph{uniformly}. A similar observation was made in the proof of \cite[Lemma~6]{han1993approximation} in the context of resolvability in total variation distance.
\begin{prop}\label{prop16}
If $(G,R)$ is $\epsilon$-achievable for $(Q_{X^n|U^n})_{n=1}^{\infty}$, then there exists $(\gamma_n)_{n=1}^{\infty}$ and $(M_n)_{n=1}^{\infty}$ such that
\begin{align}
\limsup_{n\to\infty}\frac{1}{n}\log\gamma_n&\le G;
\\
\limsup_{n\to\infty}\frac{1}{n}\log M_n&\le R;
\\
\sup_n \sup_{P_{U^n}}
\inf_{c^{M_n}}
E_{\gamma_n}(Q_{X^{n}[c^{M_n}]}\|Q_{X^n})
&\le \epsilon,
\end{align}
where $Q_{U^n}\to Q_{X^n|U^n}\to Q_{X^n}$.
\end{prop}
\begin{proof}
Fix arbitrary $G'>G$ and $R'>R$.
Observe that the sequence
\begin{align}
\sup_{Q_{U^n}}
\inf_{c^{M_n}}
E_{\exp(nG')}(Q_{X^{n}[c^{M_n}]}\|Q_{X^n})
\label{e64}
\end{align}
where $M_n:=\exp(nR')$,
must be upper-bounded by $\epsilon$ for large enough $n$. For if otherwise, there would be a sequence $(Q_{U^n})_{n=1}^{\infty}$ such that
the infimum in \eqref{e64}
is not upper-bounded by $\epsilon$ for large enough $n$, which is a contradiction since we can find $(c^{M_n})_{n=1}^{\infty}$ and $(\gamma_n)_{n=1}^{\infty}$ in Definition~\ref{defn14} such that $M_n<\exp(nR')$ and $\gamma_n\le \exp(nG')$ for $n$ large enough, and apply the monotonicity of $E_{\gamma}$ in $\gamma$. Finally, since $(G',R')$ can be arbitrarily close to $(G,R)$ and the $\epsilon$-achievable region is a closed set, we conclude that $(G,R)$ is $\epsilon$-achievable.
\end{proof}

\section{Source Resolvability}\label{sec_source}
This section is devoted to the source resolvability problem \cite{han1993approximation}, which can be viewed as channel resolvability with identity channels.
For source resolvability, we derive simple and tight one-shot bounds, the ideas of which are not easily extendable to the general channel resolvability problem.
To get a better grasp of the source resolvability problem, recall the definition of $M$-type distributions in \cite{han1993approximation}.
\begin{defn}[\cite{han1993approximation}]
A distribution $P$ on $\mathcal{X}$ is said to be an $M$-type ($M\in\mathbb{N}$) if for each $x\in\mathcal{X}$, the probability $P(x)$ is a multiple of $\frac{1}{M}$.
\end{defn}
Clearly, $P$ is an $M$-type if and only if there exists an identity random transformation $Q_{X|U}$ and a codebook $(c_m)_{m=1}^M$ such that $Q_{X[c^M]}=P$. Hence the source $E_{\gamma}$-resolvability can also be viewed as minimizing $M$ such that there exists an $M$-type distribution that approximates a given source distribution in $E_{\gamma}$.

One of the advantages of studying the source resolvability problem is that, in contrast to the general channel resolvability problem, we will be able to give a general one-shot converse.
Moreover, as noted in \cite{han1993approximation}, source resolvability is intimately connected to almost-lossless source coding and random number generation.
Now in the general case of source approximation under $E_{\gamma}$, suppose a random variable $X$ can be approximated by an $M$-type random variable $\hat{X}$ in the sense that
\begin{align}\label{e73}
E_{\gamma}(P_{\hat{X}}\|P_X)\le \delta
\end{align}
for some $\gamma,\delta>0$. Then we can approximately generate $X$ from an equiprobable $U_M\in\{1,\dots,M\}$ as $\hat{X}=\phi(U_M)$ with some deterministic function $\phi$. Again, \eqref{e73} has the operational meaning that the probability $P_{\hat{X}}(\mathcal{E})$ of any error event $\mathcal{E}$ associated with the random number is guaranteed to be small, say $\le 10^{-2}$, provided that the error probability under the target distribution is \emph{very} small, that is, provided that $P_{X}(\mathcal{E})\le \frac{1}{\gamma}\left(10^{-2}
-E_{\gamma}(P_{\hat{X}}\|P_X)\right)$.

\subsection{One-shot Achievability Bound}
We develop a simple one-shot achievability bound without using random coding, which is asymptotically tight for a general sequences of sources.\footnote{In contrast, the result obtained by particularizing the bound for channel resolvability to identity channels in Section~\ref{sec_ach} is not (Remark~\ref{rem20}), which is why source resolvability is discussed separately.}

For a discrete random variable $X\sim P$ on $\mathcal{X}$, define the \emph{information} for $x\in\mathcal{X}$
\begin{align}
\imath_{X}(x):=\log\frac{1}{P(x)}.
\end{align}
\begin{thm}\label{thm_sa}
For a discrete $X\sim P_X$, integer $M>0$, and $\gamma\ge 1$, there exists an $M$-type distribution $P_{\hat{X}}$ such that
\footnote{
For $a\in\mathbb{R}$, we use the notation
$
[a]^+:=\max\{a,0\}.
$}
\begin{align}
E_{\gamma}(P_{\hat{X}}\|P_X)\le \left[1-\frac{\gamma}{2}\mathbb{P}\left[\imath_X(X)\le\log \gamma M\right]\right]^+\label{e_73}
\end{align}
\end{thm}
\begin{proof}
Define an unnormalized measure $\mu_{\hat{X}}$ on $\mathcal{X}$ where
\begin{align}
\mu_{\hat{X}}(x):=\frac{1}{M}\left\lfloor M\gamma P_X(x)\right\rfloor
\end{align}
for each $x\in\mathcal{X}$.
Note that if $P_X(x)\ge\frac{1}{M\gamma}$, then $\mu_{\hat{X}}(x)\ge \frac{\gamma}{2} P_{X}(x)$; therefore
\begin{align}
\mu_{\hat{X}}(\mathcal{X})\ge \frac{\gamma}{2} \mathbb{P}\left[\imath_X(X)\le\log M\gamma\right].
\label{e75}
\end{align}
Moreover since $\mu_{\hat{X}}\le \gamma P_X$\footnote{For measures $\mu$ and $\nu$ on the same measurable space $(\mathcal{X},\mathscr{F})$, we write $\mu\le \nu$ if $\mu(\mathcal{A})\le \nu(\mathcal{A})$ for any $\mathcal{A}\in\mathscr{F}$.},
if $\mu_{\hat{X}}(\mathcal{X})>1$ then
$P_{\hat{X}}:=\frac{1}{\mu_{\hat{X}}(\mathcal{X})}\mu_{\hat{X}}$
is a probability measure such that
$P_{\hat{X}}\le \mu_{\hat{X}}$, hence $E_{\gamma}(P_{\hat{X}}\|P_X)
\le
E_{\gamma}(\mu_{\hat{X}}\|P_X)=0$, and so \eqref{e_73} holds;
otherwise, we introduce a symbol $\tt{e}$ not in $\mathcal{X}$, and set
\begin{align}
P_{\hat{X}}({\tt e})=1-\mu_{\hat{X}}(\mathcal{X})
\label{e_76}
\end{align}
so that $P_{\hat{X}}$ becomes a distribution on $\mathcal{X}\cup\{{\tt e}\}$. Then $P_X$ can also be viewed as a distribution on $\mathcal{X}\cup\{{\tt e}\}$, and we have
\begin{align}
E_{\gamma}(P_{\hat{X}}\|P_X)=P_{\hat{X}}({\tt e})
\end{align}
and so the result follows from \eqref{e75} and \eqref{e_76}.
\end{proof}
\begin{rem}
When $\gamma=1$, the bound \eqref{e_73} without the factor $1/2$ holds \cite{han2003information}.
\end{rem}

\subsection{One-shot Converse Bound}
The following result generalizes the bound \cite[Lemma~2.1.2]{han2003information} to values other than $\gamma = 1$.
\begin{thm}\label{thm_sc}
Suppose $P_X$ and an $M$-type distribution $P_{\hat{X}}$ are defined on the same discrete alphabet. Then for any $\gamma>0$, $a>0$,
\begin{align}
E_{\gamma}(P_{\hat{X}}\|P_X)\ge1-\gamma\mathbb{P}[\imath_X(X)<\log\gamma M+a]-\exp(-a).
\end{align}
\end{thm}
\begin{proof}
Define a set
\begin{align}
\mathcal{A}:=\{x\in\mathcal{X}:\imath_X(x)<\log\gamma M+a\}\cup\supp(\hat{X})
\end{align}
Then we have
\begin{align}
P_X(\mathcal{A})&\le \mathbb{P}[\imath_X(X)<\log\gamma M+a]
\nonumber\\
&\quad+\mathbb{P}[\imath_X(X)\ge\log\gamma M+a,X\in\supp(\hat{X})]
\\
&\le\mathbb{P}[\imath_X(X)<\log\gamma M+a]
+\frac{1}{\gamma M}\exp(-a)\cdot M \label{e77}
\end{align}
where \eqref{e77} is because $|\supp(\hat{X})|\le M$. Therefore the result follows since $E_{\gamma}(P_{\hat{X}}\|P_X)\ge P_{\hat{X}}(\mathcal{A})-\gamma P_X(\mathcal{A})$ and $P_{\hat{X}}(\mathcal{A})=1$.
\end{proof}
For the asymptotic analysis, we are interested in the regime where $\gamma$ and $M$ grow exponentially in $n$.
Combining Theorem~\ref{thm_sa} and Theorem~\ref{thm_sc}, we have
\begin{cor}\label{thm_s}
For an arbitrary sequence $\mb{X}=(X^n)_{n=1}^{\infty}$, and the identity channel, $\epsilon\in(0,1)$ and $G>0$,
\begin{align}
&\quad S_{\epsilon}(G,\mb{X})
\nonumber\\
&=\inf\left\{R:
\limsup_{n\to\infty}\frac{1}{n}\log\frac{1}{\mathbb{P}[\frac{1}{n}
\imath_{X^n}
(X^n)<R+G]}\le G\right\}.
\label{e_source}
\end{align}
\end{cor}
\begin{proof}
In Theorem~\ref{thm_sa}, note that for any $(G,R)$, if
\begin{align}
\limsup_{n\to\infty}\frac{1}{n}\log\frac{1}{\mathbb{P}[\frac{1}{n}
\imath_{X^n}
(X^n)<R+G]}\le G
\end{align}
then by Theorem~\ref{thm_sa} we see that $(G+\delta,R)$ is $\epsilon$-achievable for any $\delta>0$.
Since the set of $\epsilon$-achievable pairs is closed, $(G,R)$ is also $\epsilon$-achievable.
This shows the $\le $ part of \eqref{e_source}.
If $(G,R)$ is such that
\begin{align}
\liminf_{n\to\infty}\frac{1}{n}\log\frac{1}{\mathbb{P}[\frac{1}{n}
\imath_{X^n}
(X^n)<R+G]}>G,
\end{align}
then by Theorem~\ref{thm_sc}, $(G,R)$ is not $\epsilon$-achievable, hence $S_{\epsilon}(G,\mb{X})>R$. This shows the $\ge$ part of \eqref{e_source}.
\end{proof}
\begin{rem}\label{rem20}
In Section~\ref{sec_ach}, we develop achievability bounds for channel resolvability (e.g.~\eqref{t3}), which imply the following bound on source resolvability
\begin{align}\label{e84}
S(G,\mb{X})
&\le \sup_{\epsilon>0}\inf\left\{R:\liminf_{n\to\infty}\mathbb{P}\left[
\frac{1}{n}\imath_{X^n}(X^n)<R+G\right]\ge1-\epsilon\right\}
\\
&=\bar{H}(\mb{X})-G
\end{align}
when particularized to the identity channel, where
\begin{align}
\bar{H}(\mb{X}):=\sup_{\epsilon>0}\inf\left\{R:\liminf_{n\to\infty}
\mathbb{P}\left[\frac{1}{n}\imath_{X^n}(X^n)<R\right]\ge1-\epsilon\right\}
\end{align}
is the \emph{sup-entropy rate} defined in \cite{han1993approximation}. In view of Corollary~\ref{thm_s}, the achievability bound \eqref{e84} is not tight in general.
Indeed, the achievability construction in the proof of Theorem~\ref{thm_sa} is based on quantizing a scaling of the source distribution, which is more efficient than the random coding argument in the channel counterpart.
\end{rem}
In the case of discrete memoryless sources, we obtain a more explicit formula by performing large deviation analysis of Corollary~\ref{thm_s}:
\begin{cor}\label{thm_s1}
For any per-letter distribution ${{\pi}}_{\sf X}$ on a finite alphabet, real number $G>0$, and the identity channel,
\begin{align}
S_{\epsilon}(G,{{\pi}}_{\sf X})=\inf_{Q_{\sf X}:D(Q_{\sf X}\|{{\pi}}_{\sf X})\le G} [D(Q_{\sf X}\|{{\pi}}_{\sf X})+H(Q_{\sf X})-G]^+.
\end{align}
\end{cor}
\begin{proof}
By large deviation/the method of types,
\begin{align}
&\limsup_{n\to\infty}\frac{1}{n}\log\frac{1}{\mathbb{P}[\imath_{X^n}
(X^n)<R+G]}
\nonumber\\
&\quad=
\inf_{Q_{\sf X}\colon D(Q_{\sf X}\|{{\pi}}_{\sf X})+H(Q_{\sf X})-G\le R}
D(Q_{\sf X}\|\pi_{\sf X}),
\end{align}
thus $(G,R)$ is $\epsilon$-achievable if and only if there exists a $Q_{\sf X}$ such that
\begin{align}
D(Q_{\sf X}\|\pi_{\sf X})&\le G;
\\
D(Q_{\sf X}\|{{\pi}}_{\sf X})+H(Q_{\sf X})-G&\le R,
\end{align}
and the result follows.
\end{proof}
This result is a special case of channel resolvability for memoryless outputs (Theorem~\ref{thm_conv}).

\section{Channel Resolvability}\label{sec_channel}
In this section, we first derive achievability bounds for channel resolvability using random coding.
For discrete memoryless channels with iid target distributions, we prove a converse bound which, combined with the achievability bounds, yields the exact expression for $G_{\epsilon}(R,{{\pi}}_{\sf X})$.
For the worst-case distributions, we prove a converse bound which does not match the achievability bound, but addresses certain properties of $S(G)$.

Since the excess relative information is a trivial upper bound on $E_{\gamma}$, it suffices to derive achievability (upper-bounds) for the former and converse (lower-bounds) for the latter. In fact, by Proposition~\ref{prop_1}-\ref{pt13_1}) we know that the two metrics are equivalent in large deviation analysis.
\subsection{One-shot Achievability Bound}\label{sec_ach}
We first present the result in a simple special case where the target distribution matches the input distribution according to which the codewords are generated.
\begin{thm}[Softer-covering Lemma]\label{thm3_0}
Fix $Q_{UX}=Q_UQ_{X|U}$. For an arbitrary
codebook $[c_1,\dots,c_M]$, define
\begin{align}
Q_{X[c^M]}:=\frac{1}{M}\sum_{m=1}^M Q_{X|U=c_m}.
\label{e80_1st}
\end{align}
Then for any
$\gamma,\epsilon,\sigma>0$ satisfying $\gamma-1>\epsilon+\sigma$ and $\tau\in\mathbb{R}$,
\begin{align}
\mathbb{E}[\bar{F}_{\gamma}\left(Q_{X[U^M]}\|Q_X\right)]
&\le \mathbb{P}\left[\imath_{U;X}(U;X)>\log M\sigma\right]
\nonumber
\\
&\quad
+\frac{1}{\epsilon}\mathbb{P}[\imath_{U;X}(U;X)>\log M-\tau]
\nonumber\\
&\quad+\frac{\exp(-\tau)}{(\gamma-1-\epsilon-\sigma)^2}
\label{e_ub0}
\end{align}
where $U^M\sim Q_U\times\dots\times Q_U$, $(U,X)\sim Q_U Q_{X|U}$, and the information density
\begin{align}
\imath_{U;X}(u;x)
:=\log\frac{{\rm d}Q_{X|U=u}}{{\rm d}Q_X}(x).
\end{align}
\end{thm}
\begin{rem}
Theorem~\ref{thm3_0} implies (the general asymptotic version of) the soft-covering lemma based on total variation (see \cite{han1993approximation},\cite{hayashi2006general} \cite{cuff2012distributed}). Indeed if $M_n=\exp(nR)$ at blocklength $n$ where $R>\mb{\bar{I}}(\mb{U};\mb{X})$, we can select $\tau_n \leftarrow\frac{n}{2} (R-\mb{\bar{I}}(\mb{U};\mb{X}))$.
Moreover for any $\gamma>1$ we can pick constant $\epsilon,\sigma>0$ such that $\gamma-1>\epsilon+\sigma$ in the theorem, to show that
\begin{align}
\lim_{n\to\infty}\mathbb{E}[\bar{F}_{\gamma}
\left(Q_{X^n[U^{nM_n}]}\|Q_{X^n}\right)]
=0
\end{align}
which, by \eqref{e_hanbound} and by taking $\gamma\downarrow1$, implies that
\begin{align}
\lim_{n\to\infty}\mathbb{E}|Q_{X^n[U^{nM_n}]}-Q_{X^n}|=0.
\end{align}
We refer to Theorem~\ref{thm3_0} as the softer-covering lemma since for a larger $\gamma$ it allows us to use a smaller codebook to cover the output distribution more softly (i.e.~approximate the target distribution under a weaker metric).
\end{rem}
\begin{proof}
Define the ``atypical'' set
\begin{align}
%
\mathcal{A}_{\tau}&:=\{(u,x):\imath_{U;X}(u;x)\le \log M-\tau\}.
\label{e_at}
\end{align}
Now,
let $X^M$ be such that $(U^M,X^M)\sim Q_{UX}\times\dots\times Q_{UX}$,
The joint distribution of $(U^M,\hat{X})$ is specified by letting $\hat{X}\sim Q_{X[c^M]}$ conditioned on $U^M=c^M$.
We perform a change-of measure step using the symmetry of the random codebook:
\begin{align}
\mathbb{E}[\bar{F}_{\gamma}(Q_{X[U^M]}\|Q_X)]
&=\mathbb{P}\left[\frac{{\rm d}Q_{X[U^M]}}{{\rm d}Q_X}(\hat{X})>\gamma\right]
\\
&=
\frac{1}{M}\sum_m\mathbb{P}\left[\frac{{\rm d}Q_{X[U^M]}}{{\rm d}Q_X}(X_m)>\gamma\right]
\label{e89}
\\
&=\mathbb{P}\left[\frac{{\rm d}Q_{X[U^M]}}{{\rm d}Q_X}(X_1)>\gamma\right]
\label{e92}
\end{align}
where \eqref{e89} is because of \eqref{e80_1st}, and \eqref{e92} is because the summands in \eqref{e89} are equal. Note that $X_1$ is correlated with only the first codeword $U_1$.

Next, because of the relation
\begin{align}
\frac{{\rm d}Q_{X[c^M]}}{{\rm d}Q_X}(x)=\frac{1}{M}\sum_m\exp(\imath_{U;X}(c_m,x)),
\label{e_88}
\end{align}
we can upper-bound \eqref{e92} by the union bound as
\begin{align}
&\quad\mathbb{P}[\exp(\imath_{U;X}(U_1;X_1))>M\sigma]
\nonumber\\
&+\mathbb{P}\left[\frac{1}{M}\sum_{m=2}^M \exp(\imath_{U;X}(U_m;X_1))
1_{\mathcal{A}_{\tau}^c}(U_m,X_1)>\epsilon\right]
\nonumber\\
&+\mathbb{P}\left[
\frac{1}{M}\sum_{m=2}^M \exp(\imath_{U;X}(U_m;X_1))
1_{\mathcal{A}_{\tau}}(U_m,X_1)>\gamma-\epsilon-\sigma
\right]
\label{e_ub92}
\end{align}
where we used the fact that $1_{\mathcal{A}_{\tau}}+1_{\mathcal{A}_{\tau}^c}=1$. Notice that the first term of \eqref{e_ub92} may be regarded as the probability of ``atypical'' events and accounts for the first term in \eqref{e_ub0}.
The second term of \eqref{e_ub92} can be upper-bounded with Markov's inequality:
\begin{align}
&\quad\frac{1}{M\epsilon}\sum_{m=2}^M
\mathbb{E}[\exp\left(\imath_{U;X}(U_m;X_1)\right)
1_{\mathcal{A}_{\tau}^c}(U_m,X_1)]
\nonumber\\
&\le
\frac{1}{M\epsilon}\sum_{m=2}^M
\mathbb{E}[1_{\mathcal{A}_{\tau}^c}(U,X)]
\label{e_cm0}
\\
&\le \frac{1}{\epsilon}\,\mathbb{P}[\imath_{U;X}(U;X)>\log M-\tau]
\end{align}
accounting for the second term in \eqref{e_ub0} where \eqref{e_cm0} is a change-of-measure step using the fact that $(U_l,X_1)\sim Q_U\times Q_X$ for $m\ge 2$.

Finally we take care of the last term in \eqref{e_ub92},
again using the independence of $U_m$ and $X_1$ for $m \geq 2$.
Observe that for any $x\in\mathcal{X}$,
\begin{align}
\mu&:=\mathbb{E}\left[\frac{1}{M}\sum_{m=2}^M \exp(\imath_{U;X}(U_m;x))
1_{\mathcal{A}_{\tau}}(U_m,x)\right]
\\
&\le
\frac{1}{M}\sum_{m=2}^M\mathbb{E}\left[\exp(\imath_{U;X}(U_m;x))
\right]
\\
&=\frac{M-1}{M}
\\
&\le 1,
\end{align}
whereas
\begin{align}
&\quad\Var\left(\frac{1}{M}\sum_{m=2}^M \exp(\imath_{U;X}(U_m;x))
1_{\mathcal{A}_{\tau}}(U_m,x)\right)
\nonumber\\
&=
\frac{1}{M^2}\sum_{m=2}^M\Var\left( \exp(\imath_{U;X}(U_m;x))
1_{\mathcal{A}_{\tau}}(U_m,x)\right)
\\
&\le
\frac{1}{M}\Var\left( \exp(\imath_{U;X}(U;x))
1_{\mathcal{A}_{\tau}}(U,x)\right)
\\
&\le
\frac{1}{M}\mathbb{E}\left[ \exp(2\imath_{U;X}(U;x))
1_{\mathcal{A}_{\tau}}(U,x)\right]
\label{e112}
\\
&\le
\exp(-\tau)\mathbb{E}\left[ \exp(\imath_{U;X}(U;x))
\right]
\label{e113}
\\
&=\exp(-\tau),
\end{align}
where
the change of measure step \eqref{e113} uses \eqref{e_at}.
It then follows from Chebyshev's inequality that
{\small
\begin{align}
&\quad\mathbb{P}\left[
\frac{1}{M}\sum_{m=2}^M \exp(\imath_{U;X}(U_m;x))
1_{\mathcal{A}_{\tau}}(U_m,x)>\gamma-\epsilon-\sigma
\right]
\\
&\le\mathbb{P}\left[
\frac{1}{M}\sum_{m=2}^M \exp(\imath_{U;X}(U_m;x))
1_{\mathcal{A}_{\tau}}(U_m,x)-\mu
>\gamma-\epsilon-\sigma-1
\right]
\\
&\le \frac{\exp(-\tau)}{(\gamma-\epsilon-\sigma-1)^2}.
\end{align}
}
\end{proof}
For the asymptotic analysis,
we are interested in the regime where $M$ and $\gamma$ are growing exponentially.
In this case, the right hand side of \eqref{e_ub0} can be regarded as essentially
\begin{align}
\mathbb{P}\left[\imath_{U;X}(U;X)>\log M\gamma\right]
\end{align}
modulo nuisance parameters.
This can be seen from the choice of parameters in Corollary~\ref{cor_asymp}.
Thus the sum rate of $M$ and $\gamma$ has to exceed the sup information rate in order that the approximation error vanishes asymptotically.

Extending Theorem~\ref{thm3_0} to the more general scenario where the target distribution may not have any relation with the input distribution, we have the following result, where we allow $\pi_X$ to be an arbitrary positive measure,
\begin{thm}[Softer-covering Lemma: Unmatched Target Distribution]\label{thm3}
Fix ${{\pi}}_X$ and $Q_{UX}=Q_UQ_{X|U}$. For an arbitrary
codebook $[c_1,\dots,c_M]$,
define $Q_{X[c^M]}$
as in \eqref{e80_1st}.
Then for any
$\gamma,\epsilon,\sigma>0$ satisfying $\gamma>\epsilon+\sigma$, $\tau\in\mathbb{R}$ and $0<\delta<1$,
\begin{align}
&\quad\mathbb{E}[\bar{F}_{\gamma}\left(Q_{X[U^M]}\|{{\pi}}_X\right)]
\nonumber\\
&\le \mathbb{P}[\imath_{Q_X\|{{\pi}}_X}(X)>
\log(\gamma-\sigma-\epsilon)\textrm{ or }\imath_{U;X}(U;X)
\nonumber\\
&\quad\quad+\imath_{Q_X\|{{\pi}}_X}(X)>\log\delta M\sigma]
\nonumber
\\
&\quad
+\frac{\gamma-\epsilon-\sigma}{\epsilon}\mathbb{P}[\imath_{U;X}(U;X)>\log M-\tau]
\nonumber\\
&\quad+\frac{\exp(-\tau)(\gamma-\epsilon-\sigma)^2}{(1-\delta)^2\sigma^2}
\label{e_ub}
\end{align}
where $U^M\sim Q_U\times\dots\times Q_U$ and $(U,X)\sim Q_U Q_{X|U}$.
\end{thm}
\begin{proof}[Proof Sketch]
Similarly to the proof of Theorem~\ref{thm3_0}, we first use a symmetry argument and change-of-measure step so that the random variable of the channel output is correlated only with the first codeword, to obtain
\begin{align}
\mathbb{E}[\bar{F}_{\gamma}\left(Q_{X[U^M]}\|{{\pi}}_X\right)]
\le
\mathbb{P}\left[\frac{{\rm d}Q_{X[U^M]}}{{\rm d}\pi_X}(X_1)>\gamma\right].
\end{align}
Then in the next union bound step we have to take care of another ``atypical'' event that $\frac{{\rm d}Q_X}{{\rm d}\pi_X}(X_1)>\gamma_2$, where
\begin{align}
\gamma_2:=\gamma-\epsilon-\sigma.
\end{align}
More precisely, we have
{\small
\begin{align}
&\quad\mathbb{P}\left[\frac{{\rm d}Q_{X[U^M]}}{{\rm d}{{\pi}}_X}(X_1)>\gamma\right]
\nonumber\\
&\le \mathbb{P}[\xi(X_1)>\gamma_2\textrm{ or }\eta(U_1,X_1)>\delta M\sigma]
\nonumber\\
&\quad+\mathbb{P}\left[\frac{1}{M}\sum_{m=2}^M \eta(U_m,X_1)
1_{\mathcal{A}_{\tau}^c}(U_m,X_1)>\epsilon,~\xi(X_1)\le\gamma_2\right]
\nonumber\\
&\quad+\mathbb{P}\left[
\frac{1}{M}\sum_{m=2}^M \eta(U_m,X_1)
1_{\mathcal{A}_{\tau}}(U_m,X_1)>\gamma-\epsilon-\delta \sigma,~
\xi(X_1)\le \gamma_2
\right]
\label{e108}
\end{align}
}
where we have defined
\begin{align}
\eta(u,x)&:=\frac{{\rm d}Q_{X|U=u}}{{\rm d}{{\pi}}_X}(x);
\\
\xi(x)&:=\frac{{\rm d}Q_X}{{\rm d}{{\pi}}_X}(x).
\end{align}
As before the first term of \eqref{e108} may be regarded as the probability of ``atypical'' events and accounts for the first term in \eqref{e_ub}. The second and the third terms of \eqref{e108} can be upper-bounded by
\begin{align}
&\quad\mathbb{P}\left[\frac{1}{M}\sum_{m=2}^M \frac{\eta(U_m,X_1)}
{\xi(X_1)}
1_{\mathcal{A}_{\tau}^c}(U_m,X_1)>\frac{\epsilon}{\gamma_2}\right]
\nonumber
\\
&\le
\mathbb{P}\left[\frac{1}{M}\sum_{m=2}^M
\exp\left(\imath_{U;X}(U_m;X_1)\right)
1_{\mathcal{A}_{\tau}^c}(U_m,X_1)>\frac{\epsilon}{\gamma_2}\right]
\end{align}
and
{\small
\begin{align}
&\quad\mathbb{P}\left[
\frac{1}{M}\sum_{m=2}^M \frac{\eta(U_m,X_1)}{\xi(X_1)}
1_{\mathcal{A}_{\tau}}(U_m,X_1)>\frac{\gamma-\epsilon-\delta \sigma}{\gamma_2}
\right]
\nonumber
\\
&\le
\mathbb{P}\left[
\frac{1}{M}\sum_{m=2}^M \exp(\imath_{U;X}(U_m;X_1))
1_{\mathcal{A}_{\tau}}(U_m,X_1)>1+\frac{(1-\delta)\sigma}{\gamma_2}
\right].
\end{align}
}
The rest of the proof is similar to that of Theorem~\ref{thm3_0} and is omitted.
\end{proof}
For the purpose of asymptotic analysis in the stationary memoryless setting, the right hand side of \eqref{e_ub} can be regarded as essentially
\begin{align}
&\quad\mathbb{P}\left[\imath_{Q_X\|{{\pi}}_X}(X)>\log\gamma\right]
\nonumber\\
&+\mathbb{P}\left[\imath_{U;X}(U;X)+\imath_{Q_X\|{{\pi}}_X}(X)>\log M\gamma\right]
\end{align}
modulo nuisance parameters.

\begin{rem}\label{remweaken}
By setting $\tau\leftarrow+\infty$ and letting $\delta\uparrow 1$, the bound in Theorem~\ref{thm3} can be weakened in the following slightly simpler form:
\begin{align}
\mathbb{E}[\bar{F}_{\gamma}(Q_{X[U^M]}\|{{\pi}}_X)]
&\le \mathbb{P}\left[\frac{{\rm d}Q_X}{{\rm d}{{\pi}}_X}(X)>\gamma-\sigma-\epsilon\right]\nonumber%
\\
&\quad+\mathbb{P}\left[\frac{{\rm d}Q_{X|U}}{{\rm d}{{\pi}}_X}(X|U)\ge M\sigma\right]\nonumber
\\
&\quad+\frac{\gamma-\sigma-\epsilon}{\epsilon}.\label{tt1}
\end{align}
In fact, assuming $\tau=+\infty$ we can simplify the proof of the theorem and strengthen \eqref{tt1} to
\begin{align}
\mathbb{E}[\bar{F}_{\gamma}(Q_{X[U^M]}\|{{\pi}}_X)]
&\le \mathbb{P}\left[\frac{{\rm d}Q_X}{{\rm d}{{\pi}}_X}(X)>\gamma_2\right]\nonumber
\\
&\quad+\mathbb{P}\left[\frac{{\rm d}Q_{X|U}}{{\rm d}{{\pi}}_X}(X|U)> M(\gamma-\epsilon)\right]\nonumber
\\
&\quad+\frac{\gamma_2}{\epsilon}\label{t3}
\end{align}
for any $\gamma_2>0$ and $0<\epsilon<\gamma$.
As we show in Corollary~\ref{cor_asymp}, the weakened bounds \eqref{tt1} and \eqref{t3} are still asymptotically tight provided that the exponent with which the threshold $\gamma$ grows is strictly positive. However, when the exponent is zero (corresponding to the total variation case), we do need $\tau$ in the bound for asymptotic tightness.
\end{rem}
\begin{rem}\label{rem25}
In the case of $Q_X={{\pi}}_X$, we can set
$\gamma_2\leftarrow1$ and $\epsilon\leftarrow\frac{\gamma}{2}$ in \eqref{t3} to obtain
the simplification
\begin{align}
\mathbb{E}[\bar{F}_{\gamma}(Q_{X[U^M]}\|Q_X)]
&\le\mathbb{P}\left[\exp(\imath_{U;X}(U;X))> \log\frac{M\gamma}{2}\right]
+\frac{2}{\gamma}.
\end{align}
\end{rem}
\begin{cor}\label{cor_asymp}
Fix per-letter distributions ${{\pi}}_{\sf X}$ on $\mathcal{X}$ and $Q_{\sf UX}=Q_{\sf U}Q_{\sf X|U}$ on $\mathcal{U}\times\mathcal{X}$,
and $E,R\in (0,\infty)$.
For each $n$,
define $\gamma_n:=\exp(nE)$ and $M_n=\lfloor\exp(nR)\rfloor$;
let $U^{M_n}=(U_1,\dots,U_{M_n})$ have independent coordinates each distributed according to $Q^{\otimes n}_{\sf U}$.
Given any $c^{M_n}=(c_{m})_{m=1}^{M_n}$, where each $c_{m}\in \mathcal{U}^n$, define\footnote{We define
$Q_{\sf X|U}^{\otimes n}$ by
$Q_{\sf X|U}^{\otimes n}(\cdot|u^n):=\prod_{i=1}^nQ_{{\sf X|U}=u_i}$ for any $u^n$.
Also note that in this paper we differentiates per-letter symbols such as $\sf U$ between one-shot/block symbols such as $U$ (so that $U={\sf U}^n$ in this corollary).}
\begin{align}
Q_{X[c^{M_n}]}:=\frac{1}{M_n}\sum_{m=1}^{M_n}
Q^{\otimes n}_{{\sf X}|{\sf U}}(\cdot|c_m).
\end{align}
Then
\begin{align}
\lim_{n\to\infty}\mathbb{E}[E_{\gamma_n}
(Q_{X[U^{M_n}]}||{{\pi}}_{{\sf X}}^{\otimes n})]
&=\lim_{n\to\infty}
\mathbb{E}[\bar{F}_{\gamma_n}(Q_{X[U^{M_n}]}\|{{\pi}}_{\sf X}^{\otimes n})]
\\
&=0
\end{align}
provided that
\begin{align}\label{e_e}
E>D(Q_{\sf X}||{{\pi}}_{\sf X})+[I(Q_{\sf U},Q_{\sf X|U})-R]^+.
\end{align}
\end{cor}
\begin{proof}Choose $E'$ such that
\begin{align}
E>E'>D(Q_{\sf X}||{{\pi}}_{\sf X})+[I(Q_{\sf U},Q_{\sf X|U})-R]^+.
\end{align}
Set $\delta=\frac{1}{2}$, $\epsilon_n=\exp(nE)-\exp(nE')$ and $\sigma_n=\frac{1}{2}(\gamma_n-\epsilon_n)=\frac{1}{2}\exp(nE')$, and apply \eqref{tt1}. Notice that
\begin{align}
\mathbb{E}\left[\log\frac{{\rm d}Q_{X|U}}{{\rm d}{{\pi}}_X}(X|U)\right]=n[I(Q_{\sf U},Q_{\sf X|U})+D(Q_{\sf X}||{{\pi}}_{\sf X})]
\end{align}
where $(X,U)\sim Q_{\sf XU}^{\otimes n}$,
$Q_{X|U}:=Q_{\sf X|U}^{\otimes n}$,
and $\pi_X:=\pi_{\sf X}^{\otimes n}$.
By the law of large numbers, the first and second terms in \eqref{tt1} vanish because
\begin{align}
D(Q_{\sf X}||{{\pi}}_{\sf X})&<E';
\\
I(Q_{\sf U},Q_{\sf X|U})+D(Q_{\sf X}||{{\pi}}_{\sf X})&<E'+R
\end{align}
are satisfied.
\end{proof}
The $Q_{\sf U}$ that minimizes the right hand side of \eqref{e_e} generally does not satisfy $Q_{\sf U}\to Q_{\sf X|U}\to Q_{\sf X}$. This means that in the large deviation analysis, for the best approximation of a target distribution in $E_{\gamma}$, we generally should not generate the codewords according to a distribution that corresponds to the target through the channel. This is a remarkable distinction from approximation in total variation distance, in which case an unmatched input distribution would result in the maximal total variation distance asymptotically. 
However, if we stick to matching input codeword distributions, then a simple and general asymptotic achievability bound can be obtained. Recall that \cite{han1993approximation} defined the \emph{sup-information rate}
\begin{align}
\mb{\bar{I}}(\mb{U};\mb{X}):=\inf\left\{R:\lim_{n\to\infty}
\mathbb{P}\left[\frac{1}{n}\imath_{U^n;X^n}
(U^n;X^n)> R\right]=0\right\}
\end{align}
and the \emph{inf-information rate}
\begin{align}
\mb{\underline{I}}(\mb{U};\mb{X}):=\sup\left\{R:\lim_{n\to\infty}
\mathbb{P}\left[\frac{1}{n}\imath_{U^n;X^n}
(U^n;X^n)< R\right]=0\right\}.
\end{align}
\begin{thm}\label{thm_ma}
For any channel $\mb{W}=(Q_{X^n|U^n})_{n=1}^{\infty}$, sequence of inputs $\mb{U}=(U^n)_{n=1}^{\infty}$, and $G>0$, we have
\begin{align}
S(G,\mb{X})\le \mb{\bar{I}}(\mb{U};\mb{X})-G
\end{align}
where $\mb{X}$ is the output of $\mb{U}$ through the channel $\mb{W}$. As a consequence,
\begin{align}
S(G)\le \sup_{\mb{U}}\mb{\bar{I}}(\mb{U};\mb{X})-G
\label{e110}
\end{align}
\end{thm}
For channels satisfying the \emph{strong converse property}, the right hand side of \eqref{e110} can be related to the channel capacity because of the relations \cite{han1993approximation}\cite{verdu1994general}
\begin{align}
\sup_{\mb{U}}\mb{\bar{I}}(\mb{U};\mb{X})
=\sup_{\mb{U}}\mb{\underline{I}}(\mb{U};\mb{X})\equiv C(\mb{W}).
\end{align}

We conclude the subsection by remarking that had we used the soft-covering lemma to bound total variation distance and in turn, bounded the excess relative information with total variation distance, we would not have obtained Theorem~\ref{thm_ma}.
Indeed, consider $M_n=\exp(nR)$ and let $V_1,\dots,V_{M_n}$ be i.i.d.~according to $Q_{U^n}$.
Regardless of how fast $\gamma_n$ grows, we cannot conclude from the optimal upper bound \eqref{ub} that
\begin{align}\label{emetric}
\mathbb{E}[\bar{F}_{\gamma_n}(Q_{X^n[V^{M_n}]}
\|
{{\pi}}_{X^n})
]
\end{align}
vanishes unless $\mathbb{E}|Q_{X^n[V^{M_n}]}-{{\pi}}_{X^n}|$ vanishes, which happens when $R>\bar{\mb{I}}(\mb{U};\mb{X})$ by the conventional resolvability theorem. This gives an upper bound $\bar{\mb{I}}(\mb{U};\mb{X})$ which is looser than Theorem~\ref{thm_ma} when $G>0$.

\subsection{Tail Bound of Approximation Error for Random Codebooks}\label{sec_tail}
For applications such as secrecy and channel synthesis, it is sometimes desirable to prove that the approximation error vanishes not only in expectation (e.g.~Theorem~\ref{thm3}), but also with \emph{high probability} (see Footnote~\ref{ft_high}), in the case of a random codebook \cite{cuff2016}\cite{Goldfeld2016}\cite{tahmasbi2016}.
If the probability that the approximation error exceeds an arbitrary positive number vanishes doubly exponentially in the blocklength, then the analyses in these applications carry through because a union bound argument can be applied to exponentially many events.
Previous proofs (e.g.~\cite{cuff2016}) based on carefully applying Chernoff bounds to each $Q_{X[U^M]}(x)-Q_X(x)$ and then taking the union bound over $x$ require finiteness of the alphabets.

Here we adopt a different approach.
Using concentration inequalities we can directly bound the probability that the error $E_{\gamma}(Q_{X[U^M]}\|Q_X)$ deviates from its expectation, without any restrictions on the alphabet and in fact the bound only depends on the number of codewords.
Therefore if the rate is high enough for the approximation error to vanish in expectation (by Theorem~~\ref{thm3}),
we can also conclude that the error vanishes with high probability.
The crux of the matter is thus resolved by the following one-shot result:
\begin{thm}
Fix ${{\pi}}_X$ and $Q_{UX}=Q_UQ_{X|U}$. For an arbitrary
codebook $[c_1,\dots,c_M]$,
define $Q_{X[c^M]}$
as in \eqref{e80_1st}.
Then, for any $r>0$,
{\small
\begin{align}
\mathbb{P}[E_{\gamma}(Q_{X[U^M]}\|\pi_X)
-\mathbb{E}[E_{\gamma}(Q_{X[U^M]}\|\pi_X)]>r]
\le
\exp(-2Mr^2)
\label{e144}
\end{align}
}
where the probability and the expectation are with respect to $U^M\sim Q_X\times Q_X$.
\end{thm}
\begin{proof}
Consider $f\colon c^M\mapsto E_{\gamma}(Q_{X[c^M]}\|\pi_X)$.
By the definition \eqref{e80_1st} and the triangle inequality \eqref{e19}, we have the following uniform bound on the discrete derivative:
\begin{align}
\sup_{c,c'\in\mathcal{X}}|f(c_1^{i-1},c, c_{i+1}^M)-f(c_1^{i-1},c', c_{i+1}^M)|
\le \frac{1}{M},
\quad \forall i,\,c^M.
\end{align}
The result then follows by McDiarmid's inequality (see e.g.~\cite[Theorem~2.2.3]{raginsky2014concentration}).
\end{proof}
\begin{rem}
If we are interested in bounding both the upper and the lower tails then the right side of \eqref{e144} gains a factors of $2$.
Other concentration inequalities may also be applied here; the transportation method gives the same bound in this example.
\end{rem}

\subsection{Converse for Stationary Memoryless Outputs}\label{sec_tens}
In this section we establish a converse of resolvability for stationary memoryless outputs and discrete memoryless channels which matches the achievability bound of Corollary~\ref{cor_asymp} asymptotically.

\begin{thm}[Resolvability for Stationary Memoryless Outputs]\label{thm_conv}
For a DMC $Q_{\sf X|U}$ and a nonnegative finite measure $\pi_{\sf X}$,
\begin{align}
G_{\epsilon}(R,{{\pi}}_{\sf X})= \min_{Q_{\sf U}}
\{D(Q_{\sf X}\|{{\pi}}_{\sf X})+[I(Q_{\sf U},Q_{\sf X|U})-R]^+\}
\label{e_vi5}
\end{align}
where $Q_{\sf U}\to Q_{\sf X|U}\to Q_{\sf X}$,
for any $0<\epsilon<1$.
\end{thm}
\begin{rem}\label{rem_27}
When resolvability was introduced in \cite{han1993approximation}, the resolvability rate (under total variation distance) was formulated for
outputs of stationary memoryless \emph{inputs}, rather than all the tensor power distributions on the output alphabet, because otherwise there is no guarantee that the output process can be approximated under total variation distance even with an arbitrarily large codebook. Here we can extend the scope because all stationary memoryless distributions on the output alphabet (satisfying the mild condition of being absolutely continuous with respect to some output) can be approximated under $E_{\gamma}$ as long as $\gamma$ is sufficiently large.
\end{rem}
The achievability part of Theorem~\ref{thm_conv} is already shown in Corollary~\ref{cor_asymp}.
For the converse, we need a notion of conditional typicality specially tailored for our problem which differs from the definitions of conditional typicality in \cite{csiszar1981information} or \cite{orlitsky2001} (see also \cite{el2011network}).
This can be viewed as an intermediate of the those two definitions.
\begin{defn}[Moderate Conditional Typicality]\label{defn_moderate}
The $\delta$-typical set of $u^n\in\mathcal{U}^n$ with respect to the discrete memoryless channel with per-letter conditional distribution $Q_{\sf X|U}$ is defined as
\begin{align}
&\quad T_{[Q_{\sf X|U}]\delta}^n(u^n)
:=
\nonumber\\
&\{x^n:\forall a,b,~|\widehat{P}_{u^nx^n}(a,b)-Q_{\sf X|U}(b|a)\widehat{P}_{u^n}(a)|\le \delta Q_{\sf X|U}(b|a)\}
\label{e_vi6}
\end{align}
where $\widehat{P}_{u^nx^n}$ denotes the empirical distribution of $(u^n,x^n)$.
\end{defn}
\begin{rem}
In addition to its broad interest, Definition~\ref{defn_moderate} plays an important role in obtaining the uniform bound in Lemma 35, as well as in Lemma 36.
This definition of conditional typicality is of broad interest because of Lemma~\ref{lem_delta} and Lemma~\ref{lem_vi6} ahead, and in particular the uniform bound in Lemma~\ref{lem_delta}.
Note that the
definition in \cite{csiszar1981information}
corresponds to replacing the term $\delta Q_{\sf X|U}(b|a)$ in \eqref{e_vi6} with $\delta$, in which case we cannot bound the probability of a sequence in the typical set as in Lemma~\ref{lem_vi6}.
The ``robust typicality" definition of \cite{orlitsky2001} (see also \cite{el2011network})
corresponds to replacing this term with $\delta Q_{\sf X|U}(b|a)\widehat{P}_{u^n}(a)$, which does not give the uniform lower bound on the probability of conditional typical set as in Lemma~\ref{lem_delta}.
\end{rem}
\begin{lem}\label{lem_delta}
For fixed $\delta>0$ and $Q_{\sf X|U}$, there exists a sequence $(\gamma_n)$ such that $\lim_{n\to\infty}\gamma_n=0$ and
\begin{align}
Q_{\sf X|U}^{\otimes n}(T_{[Q_{
\sf X|U}]\delta}^n(u^n)|u^n)\ge 1-\gamma_n
\end{align}
for all $u^n\in\mathcal{U}^n$.
\end{lem}
\begin{proof}
We show that the statement holds with
\begin{align}
\gamma_n=\frac{|\mathcal{U}\|\mathcal{X}|}
{n\delta^2}\left(\frac{1}{q}-1\right),
\end{align}
where
\begin{align}
q:=\min_{(a,b):Q_{\sf X|U}(b|a)\neq 0} Q_{\sf X|U}(b|a).
\end{align}
The \emph{number of occurrences} $N(a,b|u^n,X^n)$ of $(a,b)\in\mathcal{U}\times\mathcal{X}$ in $(u^n,X^n)$, where $X^n\sim \prod_{i=1}^n Q_{{\sf X|U}=u_i}$, is binomial with mean $N(a|u^n)Q_{\sf X|U}(b|a)$ and variance $N(a|u^n)Q_{\sf X|U}(b|a)(1-Q_{\sf X|U}(b|a))$. If $Q_{\sf X|U}(b|a)=0$ the condition that defines the set in \eqref{e_vi6} is automatically true. Otherwise, by Chebyshev's inequality we have for each $(a,b)$,
\begin{align}
&\quad\mathbb{P}[|N(a,b|u^n,X^n)-N(a|u^n)Q_{\sf X|U}(b|a)|>n\delta Q_{\sf X|U}(b|a)]
\nonumber
\\
&\le \frac{N(a|u^n)Q_{\sf X|U}(b|a)(1-Q_{\sf X|U}(b|a))}{n^2\delta^2Q_{\sf X|U}^2(b|a)}
\\
&\le \frac{1}{n\delta^2}\left(\frac{1}{q}-1\right)
\end{align}
and the claim follows by taking the union bound.
\end{proof}
\begin{lem}\label{lem_vi6}
For each $u^n$, and $x^n\in T_{[Q_{\sf X|U}]\delta}^n(u^n)$, we have the bound
\begin{align}
Q_{\sf X|U}^{\otimes n}(x^n|u^n)
\ge \exp(-n[H(Q_{\sf X|U}|\widehat{P}_{u^n})+\delta|\mathcal{U}|\log|\mathcal{X}|]).
\end{align}
\end{lem}
\begin{proof}
Since
\begin{align}
Q_{\sf X|U}^{\otimes n}(x^n|u^n)=\prod_{a\in\mathcal{U},b\in\mathcal{X}}
Q_{\sf X|U}(b|a)^{N(a,b|u^n,x^n)},
\end{align}
we have
\begin{align}
&\quad\frac{1}{n}\log\frac{1}{Q_{\sf X|U}^{\otimes n}(x^n|u^n)}
\nonumber\\
&=\sum_{a,b}\widehat{P}_{u^nx^n}(a,b)\log\frac{1}{Q_{\sf X|U}(b|a)}
\\
&\le \sum_{a,b}[Q_{\sf X|U}(b|a)\widehat{P}_{u^n}(a)+\delta Q_{\sf X|U}(b|a)]\log\frac{1}{Q_{\sf X|U}(b|a)}
\\
&=H(Q_{\sf X|U}|\widehat{P}_{u^n})+\delta \sum_a H(Q_{{\sf X|U}=a})
\\
&\le H(Q_{\sf X|U}|\widehat{P}_{u^n})+\delta |\mathcal{U}|\log|\mathcal{X}|.
\end{align}
\end{proof}
\begin{lem}\label{lem_vi7}
For any type $P_{\sf X}$ and sequence $u^n$,
\begin{align}
|T_{[Q_{\sf X|U}]\delta}^n(u^n)\cap T_{P_{\sf X}}|
\le \exp(n[H_{[P_{\sf X}]\delta}+\delta |\mathcal{U}|\log|\mathcal{X}|])
\end{align}
where we have defined
\begin{align}
H_{[P_{\sf X}]\delta}:=\max_{Q_{\sf U}:|Q_{\sf X}-P_{\sf X}|\le \delta|\mathcal{U}|}H(Q_{\sf X|U}|Q_{\sf U})
\label{e_vi19}
\end{align}
where $Q_{\sf U}\to Q_{\sf X|U}\to Q_{\sf X}$,
and the maximum in \eqref{e_vi19} is understood as $-\infty$ if the set $\{Q_{\sf U}:|Q_{\sf X}-P_{\sf X}|\le \delta|\mathcal{U}|\}$ is empty.
\end{lem}
\begin{proof}
For any $u^n$, we have the upper bound
\begin{align}
|T_{[Q_{\sf X|U}]\delta}^n(u^n)\cap T_{P_{\sf X}}|
&\le |T_{[Q_{\sf X|U}]\delta}^n(u^n)|
\\
&\le \left(\min_{x^n\in T_{[Q_{\sf X|U}]\delta}^n(u^n)}Q_{\sf X|U}^{\otimes n}(x^n|u^n)\right)^{-1}
\\
&\le \exp(n[H(Q_{\sf X|U}|\widehat{P}_{u^n})+\delta|\mathcal{U}|\log|\mathcal{X}|])
\label{e_vi22}
\end{align}
where we used Lemma~\ref{lem_vi6} in \eqref{e_vi22}.
Moreover, if $u^n$ satisfies $|P_{\sf X}-Q_{\sf X|U}\circ \widehat{P}_{u^n}|>\delta |\mathcal{U}|$ where $Q_{\sf X|U}\circ \widehat{P}_{u^n}:=\int Q_{{\sf X|U}=a}{\rm d}\widehat{P}_{u^n}(a)$,
then $T_{[Q_{\sf X|U}]\delta}^n(u^n)\cap T_{P_{\sf X}}$ is empty, because
\begin{align}
|\widehat{P}_{x^n}-Q_{\sf X|U}\circ \widehat{P}_{u^n}|
&=\sum_{a,b}|\widehat{P}_{u^nx^n}(a,b)-Q_{\sf X|U}(b|a)\widehat{P}_{u^n}(a)|
\\
&\le\sum_{a,b}\delta Q_{\sf X|U}(b|a)
\\
&=\delta|\mathcal{U}|
\end{align}
implies that any $x^n$ in $T_{[Q_{\sf X|U}]\delta}^n(u^n)$ does not have the type $P_{\sf X}$. Therefore the desired result follows by taking the maximum of \eqref{e_vi22} over type $Q_{\sf U}$ satisfying $|Q_{\sf X}-P_{\sf X}|\le \delta|\mathcal{U}|$.
\end{proof}

\begin{proof}[Proof of Converse of Theorem~\ref{thm_conv}]
Fix a codebook $(c_1,\dots,c_M)$ and type $P_{\sf X}$.
Define
\begin{align}
\mathcal{A}_n:=\bigcup_{m=1}^M T_{[Q_{\sf X|U}]\delta}^n(c_m).
\end{align}
Then
\begin{align}
&\quad {{\pi}}_{\sf X}^{\otimes n}(\mathcal{A}_n\cap T_{P_{\sf X}})
\nonumber\\
&=\sum_{x^n\in \mathcal{A}_n\cap T_{P_{\sf X}}}
{{\pi}}_{\sf X}^{\otimes n}(x^n)
\\
&= \exp(-n[H(P_{\sf X})+D(P_{\sf X}\|{{\pi}}_{\sf X})])\cdot |\mathcal{A}_n\cap T_{P_{\sf X}}|
\\
&\le \exp(-n[H(P_{\sf X})+D(P_{\sf X}\|{{\pi}}_{\sf X})])\cdot \sum_{m=1}^M|T_{[Q_{\sf X|U}]\delta}^n(c_m)\cap T_{P_{\sf X}}|
\\
&\le \exp(-n[H(P_{\sf X})+D(P_{\sf X}\|{{\pi}}_{\sf X})])
\nonumber\\
&\quad\cdot
M\exp(n[H_{[P_{\sf X}]\delta}+\delta|\mathcal{U}|\log|\mathcal{X}|])
\label{e_vi28}
\\
&=\exp(-n[D(P_{\sf X}\|{{\pi}}_{\sf X})+H(P_{\sf X})
-H_{[P_{\sf X}]\delta}-R-\delta|\mathcal{U}|\log|\mathcal{X}|]).
\label{e_vi29}
\end{align}
where \eqref{e_vi28}
is from Lemma~\ref{lem_vi7}. Whence \eqref{e_vi29} and the trivial bound
\begin{align}
{{\pi}}_{\sf X}^{\otimes n}(\mathcal{A}_n\cap T_{P_{\sf X}})
&\le
{{\pi}}_{\sf X}^{\otimes n}(T_{P_{\sf X}})
\\
&\le \exp(-nD(P_{\sf X}\|{{\pi}}_{\sf X}))
\\
&\le \exp(-n[D(P_{\sf X}\|{{\pi}}_{\sf X})-\delta|\mathcal{U}|\log|\mathcal{X}|])
\end{align}
yield the bound
\begin{align}
{{\pi}}_{\sf X}^{\otimes n}(\mathcal{A}_n\cap T_{P_{\sf X}})
\le
\exp(-n[f(\delta,P_{\sf X})-\delta|\mathcal{U}|\log|\mathcal{X}|]),\label{e_vi32}
\end{align}
where we have defined the function
\begin{align}
f(\delta,P_{\sf X})
&:=D(P_{\sf X}\|{{\pi}}_{\sf X})+[H(P_{\sf X})
-H_{[P_{\sf X}]\delta}-R]^+\label{e_f}
\end{align}
for $\delta>0$ and $P_{\sf X}\ll {{\pi}}_{\sf X}$.
Define\footnote{The reason why we can write minimum in \eqref{e_g} is explained in Remark~\ref{rem_min}.}
\begin{align}
g(\delta)&:=\min_{P_{\sf X}}f(\delta,P_{\sf X}),\label{e_g}
\end{align}
Then
\begin{align}
{{\pi}}_{\sf X}^{\otimes n}(\mathcal{A}_n)
&=\sum_{P_{\sf X}}{{\pi}}_{\sf X}^{\otimes n}(\mathcal{A}_n\cap T_{P_{\sf X}})
\\
&\le \sum_{P_{\sf X}}
\exp(-n[g(\delta)-\delta|\mathcal{U}|\log|\mathcal{X}|])
\label{e_vi36}
\\
&\le (n+1)^{|\mathcal{X}|}\exp(-n[g(\delta)-\delta|\mathcal{U}|\log|\mathcal{X}|])
\label{e_vi37}
\end{align}
where the summation is over all type $P_{\sf X}$ absolutely continuous with respect to ${{\pi}}_{\sf X}$, and \eqref{e_vi36} is from \eqref{e_vi32}. Then for any real number $G<g(\delta)-\delta|\mathcal{U}|\log|\mathcal{X}|$ we have
\begin{align}
&\quad E_{\exp(nG)}(P_{X^n[c^M]}\|{{\pi}}_{\sf X|U}^{\otimes n})
\nonumber\\
&\ge
P_{X^n[c^M]}(\mathcal{A}_n)-\exp(nG){{\pi}}_{\sf X|U}^{\otimes n}(\mathcal{A}_n)
\\
&\ge \frac{1}{M}\sum_{m=1}^M Q_{\sf X|U}^{\otimes n}(T_{[Q_{\sf X|U}]\delta}^n(c_m)|c_m)-\exp(nG){{\pi}}_{\sf X|U}^{\otimes n}(\mathcal{A}_n)
\label{e_183}
\\
&\ge 1-\gamma_n-\exp(nG){{\pi}}_{\sf X|U}^{\otimes n}(\mathcal{A}_n)\label{e_vi40}
\\
&\to 1,\quad n\to\infty.\label{e_vi41}
\end{align}
where
\begin{itemize}
\item \eqref{e_183} uses $T_{[Q_{\sf X|U}]\delta}^n(c_m)\subseteq\mathcal{A}_n$, and we used the notation of the tensor power for the conditional law $Q_{\sf X|U}^{\otimes n}(\cdot|u^n):=\prod_{i=1}^n Q_{{\sf X|U}=u_i}$.
\item \eqref{e_vi40} is from Lemma~\ref{lem_delta},
\item \eqref{e_vi41} uses \eqref{e_vi37}.
\end{itemize}
Since $\delta>0$ was arbitrary, we thus conclude
\begin{align}
G_{\epsilon}(R,{{\pi}}_{\sf X})
&\ge \sup_{\delta>0}\{g(\delta)-\delta|\mathcal{U}|\log|\mathcal{X}|\}
\\
&\ge \liminf_{\delta\to 0}g(\delta)
\\
&\ge g(0) \label{e187}
\end{align}
where \eqref{e187} is from Lemma~\ref{lem_cont}.
Since $g(0)$ is the right side of \eqref{e_vi5}, the converse bound is established.
\end{proof}
\begin{lem}\label{lem_cont}
The functions $f$ and $g$ defined in \eqref{e_f} and \eqref{e_g} are both lower semicontinuous.
\end{lem}
\begin{rem}\label{rem_min}
We can write minimum instead of infimum in \eqref{e_g} and hence \eqref{e_vi5} because of the lower semicontinuity of $f$.
\end{rem}
\begin{proof}
Consider a lower semicontinuous function $\chi$ where
$\chi(\delta,P_{\sf X}, Q_{\sf U})$ equals
\begin{align}
D(P_{\sf X}\|\pi_{\sf X})+
[H(P_{\sf X})-H(Q_{\sf X|U}|Q_{\sf U})-R]^+
\end{align}
if $|Q_{\sf X}-P_{\sf X}|\le\delta|\mathcal{U}|$ and $+\infty$ otherwise.
Then $f(\delta,P_{\sf X})=\min_{Q_{\sf X}}\chi(\delta,P_{\sf X}, Q_{\sf U})$ is lower semicontinuous, as it is the pointwise infimum of a lower semicontinuous functions over a compact set (see for example the proof in \cite[Lemma~9]{liu_topo}).
The lower semicontinuity of $g$ follows for the same reason.
\end{proof}
The function $G(R,{{\pi}}_{\sf X})$ in \eqref{e_vi5} satisfies some nice properties. Below we write it as $G(R,{{\pi}}_{\sf X},Q_{\sf X|U})$ to emphasize its dependence on $Q_{\sf X|U}$, and assume that $\mathcal{X}$ and $\mathcal{U}$ can be arbitrary.
\begin{prop}
\begin{enumerate}
  \item The function being minimized in \eqref{e_vi5}, denoted as $F(Q_{\sf U},R,{{\pi}}_{\sf X},Q_{\sf X|U})$, is convex in $Q_{\sf U}$.
  \item Additivity: for any $R>0$, ${{\pi}}_{{\sf X}_i}$ and $Q_{{\sf X}_i|{\sf U}_i}$ ($i=1,2$),
      \begin{align}
      &\quad G(R,{{\pi}}_{{\sf X}_1}{{\pi}}_{{\sf X}_2},Q_{{\sf X}_1|{\sf U}_1} Q_{{\sf X}_2|{\sf U}_2})
      \nonumber\\
      &=
    \min_{R_1,R_2:R_1+R_2\le R}\{G(R_1,{{\pi}}_{{\sf X}_1},Q_{{\sf X}_1|{\sf U}_1})
    \nonumber\\
    &\quad+G(R_2,{{\pi}}_{{\sf X}_2},Q_{{\sf X}_2|{\sf U}_2})\},
      \end{align}
      where we have abbreviated ${\pi}_{{\sf X}_1}\times{\pi}_{{\sf X}_2}$ and $Q_{{\sf X}_1|{\sf U}_1}\times Q_{{\sf X}_2|{\sf U}_2}$
      as
      ${{\pi}}_{{\sf X}_1}{{\pi}}_{{\sf X}_2}$
      and $Q_{{\sf X}_1|{\sf U}_1} Q_{{\sf X}_2|{\sf U}_2}$.
  \item $G(R,{{\pi}}_{\sf X},Q_{\sf X|U})$ is continuous in $R$.
  \item $G(R,{{\pi}}_{\sf X},Q_{\sf X|U})$ is convex in $R$.
\end{enumerate}
\end{prop}
\begin{proof}
\begin{enumerate}
  \item The function of interest is the maximum of the two functions $D(Q_{\sf X}\|{{\pi}}_{\sf X})$ and
      \begin{align}
      D(Q_{\sf X}\|{{\pi}}_{\sf X})+I(Q_{\sf U},Q_{\sf X|U})-R
      =\mathbb{E}\left[\imath_{Q_{\sf X|U}\|{{\pi}}_{\sf X}}(X|U)\right]-R
      \end{align}
      where $(U,X)\sim Q_{\sf UX}$, $Q_{\sf U}\to Q_{\sf X|U}\to Q_{\sf X}$, and the conditional relative information
      \begin{align}
      \imath_{Q_{\sf X|U}\|\pi_{\sf X}}(x|u):=\log
      \frac{{\rm d}Q_{{\sf X|U}=u}}{{\rm d}\pi_{\sf X}}(x),\quad \forall u,x.
      \end{align}
      The former is convex and the latter is linear in $Q_{\sf U}$.
      \item The $\le$ direction is immediate from the single-letter formula~\eqref{e_vi5} and the inequality
          \begin{align}
          [a]^++[b]^+\ge [a+b]^+
          \label{e196}
          \end{align}
          for any $a,b\in\mathbb{R}$. For the $\ge$ direction, suppose $Q_{{\sf U}_1{\sf U}_2}$ achieves the minimum in the single-letter formula of $G(R,{{\pi}}_{{\sf X}_1}{{\pi}}_{{\sf X}_2},Q_{{\sf X}_1|{\sf U}_1} Q_{{\sf X}_2|{\sf U}_2})$. Observe that
          \begin{align}
          & F(Q_{{\sf U}^2},R,{{\pi}}_{{\sf X}_1}{{\pi}}_{{\sf X}_2},Q_{{\sf X}_1|{\sf U}_1}Q_{{\sf X}_2|{\sf U}_2})
          \nonumber\\
          &\quad-F(Q_{{\sf U}_1}Q_{{\sf U}_2},R,{{\pi}}_{{\sf X}_1}{{\pi}}_{{\sf X}_2},Q_{{\sf X}_1|{\sf U}_1}Q_{{\sf X}_2|{\sf U}_2})
          \nonumber
          \\
          &=
          I({\sf X}_1;{\sf X}_2)
          +[I(Q_{{\sf U}_1{\sf U}_2},Q_{{\sf X}_1|{\sf U}_1}Q_{{\sf X}_2|{\sf U}_2})-R]^+
          \nonumber\\
          &\quad-\left[\sum_{i=1}^2I(Q_{{\sf U}_i},Q_{{\sf X}_i|{\sf U}_i})-R\right]^+
          \label{e201}
          \\
          &\ge [I({\sf X}_1;{\sf X}_2)+I(Q_{{\sf U}_1{\sf U}_2},Q_{{\sf X}_1|{\sf U}_1}Q_{{\sf X}_2|{\sf U}_2})-R]^+
          \nonumber\\
          &\quad -\left[\sum_{i=1}^2I(Q_{{\sf U}_i},Q_{{\sf X}_i|{\sf U}_i})-R\right]^+
          \label{e_194}
          \\
          &=0
          \end{align}
          where \eqref{e201} uses $D(Q_{{\sf X}_1{\sf X}_2}\|\pi_{{\sf X}_1}\pi_{{\sf X}_2})-D(Q_{{\sf X}_1}\|\pi_{{\sf X}_1})-D(Q_{{\sf X}_2}\|\pi_{{\sf X}_2})=I({\sf X}_1;{\sf X}_2)$,
          and
          \eqref{e_194}
          follows from \eqref{e196}. Therefore
          \begin{align}
          &\quad G(R,{{\pi}}_{{\sf X}_1}{{\pi}}_{{\sf X}_2},Q_{{\sf X}_1|{\sf U}_1}Q_{{\sf X}_2|{\sf U}_2})
          \nonumber\\
          &=
          F(Q_{{\sf U}_1}Q_{{\sf U}_2},R,{{\pi}}_{{\sf X}_1}{{\pi}}_{{\sf X}_2},Q_{{\sf X}_1|{\sf U}_1}Q_{{\sf X}_2|{\sf U}_2}).
          \label{e197}
          \end{align}
            But clearly there exists $R_1$ and $R_2$ summing to $R$ such that
          \begin{align}
          &\quad \textrm{R.H.S.~of~}\eqref{e197}
          \nonumber\\
          &=F(Q_{{\sf U}_1}, R_1,{{\pi}}_{{\sf X}_1},Q_{{\sf X}_1|{\sf U}_1})
          +F(Q_{{\sf U}_2}, R_2,{{\pi}}_{{\sf X}_2},Q_{{\sf X}_2|{\sf U}_2})
          \\
          &\ge F(R_1,{{\pi}}_{{\sf X}_1},Q_{{\sf X}_1|{\sf U}_1})
          +F(R_2,{{\pi}}_{{\sf X}_2},Q_{{\sf X}_2|{\sf U}_2})
          \end{align}
          and the result follows.

   \item Fix any two numbers $0\ge R'<R$. Choose $Q_{\sf U}$ such that
       \begin{align}
       G(R,{{\pi}}_{\sf X},Q_{\sf X|U})=F(Q_{\sf U},R,{{\pi}}_{\sf X},Q_{\sf X|U}).
       \end{align}
       Then
       \begin{align}
       0&\le G(R',{{\pi}}_{\sf X},Q_{\sf X|U})-
       G(R,{{\pi}}_{\sf X},Q_{\sf X|U})
       \label{e_208}
       \\
       &\le F(Q_{\sf U},R',{{\pi}}_{\sf X},Q_{\sf X|U})
       -F(Q_{\sf U},R,{{\pi}}_{\sf X},Q_{\sf X|U})
       \\
       &=[I(Q_{\sf U},Q_{\sf X|U})-R']^+
       -[I(Q_{\sf U},Q_{\sf X|U})-R]^+
       \\
       &\le [R-R']^+\label{e204}
       \end{align}
       where \eqref{e_208} follows because $G(\cdot,\pi_{\sf X},Q_{\sf X|U})$ is non-increasing,
       and
       \eqref{e204} uses \eqref{e196} again. Thus $G(R,{{\pi}}_{\sf X},Q_{\sf X|U})$ is actually $1$-Lipschitz continuous in $R$.

    \item Fix $R_1,R_2\ge0$, $\alpha\in[0,1]$, and let $Q_{{\sf U}_i}$ maximize $F(\cdot,R_i,{{\pi}}_{\sf X},Q_{\sf X|U})$ for $i=1,2$. Define
        \begin{align}
        R_{\alpha}&:=(1-\alpha)R_0+\alpha R_1;
        \\
        Q_{{\sf U}_{\alpha}}&:=(1-\alpha)Q_{{\sf U}_0}+\alpha Q_{{\sf U}_1}.
        \end{align}
        In both $I(Q_{{\sf U}_{\alpha}},Q_{\sf X|U})>R_{\alpha}$ and $I(Q_{{\sf U}_{\alpha}},Q_{\sf X|U})\le R_{\alpha}$ cases one can explicitly calculate that
        \begin{align}
        F(Q_{{\sf U}_{\alpha}},R_{\alpha},{{\pi}}_{\sf X},Q_{\sf X|U})
        &\le (1-\alpha)F(Q_{{\sf U}_0},R_0,{{\pi}}_{\sf X},Q_{\sf X|U})
        \nonumber\\
        &+\alpha F(Q_{{\sf U}_1},R_1,{{\pi}}_{\sf X},Q_{\sf X|U})
        \end{align}
        and the convexity follows.
\end{enumerate}
\end{proof}

\subsection{Converse for Worst-case  Resolvability}\label{sec_worst}
The resolvability for the worst-case input distribution and for stationary memoryless outputs have very different flavors.
First, let us remark that the resolvability for the worst distribution on the \emph{output} alphabet (i.e.~not necessarily induced by an input distribution) is usually a degenerate unexciting problem. For any DMC $Q_{\sf X|U}$ having an output symbol $x\in\mathcal{X}$ such that the one point distribution on $x$ is not induced by any input (it may still be true that the one point distribution is absolutely continuous with respect to the output distribution corresponding to some input), the probability $Q_{X^n}(x^n)$ vanishes for any sequence of input distributions $\{Q_{U^n}\}_{n=1}^{\infty}$, where $Q_{U^n}\to Q_{\sf X|U}^{\otimes n}\to Q_{X^n}$.
Thus if we pick the output distribution ${{\pi}}_{X^n}$ as the one point distribution on $x^n$, then the approximation error
$E_{\exp(nG)}(Q_{X^n[c^{M_n}]}\|{{\pi}}_{X^n})\uparrow1$ as $n\to\infty$, no matter how large $G$ and $M_n$ are and what $c^{M_n}$ we pick.

Returning to the resolvability for the worst-case output distribution as formulated in Definition~\ref{defn_worst}, we have shown the achievability bound $S(G)\le \sup_{\mb{U}}\mb{\bar{I}}(\mb{U};\mb{X})-G$ in Theorem~\ref{thm_ma}. Is this bound tight in general? Before delving into the converse,
it is instructive to consider the following example.
\begin{ex}\label{ex25}
Consider a DMC with $\mathcal{U}=\{0,{\tt e},1\}$, $\mathcal{X}=\{0,1\}$ and
\begin{align}
Q_{\sf X|U}(x|u)=\left\{
\begin{array}{cc}
  1-\delta & u\in\{0,1\},x=u \\
  \delta & u\in\{0,1\},x\neq u \\
  \frac{1}{2} & u={\tt e}
\end{array}
\right.
\end{align}
where $0<\delta<\frac{1}{2}$. Then from the formula of resolvability under total variation distance \cite{han1993approximation}, we have
\begin{align}\label{e217}
S(0)=C(Q_{\sf X|U})=1-h(\delta)
\end{align}
where $h(\cdot)$ is the binary entropy function,
while Theorem~\ref{thm_ma} yields,
\begin{align}\label{e218}
S(G)\le[1-h(\delta)-G]^+
\end{align}
for $G>0$.
\end{ex}
Example~\ref{ex25} with $\delta=0$ reduces to Example~1 in \cite{han1993approximation}. It is argued in \cite{han1993approximation} that the worst-case input distribution requiring the maximal asymptotic randomness \eqref{e217} is the equiprobable distribution on the set of all sequences having the type $(1/2,0,1/2)$ (which is the capacity-achieving single-letter distribution).
Naively and from symmetry considerations, one might expect that this is also the worst-case distribution for approximation in $E_{\exp(nG)}$ metric when $G>0$ and the channel parameter $\delta>0$.
However, this cannot be farther from the truth.
Consider a deterministic input sequence $({\tt e},\dots,{\tt e})$,
denote by $\hat{X}^n$ the corresponding output sequence, and denote by $\hat{T}$ the number of $1$'s in the output.
Moreover, when the input to the channel is equiprobable on the set of all sequences having the type $(1/2,0,1/2)$,
denote by $X^n$ the corresponding output sequence, and denote by $T$ be the number of $1$'s.
Then by the central limit theorem, $\frac{\hat{T}-n/2}{\sqrt{n}}$ and $\frac{T-n/2}{\sqrt{n}}$ both converge weakly to some Gaussian random variables $\hat{N}$ and $N$.
By carefully bounding the probability of binomial distributions using Stirling's formula, one can show that for any $\gamma>0$, $\lim_{n\to\infty}E_{\gamma}(Q_{\hat{T}}\|Q_T)
=E_{\gamma}(Q_{\hat{N}}\|Q_{N})$. Thus
\begin{align}
E_{\exp(nG)}(Q_{\hat{T}}\|Q_T)\to0,\quad n\to\infty
\end{align}
for any $G>0$ because $\exp(nG)\to\infty$. For both inputs, the output conditioned on the number of $1$'s is the equiprobable distribution on a set of sequences with the same type, so by Proposition~\ref{prop3}-\ref{pt3_4}), $E_{\exp(nG)}(Q_{\hat{X}^n}\|Q_{X^n})\to0$ as well. This seems to suggest that the value of $S(G)$ has a jump as $G$ changes from $0$ to a positive number (see \eqref{e217}). So is $S(G)$ discontinuous at $G=0$ in Example~\ref{ex25}?

In the remainder of this section, we answer this question in the negative by developing a general converse bound, implying that the worst-case distribution is \emph{not} induced by the set of all input sequences of the same type $(1/2,0,1/2)$.
The basic idea is to use achievability results for the error exponent of identification (ID) channels. The converse bound will not match the achievability bound in Theorem~\ref{thm_ma}, and the exact formula for $S(G)$ seems to be out of reach at this point even for a DMC.

As the first step of the converse, we observe how the achievability of ID coding implies a packing lemma, which is a sharpening of an argument in \cite{han1993approximation} to the large deviation analysis. Recall that an $(N,\nu,\lambda)$-ID code \cite{ahlswede1989identification} consists of distributions $(Q_{U_i^n})_{i=1}^N$ and decoding regions $(\mathcal{D}_i)_{i=1}^N$ such that the two type of errors satisfy
\begin{align}
\max_{1\le i\le N}Q_{X_i^n}(\mathcal{D}_i^c)&\le\mu
\\
\max_{i\neq j}Q_{X_j^n}(\mathcal{D}_i)&\le\lambda
\end{align}
In the asymptotic setting, the performance of ID code is quantified as follows.
\begin{defn}\label{def_id}
The triple $(R,G_1,G_2)$ is said to be \emph{achievable} if there exists $(N_n,\nu_n,\lambda_n)$-ID code such that
\begin{align}
\liminf_{n\to\infty}\frac{1}{n}\log\log N_n&\ge R,
\\
\liminf_{n\to\infty}\frac{1}{n}\log\frac{1}{\mu_n}&\ge G_1,
\\
\liminf_{n\to\infty}\frac{1}{n}\log\frac{1}{\lambda_n}&\ge G_2.
\end{align}
\end{defn}

\begin{lem}[Packing Lemma]\label{lem_pack}
If $(R,G_1,G_2)$ is achievable for ID code, then for each $n$ there exists $N_n$ distributions $(Q_{U_i^n})_{i=1}^{N_n}$ such that
\begin{align}
\liminf_{n\to\infty}\frac{1}{n}\log\log N_n&\ge R,\label{e194}
\\
\liminf_{n\to\infty}\frac{1}{n}\log\frac{1}{1-\frac{1}{2}\min_{i\neq j}|Q_{X_i^n}-Q_{X_j^n}|}&\ge\min\{G_1,G_2\}.
\end{align}
\end{lem}
\begin{proof}
Pick the $(N_n,\nu_n,\lambda_n)$-ID code as in Definition~\ref{def_id}.
Then for any $i\neq j$,
\begin{align}
\frac{1}{2}|Q_{X_i^n}-Q_{X_j^n}|&\ge Q_{X_i^n}(\mathcal{D}_i)-Q_{X_j^n}(\mathcal{D}_i)
\\
&\ge 1-\mu_n-\lambda_n
\end{align}
and the result follows.
\end{proof}
For a fixed DMC $Q_{\sf X|U}$, define the function
\begin{align}
e(R):=\max_{Q_{\sf U}:I(Q_{\sf U},Q_{\sf X|U})\ge R}E_{\sf sp}(R,Q_{\sf U},Q_{\sf X|U})
\label{e198}
\end{align}
where
\begin{align}
E_{\sf sp}(R,Q_{\sf U},Q_{\sf X|U})
:=\min_{P_{\sf X|U}:I(Q_{\sf U},P_{\sf X|U})\le R}D(P_{\sf X|U}\|Q_{\sf X|U}|Q_{\sf U})
\end{align}
is the well known \emph{sphere packing exponent function} \cite{csiszar1981information}.
Then $e(R)$ is a nonnegative, non-increasing function of $R$ with $e(C(Q_{\sf X|U}))=0$, so that there is a unique solution $G^*(R)$ to the equation
\begin{align}\label{e199}
e(R+2G)=G
\end{align}
provided that $R<C(Q_{\sf X|U})$.

We also need an achievability result of ID coding error exponents. The exponent is known for the first type of error \cite{ahlswede1989identification}, but not for the second type of error.
\begin{lem}\label{lem38}\cite{ahlswede1989identification}
If $R,G\ge0$, and a stationary memoryless channel $Q_{\sf X|U}$ and a probability measure $Q_{\sf U}$ satisfy $I(Q_{\sf U},Q_{\sf X|U})\ge R+2G$, the triple
\begin{align}
(R,\min_{P_{\sf X|U}:I(Q_{\sf U},P_{\sf X|U})\le R+2G}D(P_{\sf X|U}\|Q_{\sf X|U}|Q_{\sf U}),G)
\end{align}
is achievable for identification coding.
\end{lem}
\begin{thm}\label{thm28}
Let $Q_{\sf X|U}$ be a stationary memoryless channel whose input alphabet $\mathcal{U}$ is finite. For any $R<C(Q_{\sf X|U})$ and $0<\epsilon<\frac{1}{2}$,
\begin{align}
S_{\epsilon}(G^*(R))&\ge R.
\label{e202}
\end{align}
\end{thm}
\begin{rem}\label{rem42}
Our proof based on identification coding only yields a ``strong converse'' in the range $\epsilon\in(0,\frac{1}{2})$, rather than a full strong converse for $\epsilon \in (0,1)$.
\end{rem}
\begin{proof}
Let $G_*:=G^*(R)$ and $Q_{\sf U}$ be a distribution achieving the maximum in the definition of $e(R+2G_*)$ (see \eqref{e198}). Then $I(Q_{\sf U},Q_{\sf X|U})\ge R+2G_*$ and
\begin{align}
G_*=e(R+2G_*)=\min_{P_{\sf X|U}:I(Q_{\sf U},P_{\sf X|U})\le R+2G_*}D(P_{\sf X|U}\|Q_{\sf X|U}|Q_{\sf U})
\end{align}
hold, so that $(R,G_*,G_*)$ is achievable for identification coding by Lemma~\ref{lem38}. Therefore by Lemma~\ref{lem_pack}, there exist $(Q_{U_i^n})_{i=1}^{N_n}$ such that \eqref{e194} holds and
\begin{align}
\liminf_{n\to\infty}\frac{1}{n}\log\frac{1}{1-\frac{1}{2}\min_{i\neq j}|Q_{X_i^n}-Q_{X_j^n}|}&\ge G_*.
\end{align}

Observe that $G_*>0$. Indeed, if $G_*=0$, we have $e(R)=0$ from \eqref{e199}; but $R<C(Q_{\sf X|U})$ implies the existence of a $Q_{\sf U}$ such that $I(Q_{\sf U},Q_{\sf X|U})> R$ for which
$\min_{P_{\sf X|U}:I(Q_{\sf U},P_{\sf X|U})\le R}D(P_{\sf X|U}\|Q_{\sf X|U}|Q_{\sf U})>0$, a contradiction.
Now fix $0<G''<G'<G_*$. By Proposition~\ref{prop16}, there exist $\left(\tilde{Q}_{U_i^n}\right)_{i=1}^{N_n}$ and $(\gamma_n)_{n=1}^{\infty}$ such that
\begin{align}
\limsup_{n\to\infty}\frac{1}{n}\log\gamma_n&\le G'';
\\
\limsup_{n\to\infty}\frac{1}{n}\log\tilde{M}_n&\le S(G'');
\end{align}
and for each $n$ and each $1\le i\le N_n$,
\begin{align}
&\textrm{$\tilde{Q}_{U_i^n}$ is $\tilde{M}_n$-type;}
\\
&E_{\gamma_n}(\tilde{Q}_{X_i^n}\|Q_{X_i^n})
\le\epsilon.
\label{e207}
\end{align}

Next, we show that distributions in $\left(\tilde{Q}_{U_i^n}\right)_{i=1}^{N_n}$ are distinctive for large $n$. Observe that
\begin{align}
&\quad\frac{\gamma_n}{2}|Q_{X_i^n}-Q_{X_j^n}|
+1-\gamma_n
\nonumber\\
&\le
E_{\gamma_n}(\tilde{Q}_{X_i^n}\|Q_{X_i^n})
+E_{\gamma_n}(\tilde{Q}_{X_i^n}\|Q_{X_j^n})
\label{e208}
\\
&\le E_{\gamma_n}(\tilde{Q}_{X_i^n}\|Q_{X_i^n})
+E_{\gamma_n}(\tilde{Q}_{X_j^n}\|Q_{X_j^n})
+\frac{1}{2}|\tilde{Q}_{X_i^n}-\tilde{Q}_{X_j^n}|
\label{e209}
\end{align}
where \eqref{e208} and \eqref{e209} are from \eqref{e20} and \eqref{e19}, respectively. However the sum of the first two terms in \eqref{e209} is bounded by $2\epsilon$ because of \eqref{e207}, and
\begin{align}
1-\frac{1}{2}\min_{i\neq j}|Q_{X_i^n}-Q_{X_j^n}| \le \exp(-nG')
\end{align}
for $n$ large enough. Therefore \eqref{e209} implies that for large $n$
\begin{align}
\frac{1}{2}\min_{i\neq j}|\tilde{Q}_{X_i^n}-\tilde{Q}_{X_j^n}|\ge
1-\gamma_n\exp(-nG')-2\epsilon>0
\end{align}
and so $\tilde{Q}_{X_i^n}\neq \tilde{Q}_{X_j^n}$ unless $i=j$.
But the number of distinctive $\tilde{M}_n$-type distributions is at most $\left(|\mathcal{U}|^n\right)^{\tilde{M}_n}$, which should upper-bound $N_n$ for large $n$, hence
\begin{align}
R&\le
\liminf_{n\to\infty}\frac{1}{n}\log\log N_n
\\
&\le\limsup_{n\to\infty}\frac{1}{n}
(\log \tilde{M}_n+\log\log|\mathcal{U}|+\log n)
\\
&\le S_{\epsilon}(G'').
\end{align}
Since $G''$ can be arbitrarily close to $G_*(R)$, we conclude that for any $R$ such that $G^*(R)>0$,
\begin{align}\label{e215}
R\le \lim_{g\uparrow G^*(R)}
S_{\epsilon}(g).
\end{align}
Finally we finish the proof using the monotonicity. The $R=0$ case is trivial so we assume $R>0$ below. From \eqref{e215},
\begin{align}
R=\lim_{r\uparrow R} r
\le
\lim_{r\uparrow R} \lim_{g\uparrow G^*(r)}
S_{\epsilon}(g).
\label{e_217}
\end{align}
Note that we have shown $e(\cdot)$ is positive in a neighborhood of $R$; this function must be \emph{strictly} decreasing on this interval because as argued in the proof of \cite[Corollary~2.5.4]{csiszar1981information},
the sphere packing exponent function $E_{\sf sp}(R,Q_{\sf U},Q_{\sf X|U})$ is \emph{strictly} decreasing in $R$ in any interval where it is finite and positive. Then $G^*(\cdot)$ is also \emph{strictly} decreasing on this neighborhood of $R$. Thus for any $r<R$ we have $G^*(r)>G^*(R)$ and so
\begin{align}
\lim_{g\uparrow G^*(r)}S_{\epsilon}(g)\le S_{\epsilon}(G^*(R))
\label{e218}
\end{align}
since $S_{\epsilon}(\cdot)$ is non-increasing.
Taking $r\uparrow R$ on both sides of \eqref{e218} and using \eqref{e_217} gives the desired bound \eqref{e202}.
\end{proof}
We do not expect the bound in Theorem~\ref{thm28} to be tight in general.
Indeed, in the case of identity channels, we can take the $X^n$ in Corollary~\ref{thm_s} to be equiprobable on the set of sequences whose empirical distribution is equiprobable on $\mathcal{U}$ if $n$ is a multiple of $|\mathcal{U}|$, and let $X^n$ have an arbitrary distribution otherwise.
Then Corollary~\ref{thm_s} gives
\begin{align}
S(G)=[\log|\mathcal{U}|-G]^+,
\quad\forall G\in (0,\log|\mathcal{U}|)
\label{e249}
\end{align}
so that the achievability bound in Theorem~\ref{thm_ma} is tight.
However, the converse bound in Theorem~\ref{thm28} is not tight.
To see this, consider an $R\in [0,\log|\mathcal{U}|)$ and let $G:=G^*(R)>0$. Note that since $e(\cdot)$ defined in \eqref{e198} is non-increasing,
\begin{align}
e(R+2G)=G\ge e(\log|\mathcal{U}|)
\end{align}
implies that
\begin{align}
R\le \log|\mathcal{U}|-2G.
\label{e251}
\end{align}
Therefore the bound $S_{\epsilon}(G)\ge R$ given by Theorem~\ref{thm28} is not tight,
in view of \eqref{e249} and \eqref{e251}.

Theorem~\ref{thm28} indicates that $S(G)$ is continuous at $G=0$ for any DMC. Moreover, the following observation, which is immediate from the operational definition of identification coding, is sometimes useful for computations:
\begin{prop}\label{prop_subset}
Fix a DMC $Q_{\sf X|U}$ and let $Q_{\sf X|U'}$ be the DMC obtained by restricting the input to a subset $\mathcal{U}'\subseteq \mathcal{U}$. Then the outer bound of Theorem~\ref{thm28} for $Q_{\sf X|U}$ is contained in the outer bound for $Q_{\sf X|U'}$.
\end{prop}
Now pick $\mathcal{U}'=\{0,1\}$ in Example~\ref{ex25}, so that $Q_{\sf X|U'}$ is a binary symmetric channel with crossover probability $\delta$. We have:
\begin{prop}\label{cor30}
The channel in Example~\ref{ex25} satisfies
\begin{align}
S(d(a^*\|\delta))\ge R
\label{e225}
\end{align}
for any $0<R<1-h(\delta)$, where $d(\cdot\|\cdot)$ is the binary divergence function and $a^*$ is the solution to the following equation in the range $a\in[\delta,\frac{1}{2}]$.
\begin{align}
R=1-h(a)
-d(a\|\delta).
\label{e226}
\end{align}
\end{prop}
The bound in Proposition~\ref{cor30} and the achievability bound Theorem~\ref{thm_ma} are illustrated in Figure~\ref{fig_exp}.
\begin{figure}
  \centering
  \includegraphics[width=3.5in]{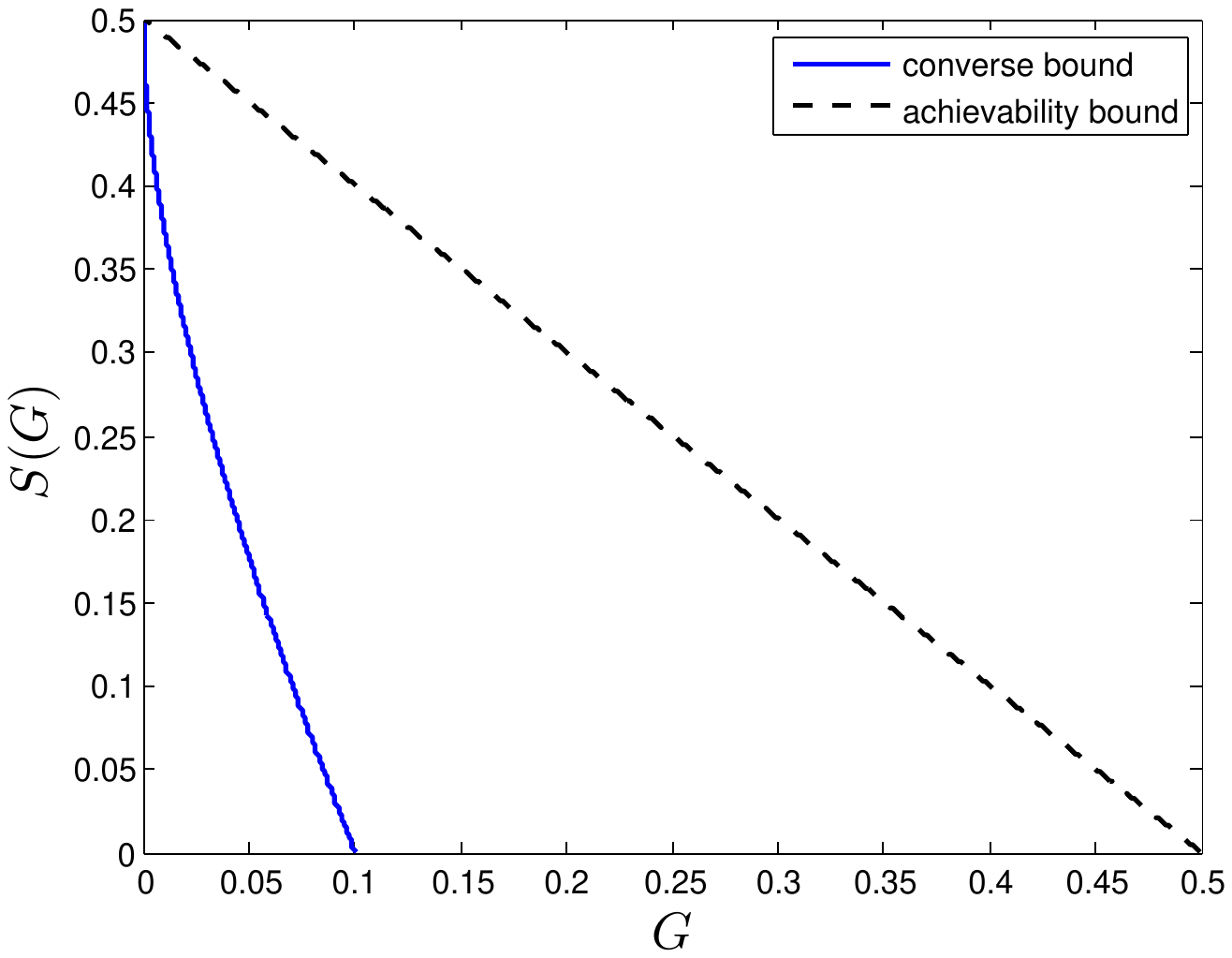}\\
  \caption{Achievability bound in Theorem~\ref{thm_ma} and converse bound in Proposition~\ref{cor30} for the channel in Example~\ref{ex25} with $\delta=0.11$.}\label{fig_exp}
\end{figure}
\begin{proof}
Clearly
\begin{align}
e(R)&\ge E_{\sf sp}(R,Q^*_{\sf U'},Q_{\sf X|U'})
\\
&=\min_{a:1-h(a)\le R}d(a\|\delta)
\end{align}
where $Q^*_{\sf U'}$ is the equiprobable distribution. Thus the $G$ satisfying
\begin{align}
\min_{a:1-h(a)\le R+2G}d(a\|\delta)=G,\label{e229}
\end{align}
denoted as $G_0$, satisfies $G_0<G^*(R)$. Since $R<1-h(\delta)$, it is clear that $G_0>0$ and the minimum in \eqref{e229} is achieved when
\begin{align}
1-h(a)=R+2G.\label{e230}
\end{align}
Substituting $G=d(a\|\delta)$ into \eqref{e230} we obtain the equation (for $a$) in \eqref{e226}. Then $G_0=d(a^*\|\delta)$, Theorem~\ref{thm28} and Proposition~\ref{prop_subset} imply \eqref{e225}.
\end{proof}

\section{Resolvability under Other Distance Measures}\label{sec_univ}
So far we have seen the tradeoffs between $G$ and $R$ for approximating either a fixed or the worst-case output distribution in $E_{\gamma}$. In this section we argue that most of these tradeoffs are insensitive to the distance metric.

\subsection{Excess Relative Information}
First, observe that the bounds relating $E_{\gamma}$ and the excess relative information \eqref{e47} immediately imply that our asymptotic results on $E_{\gamma}$-resolvability (Corollaries~\ref{thm_s}, \ref{thm_s1}, \ref{cor_asymp} and Theorems~\ref{thm_ma}, \ref{thm_conv}, \ref{thm28}) continue to hold if $E_{\gamma_n}(Q_{X^n[c^M]}\|
{{\pi}}_{X^n})$ in Definition~\ref{defn14} is replaced by $\bar{F}_{\gamma_n}(Q_{X^n[c^M]}\|
{{\pi}}_{X^n})$.

\subsection{Relative Entropy}
Next we upper-bound the relative entropy between the output distribution and the target distribution, which essentially relies on the inequality~\eqref{e_prop20}.
\begin{thm}
Fix ${{\pi}}_X$ and $Q_{UX}=Q_UQ_{X|U}$. Let $U^M:=(U_1,\dots,U_M)$ be i.i.d.~according to $Q_U$. Define for any $(c_1,\dots,c_M)\in \mathcal{U}^M$,
\begin{align}
Q_{X[c^M]}:=\frac{1}{M}\sum_{m=1}^M Q_{X|U=c_m}.
\end{align}
Then for any
$\tau_0\ge0$, $\alpha\in\mathbb{R}$ and $\beta>1$,\footnote{For a real-valued random variable $A$, we write $\mathbb{E}[A]^+$ as an abbreviation for $\mathbb{E}\left[[A]^+\right]$.}
\begin{align}
&\quad\mathbb{E}[D(Q_{X[U^M]}\|{{\pi}}_X)]
\nonumber\\
&\le
\tau_0+\mathbb{E}\left[\imath_{Q_{X|U}\|{{\pi}}_X}(X|U)
-\tau_0-\log \frac{M}{2}\right]^+
\nonumber\\
&\quad+\beta\mathbb{E}\left[\imath_{Q_X\|{{\pi}}_X}(X)
-\tau_0+\alpha\right]^+
+\frac{2\beta}{\beta-1}\exp(-\alpha)
\label{e250}
\end{align}
where $(U,X)\sim Q_U Q_{X|U}$.
\end{thm}
\begin{proof}
Setting $\epsilon\leftarrow \frac{\gamma}{2}$, $\gamma\leftarrow\exp(\tau)$ and $\gamma_2\leftarrow\exp\left(\tau_0-\alpha+\frac{1}{\beta}(\tau-\tau_0)\right)$
in \eqref{t3}, we obtain
\begin{align}
&\quad\mathbb{P}\left[\imath_{Q_{X[U^M]}\|{{\pi}}_X}(\hat{X})>\tau\right]
\nonumber\\
&\le \mathbb{P}\left[\imath_{Q_X\|{{\pi}}_X}(X)>\tau_0-\alpha+\frac{1}{\beta}(\tau-\tau_0)\right]\nonumber
\\
&\quad+\mathbb{P}\left[\imath_{Q_{X|U}\|{{\pi}}_X}(X|U)> \log \frac{M}{2}+\tau\right]\nonumber
\\
&\quad+2\exp\left(\tau_0-\alpha-\tau+\frac{1}{\beta}(\tau-\tau_0)\right)
\end{align}
where $\hat{X}\sim Q_{X[U^M]}$ conditioned on $U^M=c^M$.
Integrating both sides with respect to $\tau$ and using \eqref{e_prop20}, we find
\begin{align}
&\quad\mathbb{E}[D(Q_{X[U^M]}\|{{\pi}}_X)]
\nonumber\\
&\le
\int_0^{\infty}\mathbb{P}
\left[\imath_{Q_{X[U^M]}\|{{\pi}}_X}(\hat{X})>\tau\right]
{\rm d}\tau
\\
&\le \tau_0+ \int_{\tau_0}^{\infty}\mathbb{P}
\left[\imath_{Q_{X[U^M]}\|{{\pi}}_X}(\hat{X})>\tau\right]
{\rm d}\tau
\\
&\le
\textrm{R.H.S.~of \eqref{e250}}.
\end{align}
\end{proof}
Recall that a sequence of nonnegative random variables converging to zero in probability is uniformly integrable if and only if the sequence also converges to zero in expectation. This implies that given a sequence of real-valued random variables $A_n$,
\begin{align}
\mathbb{E}[A_n]^+=o(n)
\end{align}
provided that $\lim_{n\to\infty}\mathbb{P}[\frac{1}{n}A_n>\epsilon]=0$ for any $\epsilon>0$ (i.e.~the limsup in probability \cite{han1993approximation} of $\frac{1}{n}A_n$ does not exceed $0$) and that $\frac{1}{n}A_n$ is uniformly integrable. Therefore the $\mathbb{E}[\cdot]^+$ terms in \eqref{e250} can be easily analyzed in the asymptotic setting by setting $A_n$ to be translates of the relative information functions. If we change the definitions of achievable triple/pair by replacing \eqref{e65} and \eqref{e66} with
\begin{align}
\limsup_{n\to\infty}\frac{1}{n}D(Q_{X^n[c^{M_n}]}\|
{{\pi}}_{X^n})\le G
\label{e258}
\end{align}
then we see the achievability parts of Corollary~\ref{thm_s}, \ref{thm_s1} and Theorem~\ref{thm_ma}, \ref{thm_conv} continue to hold. In particular, the relative entropy counterpart of Theorem~\ref{thm_conv} implies that the bound in \cite[Theorem~12]{han1993approximation} is not asymptotically tight. The intuition behind this fact been explained in Section~\ref{sec_ach}, following Corollary~\ref{cor_asymp}.

Using the \emph{strong} converses of $E_{\gamma}$-resolvability and the upper-bound on relative entropy \eqref{e_prop_21}, we immediately obtain converses of resolvability in relative entropy. That is, if the definition of achievable triple/pair is changed by replacing \eqref{e65} and \eqref{e66} with \eqref{e258}, the the converse parts of Corollary~\ref{thm_s}, \ref{thm_s1} and Theorem~\ref{thm_conv} continue to hold.

Unfortunately, we only have a ``$\frac{1}{2}$-converse'' instead of a strong converse for the worst-case $E_{\gamma}$-resolvability (see Remark~\ref{rem42}). Therefore we don't have a nice counterpart of Theorem~\ref{thm28}: there is a loss of factor $2$ when \eqref{e_prop_21} is applied.

\subsection{Smooth R\'{e}nyi Entropy}
Most of the asymptotic resolvability results also hold for smooth R\'{e}nyi divergences automatically: suppose we change the definitions of achievable triple/pair by replacing \eqref{e65} and \eqref{e66} with
\begin{align}
\limsup_{n\to\infty}\frac{1}{n}D_{\alpha}^{+\epsilon}(Q_{X^n[c^{M_n}]}\|
{{\pi}}_{X^n})\le G
\end{align}
where $\epsilon\in(0,1)$ and $\alpha$ are fixed. Then the achievability parts of Corollary~\ref{thm_s}, \ref{thm_s1} and Theorem~\ref{thm_ma}, \ref{thm_conv} continue to hold for $\alpha\in[0,1)$, which is immediate from the bound \eqref{e_prop23} and the achievability results for $E_{\gamma}$. The converse parts of Corollary~\ref{thm_s}, \ref{thm_s1} and Theorem~\ref{thm_conv}, \ref{thm28} also continue to hold for $\alpha\in[0,1)$, because by Proposition~\ref{prop_1}-\ref{pt_mono}) we only need to consider $\alpha=0$, then \eqref{e_53} and the $(1-\epsilon)$-converses for $E_{\gamma}$ imply the desired result.

If \eqref{e65} and \eqref{e66} in the definitions of achievable triple/pair are changed to
\begin{align}
\limsup_{n\to\infty}\frac{1}{n}D_{\alpha}^{-\epsilon}(Q_{X^n[c^{M_n}]}\|
{{\pi}}_{X^n})\le G
\end{align}
where $\epsilon\in(0,1)$ and $\alpha$ are fixed, then the achievability parts of Corollary~\ref{thm_s}, \ref{thm_s1} and Theorem~\ref{thm_ma}, \ref{thm_conv} can still be established for $\alpha\in(1,\infty]$ because of the monotonicity (Proposition~\ref{prop_1}-\ref{pt_mono})) and the fact that $E_{\gamma}$ directly corresponds to $D^{-\epsilon}_{\infty}$ (Proposition~\ref{prop_1}-\ref{pt_13_6})).
The converse parts of Corollary~\ref{thm_s}, \ref{thm_s1} and Theorem~\ref{thm_conv} also continue to hold for $\alpha\in(1,\infty]$, in view of  \eqref{e_prop23_1} and the $\epsilon$-converses for $E_{\gamma}$; Theorem~\ref{thm28} continues to hold for $\epsilon\in (0,\frac{1}{2})$, by the ``$\frac{1}{2}$-converse'' (Remark~\ref{rem42}) for $E_{\gamma}$.

\section{Application to Lossy Source Coding}\label{seclikelihood}
The simplest application of the new resolvability result (in particular, the softer-covering lemma) is to derive a one-shot achievability bound for lossy source coding,
which is most fitting in the regime of low rate and exponentially decreasing success probability. The method is applicable to general sources. In the special case of i.i.d.~sources, it recovers the ``success exponent'' in lossy source coding originally derived by the method of types \cite{csiszar1981information} for discrete memoryless sources.
The achievability bound in \cite{song} can be viewed as the $\gamma=1$ special case, which is capable of recovering the rate-distortion function, but cannot recover the exact rate-distortion-exponent tradeoff.
\begin{thm}\label{thm_source}
Consider a source with distribution $\pi_X$ and a distortion function $d(\cdot,\cdot)$ on $\mathcal{U}\times\mathcal{X}$. For any joint distribution $Q_UQ_{X|U}$, $\gamma\ge1$, $d>0$ and integer $M$, there exists a random transformation $\pi_{U|X}$ (stochastic encoder) whose output takes at most $M$ values, and
\begin{align}
\mathbb{P}[d(\bar{U},\bar{X})\le d]\ge\frac{1}{\gamma}\left(\mathbb{P}[d(U,X)\le d]-
\mathbb{E}[E_{\gamma}(Q_{X[U^M]}\|\pi_X)]\right)\label{e21}
\end{align}
where $(\bar{U},\bar{X})\sim \pi_{U|X}\pi_X$, $(U,X)\sim Q_{UX}$, and
$U^M\sim Q_U^{\otimes M}$.
\end{thm}
\begin{proof}
Given a codebook $(c_1,\dots,c_M)\in\mathcal{U}$,
let $P_U$ be the equiprobable distribution on $(c_1,\dots,c_M)$ and set
\begin{align}
P_{UX}:=Q_{X|U}P_U.
\end{align}
The \emph{likelihood encoder} is then defined as a random transformation
\begin{align}
{{\pi}}_{U|X}:=P_{U|X}
\end{align}
so that the joint distribution of the codeword selected and the source realization $X$ is
\begin{align}
{{\pi}}_{UX}={{\pi}}_XP_{U|X}
\end{align}
From Proposition~\ref{prop3}-\ref{prop3_1}) and Proposition~\ref{prop3}-\ref{prop3_3}) we obtain
\begin{align}
\gamma\pi_{UX} (d(\cdot,\cdot)\le d)
&\ge P_{UX} (d(\cdot,\cdot)\le d)
-E_{\gamma}(P_{XU}||{{\pi}}_{XU})\label{e280}
\\
&= P_{UX} (d(\cdot,\cdot)\le d)
-E_{\gamma}(P_X||{{\pi}}_X)
\label{e272}
\end{align}
where $(\hat{U},\hat{X})\sim P_{UX}$.
Note that $P_{UX}$ and $\pi_{U|X}$ depend on the codebook $c^M$.
Now consider a random codebook $c^M\leftarrow U^M$.
Taking the expectation on both sides of \eqref{e272} with respect to $U^M$, we have
\begin{align}
&\quad\gamma\mathbb{E}[\pi_{UX} (d(\cdot,\cdot)\le d)]
\nonumber\\
&\ge
\mathbb{E}[P_{UX} (d(\cdot,\cdot)\le d)]-\mathbb{E}
[E_{\gamma}(Q_{X[U^M]}\|\pi_X)]
\\
&= \mathbb{P}[d(U,X)\le d]
-\mathbb{E}
[E_{\gamma}(Q_{X[U^M]}\|\pi_X)]\label{e_2s}
\end{align}
where in \eqref{e_2s} we used the fact that $\mathbb{E}[ P_{UX}]=Q_{UX}$.
Finally we can choose one codebook (corresponding to one $\pi_{U|X}$) such that $\pi_{UX} (d(\cdot,\cdot)\le d)$ is at least its expectation.
\end{proof}

\begin{rem}
In the i.i.d.~setting,
let $R({{\pi}}_{\sf X},d)$ be the rate-distortion function when the source has per-letter distribution ${{\pi}}_{\sf X}$. The distortion function for the block is derived from the per-letter distortion by
\begin{align}
d^n(u^n,x^n):=\frac{1}{n}\sum_{i=1}^n d(u_i,x_i).
\end{align}
Let $(\bar{\sf X}^n,\bar{\sf U}^n)$ be the source-reconstruction pair distributed according to ${{\pi}}_{{\sf X}^n{\sf U}^n}$.
If $0\le R<R({{\pi}}_{\sf X},d)$, the maximal probability that the distortion does not exceed $d$ converges to zero with the exponent
\begin{align}
\lim_{n\to\infty}\frac{1}{n}\log\frac{1}{\mathbb{P}[d^n(\bar{\sf U}^n,\bar{\sf X}^n)\le d]}=G(R,d)
\end{align}
where
\begin{align}\label{e_sexp}
G(R,d):=\min_Q[D(Q||P)+[R(Q,d)-R]^+].
\end{align}
A weaker achievability result than \eqref{e_sexp} was proved in \cite[p168]{omura1975lower}, whereas the final form \eqref{e_sexp} is given in \cite[p158, Ex6]{csiszar1981information} based on method of types. Here we can easily prove the achievability part of \eqref{e_sexp} using Theorem~\ref{thm_source} and Corollary~\ref{cor_asymp} by setting $Q_{\sf X}$ to be the minimizer of \eqref{e_sexp} and $Q_{\sf U|X}$ to be such that
\begin{align}
\mathbb{E}[d({\sf U,X})]&\le d,
\\
I(Q_{\sf U},Q_{\sf X|U})&\le R.
\end{align}
Then $\gamma_n=\exp(nE)$ with
\begin{align}\label{esuccess}
E>D(Q_{\sf X}||{\pi}_{\sf X})+[I(Q_{\sf U},Q_{\sf X|U})-R]^+,
\end{align}
ensures that
\begin{align}
\mathbb{P}[d^n(\bar{\sf U}^n,\bar{\sf X}^n)\le d]\ge\frac{1}{2}\exp(-nE)
\end{align}
for $n$ large enough, by the law of large numbers.
\end{rem}
\begin{rem}
Since the $E_{\gamma}$ metric reduces to total variation distance when $\gamma=1$, Theorem~\ref{thm_source} generalizes the likelihood source encoder based on the standard soft-covering/resolvability lemma \cite{song}. In \cite{song}, the error exponent for the likelihood source encoder at rates \emph{above} the rate-distortion function is analyzed using the exponential decay of total variation distance in the approximation of output statistics, and the exponent does not match the optimal exponent found in \cite{csiszar1981information}. It is also possible to upper-bound the success exponent of the total variation distance-based likelihood encoder at rates \emph{below} the rate-distortion function by analyzing the exponential convergence to $2$ of total variation distance in the approximation of output statistics; however that does not yield the optimal exponent \eqref{e_sexp} either.
This application illustrates one of the nice features of the $E_{\gamma}$-resolvability
method: it converts a large deviation analysis into a law of large numbers analysis, that is, we only care about whether $E_{\gamma}$ converges to $0$, but not the speed, even when dealing with error exponent problems.
\end{rem}

\section{Application to One-shot Mutual Covering Lemmas}\label{sec_mutual}
Another application of the softer-covering lemma is a one-shot generalization of the \emph{mutual covering lemma} in network information theory \cite{el1981proof}. The asymptotic mutual covering lemma says, fixing a (per-letter) joint distribution $P_{\sf UV}$, if enough ${\sf U}^n$-sequences and ${\sf V}^n$-sequences are independently generated according to $P_{{\sf U}^n}$ and $P_{{\sf V}^n}$ respectively, then with high probability we will be able to find one pair jointly typical with respect to $P_{\sf UV}$. In the one-shot version, the ``typical set'' is replaced with an arbitrarily high probability (under the given joint distribution) set.

The one-shot mutual covering lemma can be used to prove a one-shot version of Marton's inner bound for the broadcast channel with a common message\footnote{More precisely, we are referring to the three auxiliary random variables version due to Liang and Kramer \cite[Theorem~5]{liang2007rate} (see also \cite[Theorem~8.4]{el2011network}), which is equivalent to an inner bound obtained by Gelfand and Pinsker \cite{gel1980capacity} upon optimization (see \cite{liang2011equivalence} or \cite[Remark~8.6]{el2011network}).}
without time-sharing, filling a gap in the proof in \cite{verdu2012non} based on the basic covering lemma where time-sharing is necessary. More discussions about the background and the derivation of the one-shot Marton's inner bound can be found in our conference paper \cite{jingbo2015marton}.
For general discussions on single-shot covering lemmas, see \cite{verdu2012non}\cite{verdu2015}.
To avoid using time-sharing in the single-shot setting, \cite{yassaee2013technique} pursued a different approach to derive a single-shot Marton's inner bound. Moreover, a version of one-shot mutual covering lemma can be distilled from their approach \cite{yassaee2015marton}. We compare their approach and ours at the end of the section.

We proceed to provide a simple derivation of a mutual covering using the softer-covering lemma.
\begin{lem}\label{lem7}
Fix $P_{UV}$ and let
\begin{align}
P_{U^MV^L}:=\underbrace{P_U\times \dots\times P_U}_{M}\times \underbrace{P_V\times\dots\times P_V}_{L}.
\end{align}
Then
\begin{align}
&\quad
\mathbb{P}\left[\bigcap_{m=1,l=1}^{M,L}\{(U_m,V_l)\notin \mathcal{F}\}\right]
\nonumber
\\
&
\le
\mathbb{P}[(U,V)\notin\mathcal{F}]
+\mathbb{P}[\imath_{U;V}(U;V)\ge\log ML-\tau]
\nonumber
\\
&\quad+\frac{\exp(\tau)}{\max\{M,L\}}+e^{-\frac{1}{2}\exp(\tau)}.
\label{e_39}
\end{align}
for all $\tau>0$ and event $\mathcal{F}$.
\end{lem}

\begin{proof}
Assume without loss of generality that $L\ge M$. For any $u\in\mathcal{U}$, define
\begin{align}
\mathcal{F}_{u}&:=\{v:(u,v)\in\mathcal{F}\},
\end{align}
and for any $u^M\in\mathcal{U}^M$, define
\begin{align}
\mathcal{A}_{u^M}&:=\bigcup_{m=1}^M \mathcal{F}_{u_m}.
\end{align}
Now fix a $U$-codebook $c^M$ and observe that
\begin{align}
\gamma P_V(\mathcal{A}_{c^M})
&\ge P_{V[c^M]}(\mathcal{A}_{c^M})-E_{\gamma}(P_{V[c^M]}\|P_V)
\label{e292}
\\
&\ge
\frac{1}{M}\sum_{m=1}^MP_{V|U=c_m}(\mathcal{F}_{c_m})-E_{\gamma}(P_{V[c^M]}\|P_V)
\label{e293}
\end{align}
where we recall that
\begin{align}
P_{V[c^M]}:=\frac{1}{M}\sum_{m=1}^M P_{V|U=c_m}.
\end{align}
\eqref{e292} is from the definition of $E_{\gamma}$
and \eqref{e293} is because $\mathcal{F}_{c_m}\subseteq  \bigcup_{m=1}^M\mathcal{F}_{c_m} =\mathcal{A}_{c^M}$.
Denote by $\Gamma(c_1,\dots,c_M)$ the right side of \eqref{e293}, which is trivially
upper-bounded by $1$.
Next, we show that
\begin{align}
\mathbb{P}\left[\left.\bigcap_{m=1,l=1}^{M,L}\{(U_m,V_l)\notin\mathcal{F}\}
\right|U^M=c^M
\right]
&\le 1-\Gamma(c^M)+e^{-\frac{L}{\gamma}}\label{e287}
\end{align}
which is trivial when $\Gamma(c^M)<0$. In the case of
$\Gamma(c^M)\in[0,1]$,
\begin{align}
\mathbb{P}\left[\left.\bigcap_{m=1,l=1}^{M,L}\{(U_m,V_l)\notin\mathcal{F}\}
\right|U^M=c^M
\right]
&=\left[1-P_V(\mathcal{A}_{c^M})\right]^L
\label{e_293}\\
&\le\left[1-\frac{\Gamma(c^M)\frac{L}{\gamma}}{L}\right]^L
\label{e_294}
\\
&\le 1-\Gamma(c^M)+e^{-\frac{L}{\gamma}}\label{e_295}
\end{align}
where \eqref{e_293} is from the definition of $\mathcal{A}_{c^M}$, and \eqref{e_294} is from \eqref{e293}. The last step \eqref{e_295} uses the basic inequality
\begin{align}\label{e_15}
\left(1-\frac{p\alpha}{M}\right)^M\le 1-p+e^{-\alpha}
\end{align}
for $M,\alpha>0$ and $0\le p\le 1$, which has been useful in the proofs of the basic covering lemma (see \cite{cover2012elements}\cite{verdu2012non}\cite{verdubook}).
Integrating both sides of \eqref{e287} over $c^M$ with respect to $P_U\times\dots\times P_U$,
\begin{align}
&\mathbb{P}\left[\bigcap_{m=1,l=1}^{M,L}\{(U_m,V_l)\notin\mathcal{F}\}\right]
\nonumber\\
&\le P_{UV}(\mathcal{F}^c)+
\mathbb{E}[E_{\gamma}(P_{V[U^M]}\|P_V)]
+e^{-\frac{L}{\gamma}},\label{e_2}
\end{align}
where we have used the fact that
$\mathbb{E}[P_{V|U}(\mathcal{F}_{U_m}|U_m)]=P_{UV}(\mathcal{F})$
for each $m$.
Applying the ``softer-covering lemma'' as in Remark~\ref{rem25}, the middle term on the right hand side of \eqref{e_2} is upper-bounded by
\begin{align}
\mathbb{P}\left[\imath_{V;U}(V;U)
\ge\log\frac{M\gamma}{2}\right]+\frac{2}{\gamma}
\end{align}
and the result follows by $\gamma\leftarrow 2L\exp(-\tau)$.
\end{proof}
\begin{rem}
From the above derivation we see that for the proof of the basic covering lemma ($M=1$ case) we will need the ``softest-covering lemma'' (the case of one codeword) rather than the soft-covering lemma (case of $\gamma=1$ and $L>1$ codewords). However, it is still possible to prove the basic covering lemma using the soft-covering lemma using a different argument; see the discussion in \cite{yassaee2015marton}, which is essentially based on the idea in \cite{cuff2012distributed}.
\end{rem}
Lemma~\ref{lem51} below is a strengthened version of the one-shot-mutual covering lemma,
which improves Lemma~\ref{lem7} in terms of the error exponent.
The proof of Lemma~\ref{lem51} essentially combines the proof the achievability part of resolvability and the proof Lemma~\ref{lem7}, and the improvement results from not treating the two steps separately.
The proof is not as conceptually simple as Lemma~\ref{lem7} since the complexities are no longer buried under the softer-covering lemma.
\begin{lem}\label{lem51}
Under the same assumptions as Lemma~\ref{lem7},
\begin{align}
&\quad
\mathbb{P}\left[\bigcap_{m=1,l=1}^{M,L}\{(U_m,V_l)\notin \mathcal{F}\}\right]
\nonumber
\\
&\le
\mathbb{P}\left[(U,V)\notin\mathcal{F}
\textrm{ or }
\exp(\imath_{U;V}(U;V))>ML\exp(-\gamma)-\delta\right]
\nonumber
\\
&\quad
+\frac{\min\{M,L\}-1}{\delta}
+e^{-\exp(\gamma)}.\label{e301}
\end{align}
for all $\delta,\gamma>0$ and event $\mathcal{F}$.
\end{lem}
\begin{proof}
See Lemma~1 and Remark~4 in the conference version \cite{jingbo2014}.
\end{proof}
\begin{rem}
An advantage of Lemma~\ref{lem51} over Lemma~\ref{lem7} is that the upper-bound in the former contains a probability of a union of two events, rather than the sum of the probability of the two events. This yields a strict improvement in the second order rate analysis. Moreover, by setting $\delta\downarrow0$ and $M=1$ we \emph{exactly} recover the basic one-shot covering lemma in \cite{verdu2012non}.
\end{rem}
In terms of the second order rates, the one-shot Marton's inner bound for broadcast obtained from our one-shot mutual covering lemma (\cite[Theorem~10]{jingbo2015marton}) is equivalent to the achievability bound claimed in \cite[Theorem~4]{yassaee2013arxiv} based on the stochastic likelihood encoder. However, although it is not demonstrated explicitly in \cite[Theorem~10]{jingbo2015marton}, we can improve the analysis of \cite[Theorem~10]{jingbo2015marton} by using various nuisance parameters rather than a single $\gamma$, to obtain a one-shot Marton's bound which gives strictly better error exponents than \cite[Theorem~4]{yassaee2013arxiv}.
The reason for such improvement is that the third term in \eqref{e301} is doubly exponential and the second term converges to zero with a large exponent. On the other hand, the approach of \cite[Theorem~4]{yassaee2013arxiv} has the advantage of being easily extendable to the case of more than two users (which would correspond to a multivariate mutual covering lemma).

\section{Application to Wiretap Channels}\label{sec_wiretap}
%

Our final application of the $E_{\gamma}$-resolvability (in particular, the softer-covering lemma) is in the wiretap channel, whose setup is as depicted in Figure~\ref{fwiretap}.
The receiver and the eavesdropper observe $y\in\mathcal{Y}$ and $z\in\mathcal{Z}$, respectively.
Given a codebook $c^{ML}$, the input to $P_{YZ|X}$ is $c_{wl}$ where $w\in\{1,\dots,M\}$ is the message to be sent,
and $l$ is equiprobably chosen from $\{1,\dots,L\}$ to randomize the eavesdropper's observation.
We call such a $c^{ML}$ an \emph{$(M,L)$-code}.
Moreover, the eavesdropper's observation has the distribution ${{\pi}}_Z$ when no message is sent.
In this setup, we don't need to assume a prior distribution on the message/non-message.
We wish to design the codebook such that the receiver can decode the message (reliability) whereas the eavesdropper cannot detect whether a message is sent nor guess which message is sent (security).
For general wiretap channels the performance may be enhanced by appending a conditioning channel $Q_{X|U}$ at the input of the original channel \cite{hayashi2006general}.
In that case the same analysis can be carried out for the new wiretap channel $Q_{YZ|U}$. Thus the model in Figure~\ref{fwiretap} entails no loss of generality.
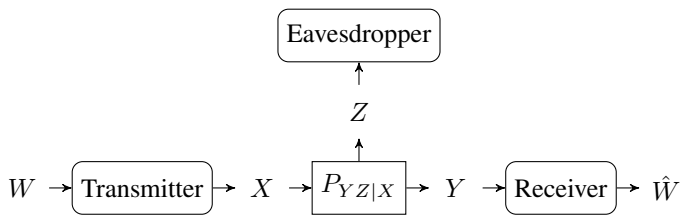
\begin{figure}[h!]
  \centering
\begin{tikzpicture}
[node distance=0.3cm,minimum height=7mm,minimum width=7mm,arw/.style={->,>=stealth'}]
  \node[rectangle,draw] (T) {$P_{YZ|X}$};
  \node[rectangle] (X) [left =of T] {$X$};
  \node[rectangle,draw,rounded corners] (A) [left =of X] {Transmitter};
  \node[rectangle] (Y) [right =of T] {$Y$};
  \node[rectangle,draw,rounded corners] (B) [right =of Y] {Receiver};
  \node[rectangle] (Z) [above =of T] {$Z$};
  \node[rectangle,draw,rounded corners] (E) [above =of Z] {Eavesdropper};
  \node[rectangle] (W) [left =of A] {$W$};
  \node[rectangle] (Wh) [right =of B] {$\hat{W}$};

  \draw [arw] (W) to node[midway,above]{} (A);
  \draw [arw] (A) to node[midway,above]{} (X);
  \draw [arw] (X) to node[midway,above]{} (T);
  \draw [arw] (T) to node[midway,above]{} (Y);
  \draw [arw] (Y) to node[midway,above]{} (B);
  \draw [arw] (B) to node[midway,above]{} (Wh);
  \draw [arw] (T) to node[midway,above]{} (Z);
  \draw [arw] (Z) to node[midway,above]{} (E);
\end{tikzpicture}
\caption{The wiretap channel}
\label{fwiretap}
\end{figure}

In Wyner's setup (see for example \cite{el2011network}),  secrecy is measured in terms of the conditional entropy of the message given the eavesdropper observation. In contrast, we measure secrecy in terms of the size of the list that the eavesdropper has to declare for the message to be included with high probability. Practically, the message $W$ is the compressed version of the plaintext. Assuming that the attacker knows which compression algorithm is used, the plaintext can be recovered by running each of the items in the eavesdropper output list through the decompressor and selecting
the one that is intelligible.

We need the following definitions to quantify the eavesdropper ability to detect/decode messages.
\begin{defn}\label{defn_achieve}
For a fixed codebook and channel $P_{Z|X}$,
we say the eavesdropper can perform $(A,T,{\epsilon})$-decoding if upon observing $Z$, it outputs a list which is either empty or of size $T$,
such that
\begin{itemize}
\item $\mathbb{P}[\textrm{list}\neq \emptyset|\textrm{no message}]\le A^{-1}$.
\item There exists $\epsilon_m\in[0,1]$, $m=1,\dots, M$ satisfying
$
{\epsilon}=\frac{1}{M}\sum_{m=1}^M\epsilon_m
$
such that
\begin{align}
\mathbb{P}[m\notin\textrm{list}|W=m]\le \epsilon_m.
\end{align}
\end{itemize}
\end{defn}
Although the decoder in Definition~\ref{defn_achieve} is reminiscent of erasure and list decoding \cite{forney1968exponential},
for the former it is possible that actually no message is sent, and we treat the undetected and detected errors together.

The logarithm of $T$ can be intuitively understood as the equivocation $H(W|Z)$ \cite{wyner1975wire}.
However, $\log T$ can be much smaller than $H(W|Z)$:
a distribution can have $99\%$ of its mass supported on a very small set,
and yet have an arbitrarily large entropy.

The quantity $A>0$ characterizes how well the eavesdropper can detect that no message is sent, which is related the notion of stealth communication \cite{hou2014effective} \cite{bloch2015}.
The ``non-stealth'' is measured by $D(P_Z\|\pi_Z)$ in \cite{hou2014effective}, and is measured by $|P_Z-\pi_Z|$ in \cite{bloch2015}.
Although both the relative entropy and total variation are related to error probability hypothesis testing, their results cannot be directly compared with ours, since they are interested in the regime where non-stealth vanishes while the transmission rate is below the secrecy capacity.
In contrast, we are mainly interested in the regime where $A$ grows exponentially (so that the ``non-stealth'' in their definition grows) in the blocklength, but the transmission rate is \emph{above} the secrecy capacity.

The asymptotic version of the eavesdropper achievability is as follows.
\begin{defn}\label{def_ach}
Fix a sequence of codebooks and a eavesdropper channel $(P_{Z^n|X^n})_{n=1}^{\infty}$. The rate pair $(\alpha,\tau)$ is ${\epsilon}$-achievable by the eavesdropper if there exist sequences $(A_n)$ and $(T_n)$ with
\begin{align}
\liminf_{n\to\infty}\frac{1}{n}\log A_n&\ge\alpha
\\
\limsup_{n\to\infty}\frac{1}{n}\log T_n&\le\tau
\end{align}
such that for sufficiently large $n$, the eavesdropper can achieve $(A_n,T_n,{\epsilon})$-decoding.
\end{defn}
By the diagonalization argument \cite[P56]{han2003information}, the set of $\epsilon$-achievable $(\alpha,\tau)$ is closed.

An \emph{$(M,L,Q_X)$-random code} is defined as the ensemble of the codebook $c^{ML}$ where each codeword $c_{wl}$ is i.i.d.~chosen according to $Q_X$, $w\in\{1,\dots,M\}$, $m\in\{1,\dots,M\}$.
We shall focus on random codes, for which reliability is guaranteed by channel coding theorems, so we only need to consider the security condition.

First, we extend the notions of achievability to the case of a random ensemble of codes by taking the average:
we say for a random ensemble of codes the eavesdropper can perform $(A,T,\epsilon)$-decoding, if there exists $\epsilon(c^{ML})$ such that for each $c^{ML}$ the eavesdropper can perform $(A,T,\epsilon(c^{ML}))$-decoding (in the sense of Definition~\ref{defn_achieve}), and the average of $\epsilon(c^{ML})$ with respect to the codebook distribution is upper-bounded by $\epsilon$.
Similarly, Definition~\ref{def_ach} can be extended to random codes.
Then, the following is our main result which characterizes the set of eavesdropper achievable pairs for stationary memoryless channels.
\begin{thm}\label{region}
Fix any $Q_{\sf X}$, $R$, $R_{\sf L}$ and $0<{\epsilon}<1$. Consider $(\exp(nR),\exp(nR_{\sf L}),Q_{\sf X}^{\otimes n})$-random codes and stationary memoryless channel with per-letter conditional distribution $Q_{\sf Z|X}$.
Then the pair $(\alpha,\tau)$ is ${\epsilon}$-achievable by the eavesdropper if and only if
\begin{align}
\left\{
\begin{array}{l}
  \alpha\le D(Q_{\sf Z}\|{{\pi}}_{\sf Z})+[I(Q_{\sf X},P_{\sf Z|X})-R-R_{\sf L}]^+; \\
  \tau \ge R-[I(Q_{\sf X},P_{\sf Z|X})-R_{\sf L}]^+.
  \end{array}
\right.
\end{align}
where $Q_{\sf X}\to Q_{\sf Z|X}\to Q_{\sf Z}$.
\end{thm}
From the noisy channel coding theorem, the supremum randomization rate $R_{\sf L}$ such that the sender can reliably transmit messages at the rate $R$ is $I(Q_{\sf X},P_{\sf Y|X})-R$. The larger $R_{\sf L}$ the less reliably the eavesdropper can decode, so the optimal encoder chooses $R_{\sf L}$ as close to this supremum as possible. Thus Theorem~\ref{region} implies the the following result
\begin{thm}
Given a stationary memoryless wiretap channel with per-letter conditional distribution $P_{\sf YZ|X}$, there exists a sequence of codebooks
such that
messages at the rate $R$ can be reliably transmitted to the intended receiver
and that $(\alpha,\tau)$ is not $\epsilon$-achievable, for any $\epsilon\in (0,1)$, by
the eavesdropper if there exists some $Q_{\sf X}$ such that
\begin{align}
R<I(Q_{\sf X},P_{\sf Y|X})
\end{align}
and either
\begin{align}\label{e_reg1}
\alpha>D(Q_{\sf Z}\|{{\pi}}_{\sf Z})+[I(Q_{\sf X},P_{\sf Z|X})-I(Q_{\sf X},P_{\sf Y|X})]^+
\end{align}
or
\begin{align}\label{e_reg2}
\tau<R-[I(Q_{\sf X},P_{\sf Z|X})-I(Q_{\sf X},P_{\sf Y|X})+R]^+.
\end{align}
\end{thm}

\begin{rem}
In general the sender-receiver want to minimize $\alpha$ and maximize $\tau$ obeying the tradeoff \eqref{e_reg1}, \eqref{e_reg2} by selecting $Q_{\sf X}$.
In the special case where $\alpha$ has no importance and $R$ is larger than the secrecy capacity $C:=\sup_{Q_{\sf X}}\{I(Q_{\sf X},P_{\sf Y|X})-I(Q_{\sf X},P_{\sf Z|X})\}$, we see from \eqref{e_reg2} that the supremum $\tau$ is $C$. The formula for the supremum of $\tau$ is the same as the equivocation measure defined as $\frac{1}{n}H(W|{\sf Z}^n)$ \cite{wyner1975wire}, but technically our result does not follow directly from the lower bound on equivocation, since it may be possible that the a posterior distribution of $W$ is concentrated on a small list but has a tail spread over an exponentially large set, resulting a large equivocation.
\end{rem}
The next two subsections prove the ``only if'' and ``if'' parts of Theorem~\ref{region}, respectively.

\subsection{Converse for the Eavesdropper}
The ``only if'' part (the eavesdropper converse) of Theorem~\ref{region} follows by applying the following non-asymptotic bounds to different regions and invoking Corollary~\ref{cor_asymp}.
\begin{thm}\label{thmlem}
In the wiretap channel, fix an arbitrary distribution $\mu_Z$.
Suppose the eavesdropper can either detect that no message is sent upon observing $z\in\mathcal{D}_0$ with
\begin{align}
\mu_Z(\mathcal{D}_0)\ge 1-A^{-1}
\end{align}
for some $A\in [1,\infty)$,
or outputs a list of $T(z)$ messages upon observing $z\notin \mathcal{D}_0$ that contains the actual message $m\in\{1,\dots,M\}$ with probability at least $1-\epsilon_m$, for some $\epsilon_m\in [0,1]$. Define the average quantities
\begin{align}
T&:=\frac{1}{\mu_Z(\mathcal{D}_0^c)}\int_{\mathcal{D}_0^c} T(z){\rm d}\mu_Z(z),
\\
{\epsilon}&:=\frac{1}{M}\sum_{m=1}^M \epsilon_m.
\end{align}
Then for any $\gamma\in [1,+\infty)$,
\begin{align}
\frac{1}{A}\ge\frac{1}{\gamma}\left(1-{\epsilon}-E_{\gamma}(P_Z\|{{\pi}}_Z)\right),
\label{e_45}
\end{align}
where we recall that ${{\pi}}_Z$ is the non-message distribution,
$P_Z:=\frac{1}{M}\sum_{m=1}^M P_{Z|W=m}$,
and $P_{Z|W=m}$ is the distribution of the eavesdropper observation for the message $m$
(assuming an arbitrary codebook is used).
Moreover,
\begin{align}
\frac{T}{MA}\ge\frac{1}{\gamma}\left(1-{\epsilon}-\frac{1}{M}\sum_{m=1}^M E_{\gamma}(P_{Z|W=m}\|\mu_Z)\right).
\label{e_46}
\end{align}
\end{thm}
We will choose $\mu_Z=P_Z$
when we use Theorem~\ref{thmlem} to prove Theorem~\ref{region}, although \eqref{e_46} holds for any $\mu_Z$.

From the eavesdropper viewpoint, a larger $A$ and a smaller $T$ is more desirable since it will then be able to find out that no message is sent with smaller error probability or narrow down to a smaller list when a message is sent. This observation agrees with \eqref{e_45} and \eqref{e_46}: a smaller $\gamma$ implies a higher degree of approximation, and hence higher indistinguishability of output distributions which is to the eavesdropper disadvantage.
\begin{proof}
To see \eqref{e_45},
\begin{align}
\frac{1}{A}
&\ge{{\pi}}_Z(\mathcal{D}_0^c)
\\
&\ge\frac{1}{\gamma}(P_{Z}(\mathcal{D}_0^c)-E_{\gamma}(P_{Z}\|{{\pi}}_Z))\label{e46}
\\
&=\frac{1}{\gamma}\left(\frac{1}{M}\sum_{m=1}^M P_{Z|W=m}(\mathcal{D}_0^c)-E_{\gamma}(P_{Z}\|{{\pi}}_Z)\right)
\\
&\ge\frac{1}{\gamma}\left(1-{\epsilon}-E_{\gamma}(P_Z\|{{\pi}}_Z)\right).
\end{align}
To see \eqref{e_46}, let $\mathcal{D}_m$ be the set of outputs $z\in\mathcal{Z}$ for which the eavesdropper list contains $m\in\{1,\dots,M\}$.
Then
\begin{align}
\frac{T}{MA}
&\ge \frac{T}{M}\mu_Z(\mathcal{D}_0^c)
\\
&=\frac{1}{M}\int_{\mathcal{D}_0^c} T(z){\rm d}\mu_Z(z)
\\
&=\frac{1}{M}\int\sum_{m=1}^M 1\{z\in\mathcal{D}_m\}{\rm d}\mu_Z(z)
\\
&=\frac{1}{M}\sum_{m=1}^M\int 1\{z\in\mathcal{D}_m\}{\rm d}\mu_Z(z)
\\
&=\frac{1}{M}\sum_{m=1}^M\mu_Z(\mathcal{D}_m)
\\
&\ge\frac{1}{M\gamma}\sum_{m=1}^M \left(P_{Z|W=m}(\mathcal{D}_m)-E_{\gamma}(P_{Z|W=m}\|\mu_Z)\right)
\label{e54}
\\
&\ge\frac{1}{\gamma}\left(1-{\epsilon}-\frac{1}{M}\sum_{m=1}^M E_{\gamma}(P_{Z|W=m}\|\mu_Z)\right).
\end{align}
\end{proof}
Next, we particularize Theorem~\ref{thmlem} to the asymptotic setting.
\begin{proof}[Proof of ``only if'' in Theorem~\ref{region}]
~\newline~
\begin{itemize}
\item Fix an arbitrary
\begin{align}
\alpha> D(Q_{\sf Z}\|{{\pi}}_{\sf Z})+[I(Q_{\sf X},P_{\sf Z|X})-R-R_{\sf L}]^+.
\end{align}
We will show that $(\alpha,\tau)$ is not ${\epsilon}$-achievable by the eavesdropper for any $\tau>0$ and $\epsilon\in (0,1)$. Pick $\sigma>0$ such that
\begin{align}\label{e6_25_0}
\alpha> D(Q_{\sf Z}\|{{\pi}}_{\sf Z})+[I(Q_{\sf X},P_{\sf Z|X})-R-R_{\sf L}]^++2\sigma
\end{align}
and define
\begin{align}\label{eantn_0}
\left\{
\begin{array}{c}
A_n=\exp(n(\alpha-\sigma))
  \\
T_n=\exp(n(\tau+\sigma)).
\end{array}
\right.
\end{align}
Assuming the eavesdropper can perform $(A_n,T_n,{\epsilon}(c^{ML}))$-decoding for a particular realization of the codebook $c^{ML}$, then applying Theorem~\ref{thmlem} with
\begin{align}
\gamma_n&=\exp(n(D(Q_{\sf Z}\|{{\pi}}_{\sf Z})+[I(Q_{\sf X},P_{\sf Z|X})-R-R_{\sf L}]^++\sigma)),
\label{e321}
\end{align}
we obtain
\begin{align}
&\exp(n(D(Q_{\sf Z}\|{{\pi}}_{\sf Z})+[I(Q_{\sf X},P_{\sf Z|X})-R-R_{\sf L}]^+-\alpha+2\sigma))
\nonumber\\
&=\frac{\gamma_n}{A_n}
\\
&\ge 1-{\epsilon}(c^{ML})-E_{\gamma_n}(P_{Z^n[c^{ML}]}\|\pi_{\sf Z}^{\otimes n}).
\end{align}
From \eqref{e6_25_0}, the above implies
\begin{align}
E_{\gamma_n}(P_{Z^n[c^{ML}]}\|\pi_{\sf Z}^{\otimes n})\ge\frac{1-{\epsilon}(c^{ML})}{2}.
\end{align}
For sufficiently large $n$.
By Corollary~\ref{cor_asymp} and \eqref{e321},
the average of the left side converges to zero as $n\to\infty$,
thus the average of the right side cannot be lower bounded by $\frac{1-\epsilon}{2}$.

  \item Fix an arbitrary
\begin{align}
\tau < R-[I(Q_{\sf X},P_{\sf Z|X})-R_{\sf L}]^+.
\end{align}
We will show that $(\alpha,\tau)$ is not ${\epsilon}$-achievable by the eavesdropper for any $\alpha>0$ and $\epsilon\in (0,1)$. Pick $\sigma>0$ such that
\begin{align}\label{e6_25}
\tau +2\sigma< R-[I(Q_{\sf X},P_{\sf Z|X})-R_{\sf L}]^+
\end{align}
and again define $A_n$ and $T_n$ as in \eqref{eantn_0}.
Assuming the eavesdropper can perform $(A_n,T_n,{\epsilon}(c^{ML}))$-decoding for a particular realization of the codebook $c^{ML}$, then applying Theorem~\ref{thmlem} with
\begin{align}
\mu_Z&=Q_{\sf Z}^{\otimes n},
\\
\gamma_n&=\exp(n([I(Q_{\sf X},P_{\sf Z|X})-R_{\sf L}]^++\sigma)),
\end{align}
and noting that $A_n\ge 1$,
we obtain
\begin{align}
&\quad\exp(n(\tau-R+[I(Q_{\sf X},P_{\sf Z|X})-R_{\sf L}]^++2\sigma))
\nonumber
\\
&=\frac{T_n\gamma_n}{M_n}
\\
&\ge 1-{\epsilon}(c^{ML})-\frac{1}{M}\sum_{m=1}^M E_{\gamma_n}(P_{Z^n|W=m}\|Q_{\sf Z}^{\otimes n})
\\
&= 1-{\epsilon}(c^{ML})-\frac{1}{M}\sum_{m=1}^M E_{\gamma_n}(P_{Z^n[{c_m}^L]}\|Q_{\sf Z}^{\otimes n})
\end{align}
where ${c_m}^L:=(c_{ml})_{l=1}^L$.
From \eqref{e6_25}, the above implies
\begin{align}
\frac{1}{M}\sum_{m=1}^M E_{\gamma_n}(P_{Z^n[{c_m}^L]}\|Q_{\sf Z}^{\otimes n})\ge\frac{1-{\epsilon}(c^{ML})}{2},
\end{align}
for sufficiently large $n$.
Invoking Corollary~\ref{cor_asymp},
we see the average of the right side with respect to the codebook converges zero as $n\to\infty$, and in particular cannot be lower-bounded by $\frac{1-\epsilon}{2}$.
\end{itemize}
\end{proof}

\subsection{Ensemble Tightness}
The (eavesdropper) achievability part of Theorem~\ref{region} follows by analyzing the eavesdropper list decoding ability for different cases of the rates $(R,R_{\sf L})$. First, consider the following one-shot achievability bounds for channel coding with possibly no message sent:
\begin{thm}\label{eaveAchieve}
Consider a random transformation $P_{Z|X}$ and a $(M,L,Q_X)$-random code.
Let $Q_X\to P_{Y|X}\to Q_Y$, and let $\pi_Z$ be the distribution of the eavesdropper observation when no message is sent.
Define
  \begin{align}
  \bar{\imath}_{Z;X}(z;x)&:=\log\frac{{\rm d}P_{Z|X=x}}{{\rm d}{{\pi}}_{Z}}(z);
  \\
  \imath_{Z;X}(z;x)&:=\log\frac{{\rm d}P_{Z|X=x}}{{\rm d}Q_{Z}}(z).
  \end{align}
Let $\delta,\beta,A,T>0$.
Then, there exist three list decoders such that for Decoder~1,
  \begin{align}
  \mathbb{E}_{\mathcal{C}}\mathbb{P}[\textrm{error}|\textrm{no message}]&\le\frac{1}{A}\exp(-\delta),
  \label{e335}
  \\
  \mathbb{E}_{\mathcal{C}}\mathbb{P}[\textrm{error}|\textrm{message is }m]
  &\le\mathbb{P}[\bar{\imath}_{Z;X}(Z;X)\le\log(LMA)+\delta]
  \nonumber\\
  &\quad+\mathbb{P}[\imath_{Z;X}(Z;X)\le\log\frac{LM}{T}+\delta]
  \nonumber
  \\
  &\quad+\frac{1}{1+\beta}+e^{-(\beta+1)}+\beta\exp(-\delta).
  \label{e336}
  \end{align}
  Here an error in the case of no message means that a non-empty list is produced. An error in the case of message $m$ means either the list does not contain $m$, or the list size exceeds $T$.\footnote{Such a decoder is a variable list-size decoder.
  However, we can add a post processor which declares no message if the list is empty, or outputs a list of fixed size $T$ otherwise (by arbitrarily deleting or adding messages to the list),
  resulting a new decoder as considered in
  Definition~\ref{defn_achieve},
  and the two types of error probability for the new decoder (i.e.~the best values of $\frac{1}{A}$ and $\epsilon_m$ in Definition~\ref{defn_achieve}) do not exceed the two types of error probability for the original variable list-size decoder.}
  For Decoder~2,
  \begin{align}
  \mathbb{E}_{\mathcal{C}}\mathbb{P}[\textrm{error}|\textrm{no message}]&\le\frac{1}{A};
  \\
  \mathbb{E}_{\mathcal{C}}\mathbb{P}[\textrm{error}|\textrm{message is }m]
  &\le\mathbb{P}\left[\imath_{Z;X}(Z;X)\le\log\frac{LM}{T}+\delta\right]
  \nonumber\\
  &\quad+\mathbb{P}[\imath_{Q_Z\|{{\pi}}_Z}(Z)\le \log A]
  \nonumber
  \\
  &\quad+\frac{1}{1+\beta}+e^{-(\beta+1)}+\beta\exp(-\delta),
  \end{align}
  where the error events are defined similarly to Decoder~1.
  Decoder~3 either output an empty list or a list of all messages, and
  \begin{align}
  \mathbb{E}_{\mathcal{C}}\mathbb{P}[\textrm{error}|\textrm{no message}]&\le\frac{1}{A};
  \\
  \mathbb{E}_{\mathcal{C}}\mathbb{P}[\textrm{error}|\textrm{message is }m]
  &\le\mathbb{P}[\imath_{Q_Z\|{{\pi}}_Z}(Z)\le \log A].
  \end{align}
\end{thm}
\begin{proof}[Proof of Theorem~\ref{eaveAchieve}]
See Appendix~\ref{app_eaveAchieve}.
\end{proof}
Under various conditions, one out of the three decoders are asymptotically optimal. By choosing appropriate parameters $\delta,\beta,A,T>0$, it is clear that Theorem~\ref{eaveAchieve} implies the following:
\begin{cor}\label{cor_ea}
Fix any $Q_{\sf X}$, $R$, $R_{\sf L}$ and $0<{\epsilon}<1$. Consider $(\exp(nR),\exp(nR_{\sf L}),Q_{\sf X}^{\otimes n})$-random codes and stationary memoryless channel with per-letter conditional distribution $Q_{\sf Z|X}$.
\begin{itemize}
  \item When $R+R_{\sf L}<I(Q_{\sf X},P_{\sf Z|X})$, the rate pair $(\alpha,\tau)$ is ${\epsilon}$-achievable by a Decoder~1\footnote{By which we mean it is possible to choose the $\delta,\beta,A,T$ parameters for the decoder to achieve the desired performance.} if
      \begin{align}
      \left\{
      \begin{array}{c}
        D(Q_{\sf Z}\|{{\pi}}_{\sf Z})+I(Q_{\sf X},P_{\sf Z|X})>R_{\sf L}+R+\alpha; \\
        I(Q_{\sf X},P_{\sf Z|X})>R+R_{\sf L}-\tau.
      \end{array}
      \right.
      \label{e342}
      \end{align}
  \item When $R+R_{\sf L}\ge I(Q_{\sf X},P_{\sf Z|X})$ but $R_{\sf L}<I(Q_{\sf X},P_{\sf Z|X})$, the rate pair $(\alpha,\tau)$ is ${\epsilon}$-achievable by a Decoder~2 if
      \begin{align}
      \left\{
      \begin{array}{c}
        I(Q_{\sf X},P_{\sf Z|X})>R_{\sf L}+R-\tau; \\
        \alpha<D(Q_{\sf Z}\|{{\pi}}_{\sf Z}).
      \end{array}
      \right.
      \end{align}
  \item When $R_{\sf L}\ge I(Q_{\sf X},P_{\sf Z|X})$, the rate pair $(\alpha,\tau)$ is ${\epsilon}$-achievable by a Decoder~3 if
      \begin{align}
      \left\{
      \begin{array}{c}
        \tau\ge R; \\
        \alpha<D(Q_{\sf Z}\|{{\pi}}_{\sf Z}).
      \end{array}
      \right.
      \end{align}
\end{itemize}
\end{cor}
\begin{proof}[Proof Sketch]
Consider the first case. To see the achievability of $(\alpha,\tau)$ satisfying \eqref{e342},
choose
\begin{align}
\delta_n&:=n^{0.9},
\\
\beta_n&:=n,
\\
A_n&:=\exp(n\alpha),
\\
T_n&:=\exp(n\tau).
\end{align}
Then the right sides of \eqref{e335} and \eqref{e336} converges to zero as $n\to\infty$.
The analyses of the other two cases are similar using the same choice of the parameters as above.
\end{proof}
The eavesdropper's achievability (``if '' part) of Theorem~\ref{region} then follows from Corollary~\ref{cor_ea} and an application of the standard diagonalization argument to show that the achievable region is closed (see \cite{han2003information}).

%
%
%

\section{Conclusion and Future Work}\label{sec_conclusion}
This paper develops general bounds among various distance metrics, and showed that, in the memoryless case, the exponential growth of $\gamma$ for the $E_{\gamma}$ between the synthesized and target distributions to vanish is the same as the linear growth of the relative entropy or smooth R\'{e}nyi $\alpha$-divergence (of order $\alpha\neq1$).
This implies that in the context of relative entropy gauged approximation the bound in \cite[Theorem~12]{han1993approximation} is not asymptotically tight. An intuitive explanation for this is that \cite[Theorem~12]{han1993approximation} uses random codebooks where the distribution of each codeword induces the target output distribution through the stationary memoryless channel. However, this is not necessary, and in fact the optimal choice of the codeword distribution  generally does \emph{not} induce the target distribution. This is in stark contrast to the conventional resolvability under total variation distance, where the codeword distribution \emph{must} induce the target distribution to cause total variation between the output distribution and the target distribution to vanish.

We have seen three examples of the application of $E_{\gamma}$-resolvability in information theory. The essence is a change-of-measure: if $E_{\gamma}(P\|Q)$ is small and $P(\mathcal{A})$ is large, then $Q(\mathcal{A})$ is essentially lower bounded by $\frac{1}{\gamma}$.
There are two motivations for such a change-of-measure:

1) change one distribution to another distribution which is simpler to analyze. Consider the source coding example; in \eqref{e280} we changed the real joint distribution of the source and the reproduction to $P_{XU}$, which is simpler to analyze since its average is the given joint distribution $Q_{XU}$. The application to one-shot mutual covering lemma (step \eqref{e292}) also falls into this category.

2) change one distribution from a family to a common distribution. In the wiretap channel example, step \eqref{e54} upper-bounds the conditional probability $P_{Z|W=m}(\mathcal{D}_m)$ for each message using the probability under a fixed distribution $\mu_Z(\mathcal{D}_m)$. The sum of $\mu_Z(\mathcal{D}_m)$ over $m$ would be $T$ since (on average) the space $\mathcal{Z}$ is covered by $(\mathcal{D}_m)$ for $T$ times.

Other distance measures have been useful for change-of-measure in information theory. For example, $\beta_{\alpha}(P\|Q)$ is used in the converse proof for the error exponent of lossy compression \cite{csiszar1981information}. The relative entropy can also play a similar role (e.g.~proof of sphere packing bound in \cite[Theorem~2.5.3]{csiszar1981information}): the Log-Sum Inequality implies that
\begin{align}
Q(\mathcal{A})\ge\exp\left(-\frac{D(P\|Q)+h(Q(\mathcal{A}))}{P(\mathcal{A})}\right)
\end{align}
where $h(\cdot)$ is the binary entropy function. Thus if ${P(\mathcal{A})}$ is close to $1$ and $D(P\|Q)\gg1$, then $Q(\mathcal{A})$ is essentially lower bounded by $\exp(-D(P\|Q))$. Nevertheless, the $E_{\gamma}$ metric is sometimes more desirable because of its nice properties and relations to other metrics (see Proposition~\ref{prop3} and \ref{prop_1}). For example, $\beta_{\alpha}$ and $D(\cdot\|\cdot)$ do not seem to share analogues of the triangle inequality (Proposition~\ref{prop3}-\ref{prop3_2})) of $E_{\gamma}$.

In Section~\ref{sec_worst} the achievability of identification coding is used to derive a converse for worst-case resolvability. The contrapositive of this argument is that achievability of worst-case resolvability would imply a converse for identification coding. Indeed, \cite{han1992}\cite{han1993approximation} have initiated this approach to prove the strong converse of identification where resolvability in total variation distance is used; with achievability of resolvability in $E_{\gamma}$ (e.g.~Theorem~\ref{thm_ma}) the
minimum of the two exponents for ID coding over a general sequence of channels
can be upper-bounded, in contrast to the approach of \cite{ahlswede1989identification} specific to DMC. But for DMC, our preliminary study indicates that the resulting bound is not as tight.

Another potential application of the $E_{\gamma}$-based analysis, suggested by Yury Polyanskiy \cite{polyanskiy2014private}, is the study of the sphere-packing exponent in channel coding.
In \cite{bennett2014quantum},
the achievability of channel synthesis (under total variation distance) is used to prove the strong converse of channel coding.
This combined with a standard change-of-measure argument (see \cite[Theorem~2.5.3]{csiszar1981information}) yields the conventional sphere-packing bound.
In \cite{polyanskiy2014private} Polyanskiy pointed out that the sphere-packing bound might be improved if one could prove that channel synthesis under the more forgiving $E_{\gamma}$ metric requires smaller communication rates.
Since resolvability under total variation distance has been used to prove channel synthesis under total variation distance \cite{cuff2012distributed}, it is natural to ask whether $E_{\gamma}$-resolvability can lead to an $E_{\gamma}$-channel synthesis result.
Unfortunately, it appears not to be the case,
because of the asymmetry of $E_{\gamma}(P\|Q)$ in $P$ and $Q$ (for $\gamma>1$).

\section{Acknowledgements}
Our initial focus was on the alternative formulation of Wyner's wiretap setup without using information measures (Section~\ref{sec_wiretap}) as well as the excess information metric for channel resolvability, as in Theorem~\ref{thm3}. We gratefully acknowledge Yury Polyanskiy for bringing the $E_{\gamma}$ metric to our attention and sharing \cite{polyanskiy2014private} (see Section~\ref{sec_conclusion}).
This work was supported by the NSF under Grants CCF-1350595, CCF-1116013, 
by the Center for Science of Information (CSoI), an NSF Science and Technology Center, under grant agreement CCF-0939370,
and the Air Force Office of Scientific Research under Grant FA9550-12-1-0196.

\appendices
\section{Proof of Proposition~\ref{prop1}}\label{pf_prop1}
We first show that it suffices to consider the case of $|\mathcal{X}|\le3$. Put
\begin{align}
\mathcal{B}_t := \left\{ x\colon \frac{{\rm d}P}{{\rm d}Q} (x) > t\right\},\quad\forall t\ge 0.
\end{align}
We partition the whole space as
\begin{align}
X= \mathcal{A}_1 \cup
\mathcal{A}_2 \cup
\mathcal{A}_3 \cup
\mathcal{A}_4
\end{align}
with
\begin{align}
\mathcal{A}_1 &= \mathcal{B}_1\cap \mathcal{B}_{\lambda};
 \label{ea1}
\\
\mathcal{A}_2 &= \mathcal{B}_1^c \cap \mathcal{B}_{\lambda};
 \label{ea2}
\\
\mathcal{A}_3 &= \mathcal{B}_1 \cap \mathcal{B}_{\lambda}^c;
 \label{ea3}
\\
\mathcal{A}_4 &= \mathcal{B}_1^c\cap\mathcal{B}_{\lambda}^c.
 \label{ea4}
\end{align}
Let $\mathcal{Z}=\{1,2,3,4\}$ and $z(x)$ be a function of $x$ indicating which of the above nonempty sets $x$ belongs to. Next observe that
\begin{align}
P_Z(1)>\lambda Q_Z(1),&\quad P_Z(1)> Q_Z(1) \textrm{ if } Q_Z(1)>0;\label{e2_11}
\\
P_Z(2)>\lambda Q_Z(2),&\quad P_Z(2)\le Q_Z(2) \textrm{ if } Q_Z(2)>0;\label{e2_12}
\\
P_Z(3)\le\lambda Q_Z(3),&\quad P_Z(3)> Q_Z(3) \textrm{ if } Q_Z(3)>0;\label{e2_13}
\\
P_Z(4)\le \lambda Q_Z(4),&\quad P_Z(4)\le Q_Z(4).\label{e2_14}
\end{align}
For example, to see \eqref{e2_12}, consider
\begin{align}
P_Z(2)&=\int1\{x\in\mathcal{A}_2\}{\rm d}P_X(x)
\\
&\le\int1\{x\in\mathcal{A}_2\}{\rm d}Q_X(x)
\\
&=Q_Z(2).
\end{align}
This establishes the second inequality in \eqref{e2_12}.
The first inequality in \eqref{e2_12} takes a little more effort since it is not strict. Define
\begin{align}
\mathcal{A}_{2,t}:=\left\{x:\frac{{\rm d}P}{{\rm d}Q}(x)>t \textrm{ and }\frac{{\rm d}P}{{\rm d}Q}(x)\le 1\right\},
\quad\forall t>0.
\end{align}
Then clearly, $\mathcal{A}_{2}=\bigcup_{t>\lambda}\mathcal{A}_{2,t}$. By continuity of measure, there exists $t>\lambda$ such that
\begin{align}
Q(\mathcal{A}_{2,t})>\frac{1}{2}Q(\mathcal{A}_2),
\end{align}
so that
\begin{align}
P_Z(2)&=\int1\{x\in\mathcal{A}_2\}{\rm d}P(x)
\\
&=\int1\{x\in\mathcal{A}_{2,t}\}{\rm d}P(x)
+\int1\{x\in\mathcal{A}_2\setminus\mathcal{A}_{2,t}\}
{\rm d}P(x)
\\
&\ge\int1\{x\in\mathcal{A}_{2,t}\}t{\rm d}Q(x)
+\int1\{x\in\mathcal{A}_2\setminus\mathcal{A}_{2,t}\}
\lambda{\rm d}Q(x)
\\
&=\int1\{x\in\mathcal{A}_{2,t}\}(t-\lambda){\rm d}Q(x)
+\int1\{x\in\mathcal{A}_2\}\lambda{\rm d}Q(x)
\\
&>\frac{t-\lambda}{2}Q(\mathcal{A}_2)+\lambda Q_Z(2).
\end{align}
Thus \eqref{e2_12} is established, and \eqref{e2_11}, \eqref{e2_13} \eqref{e2_14} can be proved similarly (Note that there is no condition in \eqref{e2_14} because neither of the inequalities in \eqref{e2_14} is strict). \eqref{e2_11}-\eqref{e2_14} imply that
\begin{align}
|P-Q|&=|P_Z-Q_Z|;
\\
\mathbb{P}\left[\frac{{\rm d}P}{{\rm d}Q}(X)>\lambda\right]&=\mathbb{P}\left[\frac{{\rm d}P_Z}{{\rm d}Q_Z}(Z)>\lambda\right],
\end{align}
where $Z=z(X)$. However depending on the value of $\lambda$, one of the four sets defined in \eqref{ea1}-\eqref{ea4} is empty. Therefore we have shown that it suffices to consider the case of $|\mathcal{X}|\le3$, so that verifying \eqref{ub} and \eqref{lb} becomes elementary. We begin with the proof of the upper bound in \eqref{ub}:
\begin{enumerate}
  \item $\lambda\ge\frac{1}{1-\delta}$.\\
  The upper bound \eqref{ub} follows from \cite[Theorem~9]{verdu2014total}. To verify its tightness, consider $|\mathcal{X}|=2$, and
    \begin{align}
    P&:=\left[1-\frac{\lambda^+\delta}{\lambda^+-1},\frac{\lambda^+\delta}{\lambda^+-1}\right];\\
    Q&:=\left[1-\frac{\delta}{\lambda^+-1},\frac{\delta}{\lambda^+-1}\right];
    \end{align}
    where $\lambda^+>\lambda$ ensures that $P$ and $Q$ are distributions, and $\frac{\lambda^+\delta}{\lambda^+-1}$ can be made arbitrarily close to $\frac{\lambda\delta}{\lambda-1}$ as $\lambda^+\downarrow\lambda$.

  \item $\lambda<\frac{1}{1-\delta}$.
    The upper bound is trivial. Its tightness can be seen by choosing
    \begin{align}
    P&:=[0,1];\\
    Q&:=[\delta,1-\delta].
    \end{align}
\end{enumerate}
The lower bound \eqref{lb} is proved as follows:
\begin{enumerate}
  \item $\lambda>\frac{1}{1-\delta}$.\\
  The lower bound is trivial. Its tightness can be seen by choosing
  \begin{align}
    P&:=\left[1-\frac{\lambda\delta}{\lambda-1},\frac{\lambda\delta}{\lambda-1}\right];\\
    Q&:=\left[1-\frac{\delta}{\lambda-1},\frac{\delta}{\lambda-1}\right];
    \end{align}

  \item $1<\lambda\le\frac{1}{1-\delta}$.
   Without loss of generality, we may assume that $\mathcal{X}=\{1,2,3\}$ and
   \begin{align}
   \frac{P(1)}{Q(1)}&\le 1;
   \\
   1<\frac{P(2)}{Q(2)}&\le \lambda;
   \\
   \lambda<\frac{P(3)}{Q(3)};
   \end{align}
   Then the lower bound \eqref{lb} follows from
   \begin{align}
   P(3)&= 1-P(1)-P(2)
   \\
   &\ge 1-P(1)-\lambda Q(2)
   \\
   &= 1-P(1)-\lambda [1-Q(1)-Q(3)]
   \\
   &= 1-P(1)-\lambda [1-(P(1)+\delta)-Q(3)]
   \\
   &= 1+(\lambda-1)P(1)+\lambda Q(3)-\lambda+\lambda\delta
   \\
   &\ge 1-\lambda+\lambda\delta
   \end{align}
   and its tightness can be verified by considering the following distribution as $\epsilon\to0$:
   \begin{align}
   P&:=[0,(1-\delta-\epsilon)\lambda,1-(1-\delta-\epsilon)\lambda]
   \\
   Q&:=[\delta,1-\delta-\epsilon,\epsilon].
   \end{align}

   \item $1-\delta < \lambda \leq 1$. The lower bound \eqref{lb} follows from
   \begin{align}
   \mathbb{P}\left[\frac{{\rm d}P}{{\rm d}Q}(X)>\lambda\right]
   &\ge \mathbb{P}\left[\frac{{\rm d}P}{{\rm d}Q}(X)>1\right]-\mathbb{P}\left[\frac{{\rm d}P}{{\rm d}Q}(\tilde{X})>1\right]
   \\
   &=\delta
   \end{align}
   where $\tilde{X}\sim Q$.
   The tightness can be see by considering
   \begin{align}
   P&:=[1-\delta,\delta];
   \\
   Q&:=[1,0].
   \end{align}

   \item $\lambda\le 1-\delta$.\\
   Without loss of generality, assume that $\mathcal{X}=\{1,2,3\}$ and
   \begin{align}
   \frac{P(1)}{Q(1)}&\le \lambda;
   \\
   \lambda<\frac{P(2)}{Q(2)}&\le 1;
   \\
   \frac{P(2)}{Q(2)}&>1.
   \end{align}
   Then observe that
   \begin{align}
   1-P(1)-Q(3)-\delta&=1-P(1)-P(3)
   \\
   &=P(2)
   \\
   &\le Q(2)
   \\
   &=1-Q(1)-Q(3)
   \\
   &\le 1-\frac{1}{\lambda}P(1)-Q(3)
   \end{align}
   which upon rearrangements gives
   \begin{align}
   P(1)\ge \frac{\lambda\delta}{1-\lambda}
   \end{align}
   implying the lower bound \eqref{lb}. To see its tightness, consider
   \begin{align}
   P&:=\left[\frac{\lambda\delta}{1-\lambda},1-\frac{\delta}{1-\lambda},\delta\right];
   \\
   Q&:=\left[\frac{\delta}{1-\lambda},1-\frac{\delta}{1-\lambda},0\right].
   \end{align}
\end{enumerate}

\section{Proof of Proposition~\ref{prop_1}}\label{pf_prop_1}
\begin{enumerate}
  \item The first inequality in \eqref{e47} is evident from the definition. For the second inequality in \eqref{e47}, consider the event
      \begin{align}
      \mathcal{A}:=\{x\in\mathcal{X}:\imath_{P\|Q}(x)>\log\gamma\}.
      \end{align}
      Then
      \begin{align}
      &\quad E_{\frac{\gamma}{a}}(P\|Q)
      \nonumber\\
      &\ge
      P(\mathcal{A})-\frac{\gamma}{a}Q(\mathcal{A})
      \nonumber\\
      &\ge \mathbb{P}[\imath_{P\|Q}(X)>\log\gamma]
      -\frac{\gamma}{a}\cdot \frac{1}{\gamma}\mathbb{P}[\imath_{P\|Q}(X)>\log\gamma]
      \\
      &\ge \frac{a-1}{a}\mathbb{P}[\imath_{P\|Q}(X)>\log\gamma]
      \end{align}
      and the result follows by rearrangement.

  \item Direct from the monotonicity of R\'{e}nyi divergences in $\alpha$ and the definitions of their smooth versions.
  \item The bound \eqref{e_prop20} follows from $D(P\|Q)=\mathbb{E}[\imath_{P\|Q}(\hat{X})]$ and
  \eqref{e_prop_21} can be seen from
    \begin{align}
    D(P\|Q)&\ge \mathbb{E}[|\imath_{P\|Q}(\hat{X})|]-2e^{-1}\log e
    \label{e_prop24}
    \\
    &\ge \log\gamma \mathbb{P}[\imath_{P\|Q}(\hat{X})>\log\gamma]-2e^{-1}\log e
    \label{e_prop25}
    \\
    &\ge \log\gamma E_{\gamma}(P\|Q)-2e^{-1}\log e
    \label{e_prop26}
      \end{align}
   where \eqref{e_prop24} is due to Pinsker \cite[(2.3.2)]{pinsker}, \eqref{e_prop25} uses Markov's inequality, and \eqref{e_prop26} is from the definition \eqref{e_prop1}.
   \item 
By considering the $\frac{{\rm d}\mu}{{\rm d}Q}>\gamma$ and $\frac{{\rm d}\mu}{{\rm d}Q}\le\gamma$ cases separately, we can check the following (homogeneous) inequalities:
for each $\alpha<1$,
       \begin{align}
\left|\frac{{\rm d}\mu}{{\rm d}Q}-\gamma\right|
&\ge\frac{{\rm d}\mu}{{\rm d}Q}-2\gamma^{1-\alpha}\left(\frac{{\rm d}\mu}{{\rm d}Q}\right)^{\alpha}+\gamma,
\label{e_prop29}
\end{align}
and for each $\alpha>1$,
\begin{align}
\left|\frac{{\rm d}\mu}{{\rm d}Q}-\gamma\right|
&\le-\frac{{\rm d}\mu}{{\rm d}Q}+2\gamma^{1-\alpha}\left(\frac{{\rm d}\mu}{{\rm d}Q}\right)^{\alpha}+\gamma.
\label{e_prop30}
\end{align}
Integrating with respect to ${\rm d}Q$ both sides of \eqref{e_prop29}, we obtain
\begin{align}
|\mu-\gamma Q|\ge \mu(\mathcal{X})-2\gamma^{1-\alpha}\int\left(\frac{{\rm d}\mu}{{\rm d}Q}\right)^{\alpha}{\rm d}Q+\gamma
\end{align}
and \eqref{e_prop21} follows by rearrangement.
Integrating with respect to ${\rm d}Q$ on both sides of \eqref{e_prop30}, we obtain
\begin{align}
|\mu-\gamma Q|\le -\mu(\mathcal{X})+2\gamma^{1-\alpha}\int\left(\frac{{\rm d}\mu}{{\rm d}Q}\right)^{\alpha}{\rm d}Q+\gamma
\end{align}
and \eqref{e_prop21_1} follows by rearrangement.

\item Immediate from the previous result, the definition of smooth R\'{e}nyi divergence, and the triangle inequality \eqref{e19}.

\item $\Rightarrow$: By assumption there exists a nonnegative finite measure $\mu$ such that $E_1(P\|\mu)\le \epsilon$ and $\mu\le \gamma Q$. Then from Proposition~\ref{prop_tri},
\begin{align}
E_{\gamma}(P\|Q)&\le E_1(P\|\mu)+E_{\gamma}(\mu\|Q)\nonumber
\\
&\le\epsilon+0.\nonumber
\end{align}
$\Leftarrow$: Define $\mu$ by $\frac{{\rm d}\mu}{{\rm d}Q}:=\min\{\frac{{\rm d}P}{{\rm d}Q},\gamma\}$.
Since $\left\{\frac{{\rm d}P}{{\rm d}Q}>\gamma\right\}
=\left\{\frac{{\rm d}P}{{\rm d}\mu}>1\right\}$,
\begin{align}
E_1(P\|\mu)&=P\left(\frac{{\rm d}P}{{\rm d}Q}>\gamma\right)-\mu\left(\frac{{\rm d}P}{{\rm d}Q}>\gamma\right)\nonumber
\\
&=P\left(\frac{{\rm d}P}{{\rm d}Q}>\gamma\right)-\gamma Q\left(\frac{{\rm d}P}{{\rm d}Q}>\gamma\right)\nonumber
\\
&\le E_{\gamma}(P\|Q).\nonumber
\end{align}
Then $D_{\infty}(\mu\|Q)\le\log\gamma$ implies that
$D^{-\epsilon}_{\infty}(P\|Q)\le \log\gamma$.

\item ``$\ge$'': let $\mathcal{A}$ be a set achieving the minimum in \eqref{eq46}. Let $\mu$ be the restriction of $P$ on $\mathcal{A}$. Then by definition \eqref{eq46} we have
    \begin{align}
    E_1(P\|\mu)=P(\mathcal{A}^c)\le\epsilon
    \end{align}
    therefore
    \begin{align}
    D_0^{+\epsilon}(P\|Q)\ge-\log\beta_{1-\epsilon}(P,Q).
    \end{align}
    ``$\le$'': fix arbitrary $\delta>0$ and let $\mu$ be a nonnegative measure satisfying
    \begin{align}
    D_0^{+\epsilon}(P\|Q)<D_0(\mu\|Q)+\delta.
    \end{align}
    Define $\mathcal{A}:=\supp(\mu)$. Then
    \begin{align}
    P(\mathcal{A}^c)&=P(\mathcal{A}^c)-\mu(\mathcal{A}^c)
    \\
    &\le E_1(P\|\mu)
    \\
    &\le \epsilon,
    \end{align}
    hence
    \begin{align}
D_0^{+\epsilon}(P\|Q)-\delta
    &\le
    D_0(\mu\|Q)
    \\
    &=
    -\log Q(\mathcal{A})
    \\
    &\le
        -\log\beta_{1-\epsilon}(P,Q)
    \end{align}
    and the result follows by setting $\delta\downarrow 0$.

   \item
   In the case of non-atomic $P$, there exists a set $\mathcal{A}\subseteq \{x\colon\imath_{P\|Q}(x)>\tau\}$ such that
   $P(\mathcal{A})= 1-\epsilon$. Then
\begin{align}
\beta_{1-\epsilon}(P,Q)
&\le Q(\mathcal{A})
\\
&\le \exp(-\tau)P(\mathcal{A}).
\end{align}
We obtain \eqref{e_53} by taking the logarithms on both sides of the above and invoking Part~\ref{ptD0}.
When $P$ is not necessarily non-atomic,
since $P(\imath_{P\|Q}>\tau)\ge 1-\epsilon$,
\begin{align}
\beta_{1-\epsilon}(P,Q)
&\le Q(\imath_{P\|Q}>\tau)
\\
&\le \exp(-\tau)P(\imath_{P\|Q}>\tau)
\\
&\le \exp(-\tau)
\end{align}
and again the result follows by taking the logarithms on both sides of the above.
\end{enumerate}

\section{Proof of Theorem~\ref{eaveAchieve}}\label{app_eaveAchieve}
\begin{itemize}
  \item Codebook generation: $
  (c_{ij})_{1\le i\le M,1\le j\le L}$ according to $Q_{X}^{\otimes ML}$.
  \item Decoders:
  Fix an arbitrary constant $\delta>0$.
  Upon observing $z$, Decoder~1 outputs as a list all $1\le i\le M$ such that there exists $1\le j\le L$ satisfying
  \begin{align}
  \left\{\begin{array}{rl}
      \bar{\imath}_{Z;X}(z;c_{ij})&>\log(LMA)+\delta \\
      \imath_{Z;X}(z;c_{ij})&>\log\frac{LM}{T}+\delta
    \end{array}\right.
  \end{align}
  if there is at least one such an $i$,
  or declares that no message is sent (i.e.~outputs an empty list) if otherwise. Decoder~2 outputs as a list all $1\le i\le M$ such that there exists $1\le j\le L$ satisfying
  \begin{align}
      \imath_{Z;X}(z;c_{ij})>\log\frac{LM}{T}+\delta
  \end{align}
  if there exists at least one such $i$ and in addition,
  \begin{align}\label{edthres}
  \imath_{Q_Z\|{{\pi}}_Z}(z)> \log A,
  \end{align}
  or declares that no message is sent if otherwise. Decoder~3 outputs $\{1,\dots,M\}$ as the list if \eqref{edthres} holds (so that the list size equals $M$), or otherwise declares no message.

  \item Error analysis: we denote by $\mathcal{L}$ the list of messages recovered by the eavesdropper.\\
  Decoder~1:
\begin{align}
  &\quad\mathbb{P}[\mathcal{L}\neq\emptyset|\textrm{no message}]
  \nonumber\\
  &\le \mathbb{P}\left[\max_{1\le m\le M,1\le l\le L}\bar{\imath}_{Z;X}(\bar{Z};X_{ml})>\log (LMA)+\delta\right]
  \\
  &\le \frac{1}{A}\exp(-\delta)\label{echm0}
  \end{align}
  where the probability is averaged over the codebook, $(X^{ML},\bar{Z})\sim Q_X^{\otimes ML}\times {{\pi}}_Z$, and \eqref{echm0} used the packing lemma \cite{verdu2012non}. Moreover
  \begin{align}
  &\quad\mathbb{P}[\textrm{$1\notin\mathcal{L}$ or $\mathcal{L}=\emptyset$}|W=1]
  \\
  &\le \mathbb{P}[\bar{\imath}_{Z;X}(Z;X)\le\log(LMA)+\delta]
  \nonumber\\
  &\quad+\mathbb{P}\left[\imath_{Z;X}(Z;X)\le\log\frac{LM}{T}+\delta\right]
  \end{align}
  where $W$ denotes the message sent, and $(X,Z)\sim Q_{XZ}$.
  Further,
  \begin{align}
  &\quad\mathbb{P}[|\mathcal{L}|\ge T+1|W=1]
  \nonumber\\
  &\le \mathbb{P}\left[\left.|\mathcal{L}|\ge T+1,\,\mathcal{L}\cap\left\{2,\dots,\frac{\beta M}{T}+1\right\}
  =\emptyset\right|W=1\right]
  \nonumber
  \\
  &\quad+\mathbb{P}\left[\left.
  \mathcal{L}\cap\left\{2,\dots,\frac{\beta M}{T}+1\right\}\neq\emptyset\right|W=1\right]
  \\
  &\le \left(1-\frac{\beta M/T}{M}\right)^T
  \nonumber\\
  &\quad+\mathbb{P}
  \left[\left.\mathcal{L}\cap\left\{2,\dots,\frac{\beta M}{T}+1\right\}\neq\emptyset\right|W=1\right]
  \label{e7_16}
  \\
  &\le \frac{1}{1+\beta}+e^{-(\beta+1)}
  \nonumber\\
  &~+
  \mathbb{P}\left[
  \max_{2\le m\le \frac{\beta M}{T}+1,1\le l\le L}\imath_{Z;X}(\hat{Z};X_{ml})>\log\frac{LM}{T}+\delta
  \right]\label{e7_17}
  \\
  &\le \frac{1}{1+\beta}+e^{-(\beta+1)}+\beta\exp(-\delta)\label{e7_18}
  \end{align}
  where
  \begin{itemize}
  \item To see  \eqref{e7_16}, note that by the symmetry among the messages $2,\dots,M$,
      for any $t\ge T$,
  \begin{align}
  &\mathbb{P}\left[\left.\mathcal{L}
  \cap\left\{2,\dots,\frac{\beta M}{T}+1\right\}=\emptyset\right|W=1,|\mathcal{L}\setminus\{1\}|
  =t\right]
  \\
  &=\left(1-\frac{\beta M/T}{M-1}\right)^t
  \\
  &\le \left(1-\frac{\beta M/T}{M}\right)^T.
  \end{align}
  \item In \eqref{e7_17} $(X^{ML},\hat{Z})\sim Q_X^{\otimes ML}\times Q_Z$, and we used the inequality \eqref{e_15}.
  \item \eqref{e7_18} used the packing lemma \cite{verdu2012non}.
  \end{itemize}
In summary,
  \begin{align}
  \mathbb{P}[\textrm{error}|\textrm{no message}]&\le\frac{1}{A}\exp(-\delta),
  \end{align}
  and for each $m=1,\dots,M$, by the union bound and by the symmetry in codebook generation we have
  \begin{align}
  \mathbb{P}[\textrm{error}|W=m]
  &\le\mathbb{P}[\bar{\imath}_{Z;X}(Z;X)\le\log(LMA)+\delta]
  \nonumber\\
  &\quad+\mathbb{P}\left[\imath_{Z;X}(Z;X)\le\log\frac{LM}{T}+\delta\right]
  \nonumber\\
  &\quad+\frac{1}{1+\beta}+e^{-(\beta+1)}+\beta\exp(-\delta).
  \end{align}
  Decoder~2:
  \begin{align}
  \mathbb{P}[\mathcal{L}\neq\emptyset|\textrm{no message}]
  &\le \mathbb{P}[\imath_{Q_Z\|{{\pi}}_Z}(\bar{Z})>\log A]
  \\
  &\le \frac{1}{A}\mathbb{P}[\imath_{Q_Z\|{{\pi}}_Z}(Z)>\log A]
  \label{echm}
  \\
  &\le \frac{1}{A}
  \end{align}
  where $\bar{Z}\sim {{\pi}}_Z$ and $Z\sim Q_Z$, and \eqref{echm} used the change of measure. On the other hand,
  \begin{align}
  &\mathbb{P}[\mathcal{L}\neq\emptyset,\,1\notin\mathcal{L}|W=1]
  \nonumber\\
  &\le \mathbb{P}\left[\imath_{Z;X}(Z;X)\le\log\frac{LM}{T}
  +\delta\right]
  \end{align}
  and
  \begin{align}\label{e6_66}
  \mathbb{P}[\mathcal{L}=\emptyset|W=1]
  \le  \mathbb{P}[\imath_{Q_Z\|{{\pi}}_Z}(Z)\le \log A],
  \end{align}
  Moreover, as in \eqref{e7_18}, we have
  \begin{align}
  &\mathbb{P}[|\mathcal{L}|\ge T+1|W=1]
  \nonumber\\
  &\le \frac{1}{1+\beta}+e^{-(\beta+1)}+\beta\exp(-\delta)\label{e7_18_n}
  \end{align}
  By union bound,
  \begin{align}
  \mathbb{P}[\textrm{error}|\textrm{no message}]&\le\frac{1}{A};
  \end{align}
  and for each $m=1,\dots,M$,
  \begin{align}
  \mathbb{P}[\textrm{error}|
  W=m]
  &\le\mathbb{P}\left[\imath_{Z;X}(Z;X)\le\log\frac{LM}{T}
  +\delta\right]
  \nonumber\\
  &\quad+\mathbb{P}[\imath_{Q_Z\|{{\pi}}_Z}(Z)\le \log A]
  \nonumber
  \\
  &\quad+\frac{1}{1+\beta}+e^{-(\beta+1)}+\beta\exp(-\delta).
  \end{align}
  \item Decoder~3:\\
  The analysis is similar to that of Decoder~2 and the result follows from \eqref{echm} and \eqref{e6_66}.
\end{itemize}

\bibliographystyle{ieeetr}
\bibliography{SM}
\end{document}